\documentclass[12pt]{article}

\usepackage{myarticle-macrosShared}
\usepackage{enumerate,bm}
\usepackage{color}

\usepackage{setspace}
\usepackage[numbers]{natbib}
\usepackage{url}
\usepackage{graphicx,psfrag,epsf}

\usepackage{sidecap}
\usepackage{subfig}
\usepackage{times}
\usepackage{multirow}
\usepackage{amsmath, amssymb, amsfonts, amsthm}
\usepackage{latexsym}
\usepackage{mathtools}
\usepackage{bbm}
\usepackage{booktabs} 
\usepackage{pbox}
\usepackage[normalem]{ulem}
\usepackage[inline]{enumitem}

\newcommand{\lowerBoundWs}{\eta}
\newcommand{\bHat}{\widehat{b}}
\newcommand{\gHat}{\widehat{g}}
\newcommand{\resid}{e}
\newcommand{\prederror}{\tilde{\resid}_{i(i)}}
\newcommand{\prederrorAtj}{\tilde{\resid}_{j(i)}}
\newcommand{\stresid}{r}
\newcommand{\stprederror}{\tilde{\stresid}_{i(i)}}
\newcommand{\stprederrorAtj}{\tilde{\stresid}_{j(i)}}
\newcommand{\rTilde}{\widetilde{r}}
\newcommand{\varJack}{\textrm{varJACK}}

\newcommand{\epsdist}{G}
\newcommand{\SM}{Supplementary Text, }
\newcommand{\multGauss}{\Gamma}
\newcommand{\gammaHat}{\hat{\gamma}}
\date{First submitted: February 4, 2015\\ This version: October 14, 2015}
\title{Can we trust the bootstrap in high-dimension?}
\begin{document}

\def\spacingset#1{\renewcommand{\baselinestretch}
{#1}\small\normalsize} \spacingset{1}
\title{\bf Can we trust the bootstrap in high-dimension?}
  \author{Noureddine El Karoui\thanks{
    The authors gratefully acknowledge grants NSF DMS-1026441, NSF DMS-0847647 (CAREER) and NSF DMS-1510172. They would also like to thank Peter Bickel and Jorge Banuelos for discussions.}\hspace{.2cm}\\
    and \\
    Elizabeth Purdom \\
    Department of Statistics, University of California, Berkeley}
\maketitle

\begin{abstract}
We consider the performance of the bootstrap in high-dimensions for the setting of linear regression, where $p<n$ but $p/n$ is not close to zero. We consider ordinary least-squares as well as robust regression methods and adopt a minimalist performance requirement: can the bootstrap give us good confidence intervals for a single coordinate of $\beta$? (where $\beta$ is the true regression vector).

We show through a mix of numerical and theoretical work that the bootstrap is fraught with problems. Both of the most commonly used methods of bootstrapping for regression -- residual bootstrap and pairs bootstrap -- give very poor inference on $\beta$ as the ratio $p/n$ grows. We find that the residuals bootstrap tend to give anti-conservative estimates (inflated Type I error), while the pairs bootstrap gives very conservative estimates (severe loss of power) as the ratio $p/n$ grows. We also show that the jackknife resampling technique for estimating the variance of $\hat{\beta}$ severely overestimates the variance in high dimensions.  

We contribute alternative bootstrap procedures based on our theoretical results that mitigate these problems. However, the corrections depend on assumptions regarding the underlying data-generation model, suggesting that in high-dimensions it may be difficult to have universal, robust bootstrapping techniques.

\end{abstract}

\vfill

\newpage
\spacingset{1.45} 

\section{Introduction}

The bootstrap \citep{EfronBootstrap1979AoS} is a ubiquitous tool in applied statistics, allowing for inference when very little is known about the properties of the data-generating distribution. The bootstrap is a powerful tool in applied settings because it does not make the strong assumptions common to classical statistical theory regarding this data-generating distribution. Instead, the bootstrap resamples the observed data to create an estimate, $\hat{F}$, of the unknown data-generating distribution, $F$. $\hat{F}$ then forms the basis of further inference.

Since its introduction, a large amount of research has explored the theoretical properties of the bootstrap, improvements for estimating $F$ under different scenarios, and how to most effectively estimate different  quantities from $\hat{F}$ (see the pioneering \cite{BickelFreedmanTheoryBootAoS81} for instance and many many more references in the book-length review of \cite{DavisonHinkley97}, as well as \cite{vandervaart} for a short  summary of the modern point of view on these questions). Other resampling techniques exist of course, such as subsampling, m-out-of-n bootstrap, and jackknifing, and have been studied and much discussed (see \cite{EfronBook82}, \cite{HallBootstrapAndEdgeworthExpansion92}, \cite{PolitisRomanoWolfSubsampling99}, \cite{BickelGoetzevanZwetMoutOfN1997}, and \cite{EfronTibshirani} for a practical introduction). 

An important limitation for the bootstrap is the quality of $\hat{F}$. The standard bootstrap estimate of $F$ based on the empirical distribution of the data may be a poor estimate when the data has a non-trivial dependency structure,  when the quantity being estimated, such as a quantile, is sensitive to the discreteness of $\hat{F}$, or when the functionals of interest are not smooth (see e.g \cite{BickelFreedmanTheoryBootAoS81} for a classic reference, as well as \cite{BeranSrivastava85} or \cite{EatonTyler91} in the context of multivariate statistics). 

An area that has received less attention is the performance of the bootstrap in high dimensions and this is the focus of our work -- in particular in the setting of standard linear models where data $y_i$ are drawn from the linear model
$$
 \forall i, y_i=\beta'X_i + \epsilon_i\;, 1\leq i \leq n\;, \; \text{ where } X_i \in \mathbb{R}^p\;.
$$
We are interested in the bootstrap or resampling properties of the estimator defined as 
$$
\betaHat_\rho=\argmin_{b \in \mathbb{R}^p} \sum_{i=1}^n \rho(y_i-X_i\trsp b)\;, \text{ where } \rho \text{ is a convex function.}
$$

We consider the two standard methods for resampling to create a bootstrap distribution in this setting. The first is \emph{pairs resampling}, where bootstrap samples are drawn from the empirical distribution of the pairs $(y_i,X_i)$. The second resampling method is \emph{residual resampling}, where the bootstrapped data consists of $y_i^*=\betaHat'X_i + \hat{\epsilon}^*_i$, where $\hat{\epsilon}^*_i$ is drawn from the empirical distribution of the estimated residuals, $\resid_i$. We also consider the jackknife,  a resampling method focused specifically on estimating the variance of functionals of $\betaHat$. These three  methods are extremely flexible for linear models regardless of the method of fitting $\beta$ or the error distribution of the $\epsilon_i.$

\paragraph{The high dimensional setting: \bm{$p/n\tendsto \kappa \in (0,1)$}} In this work we call a high-dimensional setting one  where the number of predictors, $p$, is of the same order of magnitude as the number of observations, $n$, formalized mathematically by assuming that $p/n\tendsto \kappa \in (0,1)$. 
Several reasons motivate our theoretical study in this regime. The asymptotic behavior of the estimate $\betaHat_\rho$ is known to depend heavily on whether one makes the classical theoretical assumption that $p/n\tendsto 0$ or instead assumes $p/n\tendsto \kappa \in (0,1)$ (see Section \ref{sec:background} and \SM \ref{supp:Reminders} and references therein). But from the standpoint of practical usage on moderate-sized datasets  (i.e $n$ and $p$ both moderately sized with $p<n$), it is not always obvious which assumption is justified.  We think that working in the high-dimensional regime of $p/n\tendsto \kappa \in (0,1)$ captures better the complexity encountered even in reasonably low-dimensional practice than using the classical assumption $p/n\tendsto 0$. In fact, asymptotic predictions based on the high-dimensional assumption can work surprisingly well in very low-dimension (see \cite{imj}). Furthermore, in these high-dimensional settings -- where much is still unknown theoretically -- the bootstrap is a natural and compelling alternative to asymptotic analysis. 

\paragraph{Defining success: accurate inference on $\beta_1$} \label{review:WhyBeta1} The common theoretical definition of whether the bootstrap ``works" is that the bootstrap distribution of the entire bootstrap estimate $\betaHat^*$ converges conditionally almost surely to the sampling distribution of the estimator $\betaHat$ (see \cite{vandervaart} for instance). The work of \cite{BickelFreedmanHighDBoot83} on the residual bootstrap for least squares regression, which we discuss in the background section \ref{sec:background}, shows that this theoretical requirement is not fulfilled even for the simple problem of least squares regression. 

In this paper, we choose to focus only on accurate inference for the projection of our parameter on a pre-specified direction $\upsilon$. More specifically, we concentrate only on whether the bootstrap  gives accurate confidence intervals for $\upsilon'\beta$. We think that this is the absolute minimal requirement we can ask of a bootstrap inferential method, as well as one that is meaningful from an applied statistics standpoint. This is of course a much less stringent requirement than doing well on complicated functionals of the whole parameter vector, which is the implicit demand of standard definitions of bootstrap success. For this reason, we focus throughout the exposition on inference for $\beta_1$ (the first element of $\beta$) as an example of a pre-defined direction of interest (where $\beta_1$ corresponds to choosing $\upsilon=e_1$, the first canonical basis vector). 

We note that considering the asymptotic behavior of $\upsilon'\beta$ as $p/n\tendsto \kappa\in (0,1)$ implies that $\upsilon=\upsilon(p)$ changes with $p$. By ``pre-defined'' we will mean simply a deterministic sequence of directions $\upsilon(p)$. We will continue to suppress the dependence on $p$ in writing $\upsilon$ in what follows for the sake of clarity. \label{review:predefinedContrast}

\subsection{Organization and main results of the paper}

In Section \ref{sec:ResidBoot} we demonstrate that in high dimensions residual-bootstrap resampling results in extremely poor inference on the coordinates of $\beta_\rho$ with error rates much higher than the reported Type I error. We show that the error in inference based on residual bootstrap resampling is due to the fact that the distribution of the residuals $\resid_i$ are a poor estimate of the distribution of $\epsilon_i$; we further illustrate that common methods of standardizing the $\resid_i$ do not resolve the problem for general $\rho$. We propose two new methods of residual resampling, including one based on scaled leave-one-out predicted errors that seems to perform better than the other one in our simulations. We also provide some theoretical results for the behavior of this method as $p/n \tendsto 1$. 

In Section \ref{sec:PairsBoot} we examine \textit{pairs-bootstrap resampling} and show that confidence intervals based on bootstrapping the pairs also perform very poorly. Unlike in the residual-bootstrap case discussed in Section \ref{sec:ResidBoot}, the confidence intervals obtained from the pairs-bootstrap are instead conservative to the point of being non-informative. This results in a dramatic loss of power. We prove in the case of $L_2$ loss, i.e $\rho(x)=x^2$, that the variance of the bootstrapped $v\trsp \betaHat^*$ is greater than that of $v\trsp \betaHat$, leading to the overly conservative performance we see in simulations. We demonstrate that a different resampling scheme we propose can alleviate the problems to a certain extent, but we also highlight the practical limitations in such a strategy, since it relies heavily on having strong knowledge about the data-generating model. 

In Section \ref{sec:jackknife}, we discuss another resampling scheme, the jackknife. We focus on the jackknife estimate of variance and show that it has similarly poor behavior in high dimensions. In the case of $L_2$ loss with Gaussian design matrices, we further prove  that the jackknife estimator over estimates the variance of our estimator by a factor of $1/(1-p/n)$; we briefly mention corrections for other losses.

We rely on simulation results to demonstrate the practical impact of the failure of the bootstrap. The settings for our simulations and corresponding theoretical analyses are idealized, without many of the common problems of heteroskedasticity, dependency, outliers and so forth that are known to be a problem for robust bootstrapping. This is intentional, since even these idealized settings are sufficient to demonstrate that the standard bootstrap methods have poor performance. For brevity, we give only brief descriptions of the simulations in what follows; detailed descriptions can be found in Supplementary Text, Section \ref{supp:Simulations}.

Similarly, we focus on the basic implementations of the bootstrap for linear models. While there are many proposed alternatives -- often for specific loss functions or types of data -- the standard methods we study are most commonly used  and recommended  in practice. Furthermore, to our knowledge none of the alternative bootstrap methods we have seen specifically address the underlying theoretical problems that appear in high dimensions and therefore are likely to suffer from the same fate as standard methods. We have also tried more complicated ways to build confidence intervals (e.g. bias correction methods), but have found their performance  to be erratic in high-dimension.

We first give some background regarding the bootstrap and estimation of linear models in high dimensions before presenting our new results. 

\subsection{Background: Inference using the Bootstrap}\label{sec:background}
We consider the setting 
$
y_i=\beta'X_i + \epsilon_i,
$
where $E(\epsilon_i)=0$ and $\var{\epsilon_i}=\sigma^2_\eps$. $\beta$ is estimated as minimizing the average loss,
\begin{equation}\label{eq:loss}
\betaHat_\rho=\argmin_{b \in \mathbb{R}^p} \sum_{i=1}^n \rho(y_i-X_i'b), 
	\end{equation} 
where $\rho$ defines the loss function for a single observation. $\rho$ is assumed to be convex in all the paper. Common choices are $\rho(x)=x^2$, i.e least-squares, $\rho(x)=|x|$, which defines $L_1$ regression, or $\text{Huber}_k$ loss where $\rho(x)=(x^2/2) 1_{|x|<k}+(k|x|-k^2/2) 1_{|x|\geq k}$. 

 Bootstrap methods are used in order to estimate the distribution of the estimate $\betaHat_\rho$ under the true data-generating distribution, $F$. The bootstrap estimates this distribution with the distribution obtained when  the data is drawn from an estimate $\hat{F}$ of $F$. Following standard convention, we designate this bootstrapped estimator  $\betaHat^*_\rho$ to note that this is an estimate of $\beta$ using loss function $\rho$ when the data-generating distribution is known to be exactly equal to $\hat{F}$. Since $\hat{F}$ is completely specified, we can in principle exactly calculate the distribution of $\betaHat^*_\rho$ and use it as an approximation of the distribution of $\betaHat_\rho$ under $F$. In practice, we simulate $B$ independent draws of size $n$ from the distribution $\hat{F}$ and perform inference based on the empirical distribution of  $\betaHat^{*b}_\rho$, $b=1,\ldots, B$.

In bootstrap inference for the linear model, there are two common methods for resampling, which results in different estimates $\hat{F}$. In the first method, called the residual bootstrap,  $\hat{F}$ is an estimate of the conditional distribution of $y_i$ given $\beta$ and $X_i$. In this case, the corresponding resampling method consists of resampling $\epsilon_i^*$ from an estimate of the distribution of $\epsilon$ and forming data $y_i^*=X_i'\betaHat+\epsilon_i^*$, from which $\betaHat_\rho^*$ is computed. This method of bootstrapping assumes that the linear model is correct for the mean of $y$ (i.e. that $\Exp{y_i}=X_i\trsp \beta$); it is also assuming fixed $X_i$ design vectors because the sampling is conditional on the $X_i$. In the second method, called pairs bootstrap, $\hat{F}$ is an estimate of the joint distribution of the vector $(y_i,X_i)\in R^{p+1}$ given by the empirical joint distribution of $\{(y_i,X_i)\}_{i=1}^n$; the corresponding resampling method resamples the pairs $(y_i,X_i)$. This method makes no assumption about the mean structure of $y$ and, by resampling the $X_i$, also does not condition on the values of $X_i$. For this reason, pairs resampling is often considered to be more generally applicable than  residuals resampling - see e.g \cite{DavisonHinkley97}.

\subsection{Background: High-dimensional inference of linear models}
Recent research shows that $\betaHat_\rho$ has very different asymptotic properties when  $p/n$ has a limit $\kappa$ that is  bounded away from zero than it does in the classical setting where $p/n\tendsto 0$ (see e.g \cite{HuberRobustRegressionAsymptoticsETCAoS73,HuberRonchettiRobustStatistics09,PortnoyMestLargishPNConsistencyAoS84,PortnoyMestLargishPNCLTAoS85,portnoy1986,PortnoyCLTRobustRegressionJMVA87,MammenRobustRegressionAos89} for $\kappa=0$; \cite{NEKRobustPaperPNAS2013Published} for $\kappa\in (0,1)$). A simple example is that the vector $\betaHat_\rho$ is no longer consistent in Euclidean norm when $\kappa>0$. 
We should be clear, however, that projections on fixed non-random directions such as we consider, i.e $\upsilon'\betaHat_\rho$, are $\sqrt{n}$ consistent for $\upsilon\trsp \beta$, even when $\kappa>0$. In particular, the coordinates of $\betaHat_\rho$ are $\sqrt{n}-$consistent for the coordinates of $\beta$. Hence, in practice the estimator $\betaHat_\rho$ is still a reasonable quantity to consider (see \SM \ref{supp:Reminders} for much more detail).

\paragraph{Bootstrap in high-dimensional linear models} Very interesting work exists already in the literature about bootstrapping regression estimators when $p$ is allowed to grow with $n$ (\cite{ShorackBootstrappingRobustRegression82,WuAos86Resampling,MammenRobustRegressionAos89,MammenBootWildBootAndAsympNormalityPTRF92,MammenHighDBootAoS93,ParzenEtAlBoot94},  Section 3.9 of \cite{KoenkerQuantileRegressionBook05}).  With a few exceptions, this work has been in the classical, low-dimensional setting where either $p$ is held fixed or $p$ grows slowly relative to $n$ (i.e $\kappa=0$ in our notation). For instance, in \cite{MammenHighDBootAoS93}, it is shown that under mild technical conditions and assuming that $p^{1+\delta}/n\tendsto 0$, $\delta>0$, the pairs bootstrap distribution of linear contrasts $v\trsp (\betaHat^*-\betaHat)$ is in fact very close to the sampling distribution of $v\trsp(\betaHat-\beta)$ with high-probability, when using least-squares.  Other results such as \cite{ShorackBootstrappingRobustRegression82} and \cite{MammenRobustRegressionAos89}, also allow for increasing dimensions, for example in the case of linear contrasts in robust regression, by making assumptions on the diagonal entries of the hat matrix. In our context, these assumptions would be satisfied only if $p/n\tendsto 0$. Hence those interesting results do not apply to the present study. We also note that \cite{HallBootstrapAndEdgeworthExpansion92} contains on p. 167 cautionary notes about using the bootstrap in high-dimension.

While there has not been much theoretical work on the bootstrap in the setting where $p/n\tendsto \kappa\in (0,1)$, one early work of \cite{BickelFreedmanHighDBoot83} considered bootstrapping scaled residuals for least-squares regression when $\kappa>0$. They show (Theorem 3.1 p.39 in \mbox{\cite{BickelFreedmanHighDBoot83}}) that when $p/n\tendsto \kappa \in (0,1)$, there exists a  data-dependent direction $c$, such that $c\trsp \betaHat^*$ does not have the correct asymptotic distribution, i.e its distribution is not conditionally in probability close to the sampling distribution of $ c\trsp \betaHat$. Furthermore, they show that when the errors in the model are Gaussian, under the assumption that the diagonal entries of the hat matrix are not all close to a constant, the empirical distribution of the residuals is a scaled-mixture of Gaussian, which is not close to the original error distribution.  

As we previously explained, in this work we instead only consider inference for \emph{predefined} contrasts $\upsilon'\beta$. The important and interesting problems pointed out in \cite{BickelFreedmanHighDBoot83} disappear if we focus on fixed, non-data-dependent projection directions. Hence, our work complements the work of \cite{BickelFreedmanHighDBoot83} and is not redundant with it.

\paragraph{The role of the distribution of $X$} An important consideration in interpreting theoretical work on linear models in high dimensions is the role of the design matrix $X$. In classical asymptotic theory, the analysis is conditional on $X$ so that the assumptions in most theoretical results are stated in terms of conditions that can be evaluated on a specific design  matrix $X$. In the high dimensional setting, the available theoretical tools do not yet allow for an asymptotic analysis  conditional on $X$; instead the results make assumptions about the distribution of $X$. Theoretical work in the nascent literature for the high dimensional setting usually allows for a fairly general class of distributions for the individual elements of $X_i$ and can handle covariance between the predictor variables. However, the $X_i$'s are generally considered i.i.d., which limits the ability of any $X_i$ to be too influential in the fit of the model (see \SM \ref{supp:Reminders} for more detail). For discussion of limitations of the corresponding models for statistical purposes, see \cite{DiaconisFreedmanProjPursuit84,HallMarronNeemanJRSSb05,nekCorrEllipD}.

\subsection{Notations and default conventions} When referring to the Huber loss in a numerical context, we refer (unless otherwise noted) to the default implementation in the \texttt{rlm} package in R, where the transition from quadratic to linear behavior is at $k=1.345$. We call $X$ the design matrix and $\{X_i\}_{i=1}^n$ its rows. We have $X_i \in \mathbb{R}^p$. $\beta$ denotes the true regression vector, i.e the population parameter. $\betaHat_\rho$ refers to the estimate of $\beta$ using loss $\rho$; from this point on, however, we will often drop the $\rho$ and refer to simply $\betaHat$. $\resid_i$ denotes the $i$-th residual, i.e $\resid_i=y_i-X_i\trsp \betaHat$.
Throughout the paper, we assume that the linear model holds, i.e $y_i=X_i\trsp \beta+\eps_i$ for some fixed $\beta\in\mathbb{R}^p$ and that $\eps_i$'s are i.i.d with mean 0 and $\var{\eps_i}=\sigma^2_\eps$. We call $\epsdist$ the distribution of $\eps$. When we need to stress the impact of the error distribution on the distribution of $\betaHat_\rho$, we will write $\betaHat_\rho(\epsdist)$ or $\betaHat_\rho(\eps)$ to denote our estimate of $\beta$ obtained assuming that $\eps_i$'s are i.i.d $\epsdist$. 

We denote generically by $\kappa=\lim_{n\tendsto \infty} p/n$. We restrict ourselves to $\kappa\in (0,1)$.  The standard notation $\betaHat_{(i)}$ refers to the leave-one-out estimate of $\betaHat$ where the $i$-th pair $(y_i,X_i)$ is excluded from the regression.   $\prederror\triangleq y_i-X_i\trsp \betaHat_{(i)}$ is the $i$-th predicted error (based on the leave-one-out estimate of $\betaHat$). We also use the notation $\prederrorAtj\triangleq y_j-X_j\trsp \betaHat_{(i)}$. The hat matrix is of course $H=X(X\trsp X)^{-1}X\trsp$. $\lo_P$ denotes a ``little-oh'' in probability, a standard notation (see \cite{vandervaart}). When we say that we work with a Gaussian design with covariance $\Sigma$, we mean that $X_i\iid {\cal N}(0,\Sigma)$. Throughout the paper, the loss function $\rho$ is assumed to be convex, $\mathbb{R}\mapsto \mathbb{R}^+$. We use the standard notation $\psi=\rho'$. We finally assume that $\rho$ is such that there is a unique solution to the robust regression problem - an assumption that applies to all classical losses in the context of our paper.

\section{Residual Bootstrap}\label{sec:ResidBoot}

We first focus on the method of bootstrap resampling where $\hat{F}$ is the conditional distribution $y | \betaHat, X.$ In this case the distribution of $\betaHat^{*}$ under $\hat{F}$ is formed by independent resampling of $\epsilon_i^*$ from an estimate $\hat{\epsdist}$ of the distribution $\epsdist$ that generated $\epsilon_i$. Then new data $y_i^*$ are formed as $y_i^*=X_i'\betaHat+\epsilon_i^*$ and the model is fitted to this new data to get $\betaHat^{*}$. Generally the estimate of the error distribution, $\hat{\epsdist}$, is taken to be empirical distribution of the observed residuals, so that the $\epsilon_i^*$ are found by sampling with replacement from the $\resid_i$.

Yet, even a cursory evaluation of $\resid_i$ in the simple case of least-squares regression ($\rho(x)=x^2$) reveals that the empirical distribution of the $\resid_i$ may be a poor approximation to the error distribution of $\epsilon_i$; in particular, it is well known that $\resid_i$ has variance equal to $\sigma^2_\eps(1-h_i)$ where $h_i$ is the $i$th diagonal element of the hat matrix. This problem becomes particularly pronounced in high dimensions. For instance, if $X_i\iid{\cal N}(0,\Sigma)$, $h_i=p/n+\lo_P(1)$ so that $\resid_i$ has variance approximately $\sigma^2_\eps (1-p/n)$, i.e. generally much smaller than the true variance of $\epsilon$ for $\lim p/n>0$. This fact is also true in much greater generality for the distribution of the design matrix $X$ (see e.g \cite{wachter78}, \cite{Haff79IdentityWishartDWithApps}, \cite{silverstein95}, \cite{PajorPasturPub09}, \cite{NEKHolgerShrinkage11}, where the main results of some of these papers require minor adjustments to get the approximation of $h_i$ we just mentioned). 

In Figure \ref{fig:basicCIError}, we plot the error rate of 95\% bootstrap confidence intervals based on resampling from the residuals for different loss functions, based on a simulation when the entries of $X$ are i.i.d ${\cal N}(0,1)$ and $\epsilon\sim N(0,1)$. Even in this idealized situation, as the ratio of $p/n$ increases the error rate of the confidence intervals in least squares regression increases well beyond the expected 5\%: we observe error rates of 10-15\% for $p/n=0.3$ and approximately $20\%$ for $p/n=0.5$ (Table \ref{tab:basicCIError}). We see similar error rates for other robust-regression methods, such as $L_1$ and Huber loss, and also for different error distributions and distributions of $X$ (Supplementary Figures \ref{fig:basicCIErrorLap} and \ref{fig:basicCIErrorDesign}). 
We explain some of the reasons for these problems in Subsection \ref{subsec:ResidTheory} below.

\begin{figure}[t]
	\centering
	\subfloat[][$L_1$ loss]{\includegraphics[width=.3\textwidth]{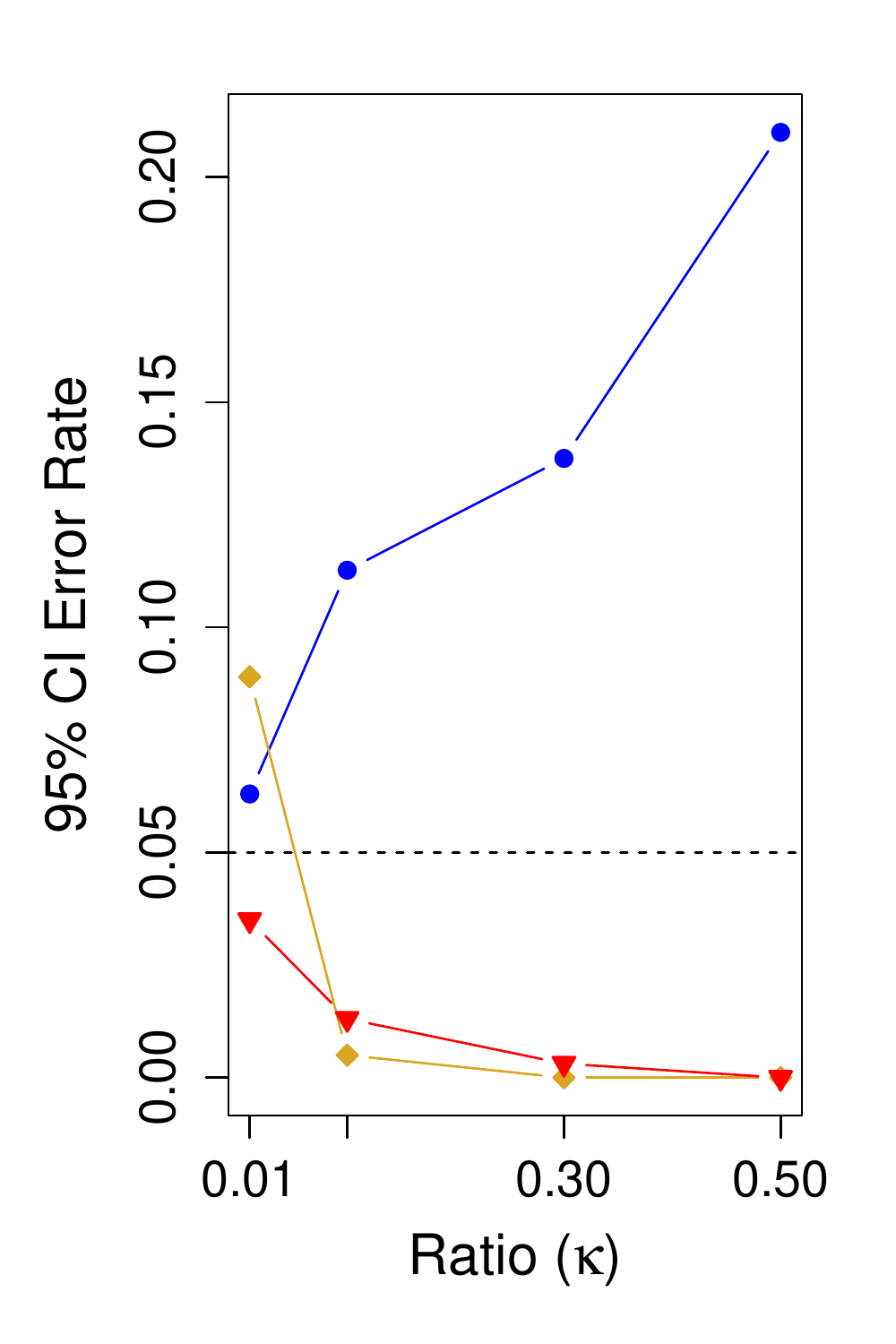} \label{subfig:basicCIError:L1} }
	\subfloat[][Huber loss]{\includegraphics[width=.3\textwidth]{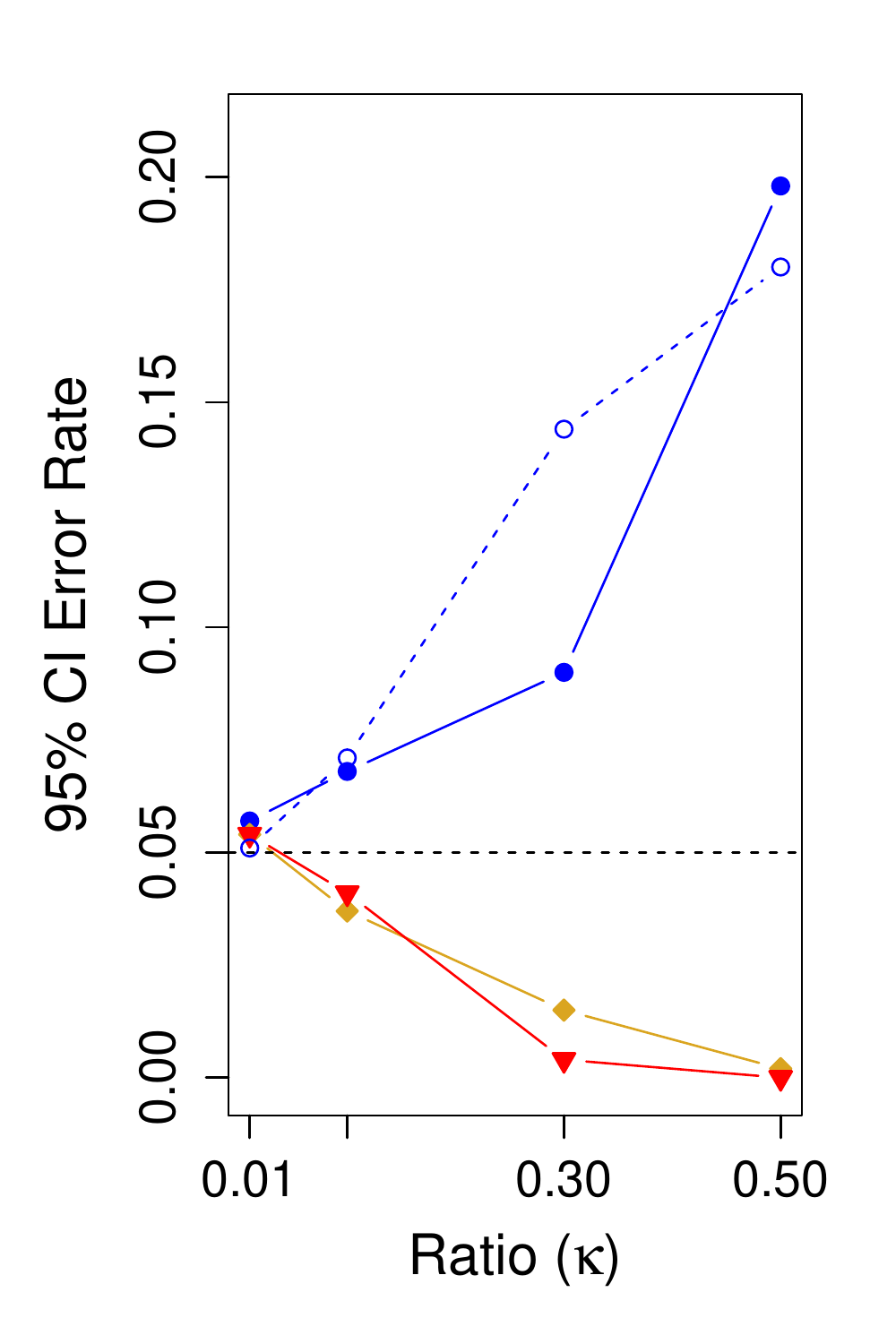} \label{subfig:basicCIError:Huber} }
	\subfloat[][$L_2$ loss]{\includegraphics[width=.3\textwidth]{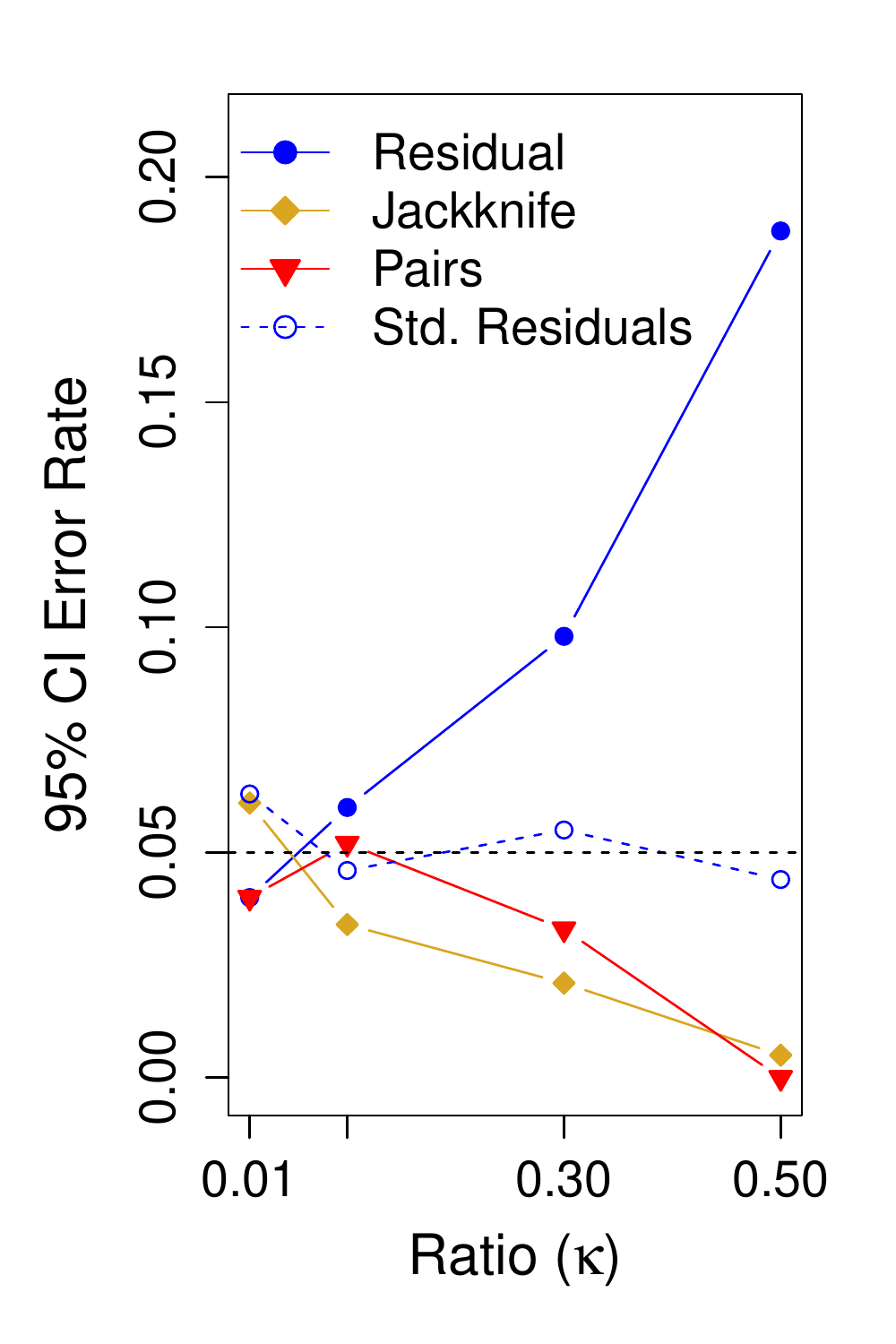} \label{subfig:basicCIError:L2} }
\caption{
 \textbf{Performance of 95\% confidence intervals of $\beta_1$ : }  Here we show the coverage error rates for 95\% confidence intervals for $n=500$ based on applying common resampling-based methods to simulated data: pairs bootstrap (red), residual bootstrap (blue), and jackknife estimates of variance (yellow). These bootstrap methods are applied with three different loss functions shown in the three plots above: \protect\subref{subfig:basicCIError:L1} $L_1$, 
\protect\subref{subfig:basicCIError:Huber} Huber, and \protect\subref{subfig:basicCIError:L2} $L_2$. 
For $L_2$ and Huber loss, we also show the performance of methods for standardizing the residuals before bootstrapping described in the text (blue, dashed line). If accurate, all of these methods should have an error rate of 0.05 (shown as a horizontal black line). The error rates are based on 1,000 simulations, see the description in Supplementary Text, Section \ref{supp:Simulations} for more details; exact values are given in Table \ref{tab:basicCIError}. Error rates above 5\% correspond to anti-conservative methods. Error rates below 5\% correspond to conservative methods.
}\label{fig:basicCIError}
\end{figure}
\subsection{Bootstrapping from Corrected Residuals}\label{subsec:bootCorrectedResiduals}
While resampling directly from the uncorrected residuals is widespread and often given as a standard bootstrap procedure (e.g. \cite{KoenkerQuantileRegressionBook05, Chernick}), the discrepancy between the distribution of $\epsilon_i$ and $\resid_i$ has spurred more refined recommendations in the case of least-squares: form corrected residuals $\stresid_i=\resid_i/\sqrt{1-h_i}$  and sample the $\epsilon_i^*$ from the empirical distribution of the $\stresid_i-\bar{r}$ (see e.g \cite{DavisonHinkley97}). 

This correction is known to exactly align the variance of $\stresid_i$ with that of $\epsilon_i$ regardless of the design vectors $X_i$ or the true error distribution, using simply the fact that the hat matrix is a rank $min(n,p)$ orthogonal projection matrix. We see that for $L_2$ loss it corrects the error in bootstrap inference in our simulations (Figure \ref{fig:basicCIError}). This is not so surprising, given that with $L_2$ loss, the error distribution $G$ impacts the inference on $\beta$  only through $\sigma^2_\eps$, in the case of homoskedastic errors (see Section \ref{subsec:PerfMethod2} for much more detail).

However, this adjustment of the residuals is a correction specific to the least-squares problem. Similar corrections for robust estimation procedures using a loss function $\rho$ are given by \cite{McKean:1993gz} with standardized residuals $r_i$ given by,
\begin{equation}\label{eq:robustCorrect}
	\stresid_i=\frac{\resid_i}{\sqrt{1-d h_i}},\text{ where } d= \frac{2 \sum e'_j \psi(e'_j)}{\sum  \psi(e'_j)} - \frac{\sum \psi(e'_j)^2}{(\sum  \psi(e'_j))^2},
\end{equation}
where $h_i$ is the $i$-th diagonal entry of the hat matrix, $e'_j=e_j/s$, $s$ is a estimate of $\sigma$, and $\psi$ is the derivative of $\rho$, assuming $\psi$ is a bounded and odd function (see \cite{DavisonHinkley97} for a complete description of its implementation for the bootstrap and  \cite{McKean:1993gz} for a full description of regularity conditions). 

Unlike the correction for $L_2$ loss mentioned earlier, however, the scaling described in Equation \eqref{eq:robustCorrect} for the residuals is an approximate variance correction and the approximation depends on assumptions that do not hold true in higher dimensions. The error rate of confidence intervals in our simulations based on this rescaling show no improvement in high dimensions over that of simple bootstrapping of the residuals. This could be explained by the fact that standard perturbation analytic methods used for the analysis of M-estimators in low-dimension - which are at the heart of the correction in Equation \eqref{eq:robustCorrect} - fail in high-dimension.

\subsection{Understanding the behavior of residual bootstrap}\label{subsec:ResidTheory}

At a high-level, this misbehavior of the residual bootstrap can be explained by the fact that in high-dimension, the residuals tend to have a very different distribution from that of the true errors. This is in general true both in terms of simple properties such as variance and in terms of more general aspects, such as the whole marginal distribution. To make these statements precise, we make use of the previous work of \citep{NEKRobustPaperPNAS2013Published,NEKRobustRegressionRigorous2013}. These papers do not discuss bootstrap or resampling issues, but rather are entirely focused on providing asymptotic theory for the behavior of $\betaHat_\rho$ as $p/n\tendsto \kappa \in (0,1)$; in the course of doing so, they characterize the asymptotic relationship of $e_i$ to $\eps_i$ in high-dimensions. We make use of this relationship to characterize the behavior of the residual bootstrap and to suggest an alternative estimates of $\hat{\epsdist}$ for bootstrap resampling.   

\paragraph{Behavior of residuals in high-dimensional regression} We now summarize the asymptotic relationship between $\resid_i$ and $\eps_i$ in high-dimensions given in the above cited work (see \SM Section \ref{supp:Reminders} for a more detailed and technical summary). Let $\betaHat_{(i)}$ be the estimate of $\beta$ based on fitting  the linear model of Equation \eqref{eq:loss} without using observation $i$, and $\prederrorAtj$ be the error of observation $j$ from this model (the leave-one-out or predicted error), i.e 
$
\prederrorAtj=y_j-X_j\trsp \betaHat_{(i)}
$
For simplicity of exposition, $X_i$ is assumed to have an elliptical distribution, i.e $X_i=\lambda_i \multGauss_i,$ where $\multGauss_i\sim N(0,\Sigma)$, and $\lambda_i$ is a scalar random variable independent of $\multGauss_i$ with $\Exp{\lambda_i^2}=1$. 
For simplicity in restating their results, we will assume $\Sigma=\id_p$, but equivalent statements can be made for arbitrary $\Sigma$; similar results also apply when $\multGauss_i=\Sigma^{1/2}\xi_i$, with $\xi_i$ having i.i.d non-Gaussian entries, satisfying a few technical requirements (see \SM Section \ref{supp:Reminders}).

With this assumption on $X_i$, for any sufficiently smooth loss function $\rho$ and any size dimension where $p/n\rightarrow \kappa<1$, the relationship between the $i$-th residual $\resid_i$ and the true error $\epsilon_i$ can be summarized as,
\begin{gather}
	\prederror = \epsilon_i + |\lambda_i|\norm{\betaHat_{\rho(i)}-\beta}_2 Z_i +\lo_P(u_n) \label{eq:residual}\\
	\resid_i + c_i\lambda_i^2 \psi(\resid_i)= \prederror +\lo_P(u_n) \label{eq:residual2ndPart}
\end{gather}
where $Z_i$ is a random variable distributed $N(0,1)$ and independent of $\epsilon_i$. $u_n$ is a sequence of numbers tending to 0. $c_i$, $\lambda_i$ and $\norm{\betaHat_{\rho(i)}-\beta}_2$ are all of order 1, i.e they are not close to 0 in general in the high-dimensional setting. \label{review:orderMagnitude} 
The scalar $c_i$ is given as $\frac{1}{n}\trace{S_i^{-1}},$ where $S_i=\frac{1}{n}\sum_{j\neq i}\psi'(\prederrorAtj)X_jX_j'$. For $p, n$ large the $c_i$'s are approximately equal and $\norm{\betaHat_{\rho(i)}-\beta}_2\simeq \norm{\betaHat_{\rho}-\beta}_2\simeq\Exp{\norm{\betaHat_{\rho}-\beta}_2}$; furthermore $c_i\lambda_i^2$ can be approximated by $X_i\trsp S_i^{-1}X_i/n$. Note that when $\rho$ is either non-differentiable at all points ($L_1$) or not twice differentiable (Huber), arguments can be made that make these expressions valid, using  for instance the notion of sub-differential for $\psi$ \citep{HiriartLemarechalConvexAnalysisAbridged2001}.

\paragraph{Interpretation of Equations \eqref{eq:residual} and \eqref{eq:residual2ndPart}} 
Equation \eqref{eq:residual} means that the marginal distribution of the leave-$i$-th-out predicted error, $\prederror$, is asymptotically a convolution of the true error, $\epsilon_i$, and an independent scale mixture of Normals. Furthermore, Equation \eqref{eq:residual2ndPart} means that the $i$-th residual $\resid_i$ can be understood as a non-linear transformation of $\prederror$. As we discuss below, these relationships are  qualitatively very different  from the classical case $p/n\tendsto 0$.

\subsubsection{Consequence for the residual bootstrap} We apply these results to the question of the residual bootstrap to give an understanding of why bootstrap resampling of the residuals can perform so badly in high-dimension. The distribution of the $\resid_i$ is far removed from that of the $\epsilon_i$, and hence bootstrapping from the residuals effectively amounts to sampling errors from a distribution that is very different from the original error distribution, $\eps$.

The impact of these discrepancies for bootstrapping is not equivalent for all dimensions, error distributions, or loss functions. It depends on the constant $c_i$ and the risk, $\norm{\betaHat_{\rho (i)}-\beta}_2$, both of which are highly dependent on the dimensions of the problem, the distribution of the errors and the choice of loss function. We now discuss some of these issues.

\paragraph{Least Squares regression} In the case of least squares regression, the relationships given in Equation \eqref{eq:residual} are exact, i.e $u_n=0$. Further, $\psi(x)=x$, and $c_i= h_i/(1-h_i)$, giving the well known linear relationship $\resid_i=(1-h_i)\prederror$ \citep{WeisbergLinearRegressionBook14}. This linear relationship is exact regardless of dimension, though the dimensionality aspects are captured by $h_i$. This expression 
can be used to show that asymptotically $\Exp{\sum_{i=1}^n \resid_i^2}=\sigma^2_\eps (n-p)$, when $\eps_i$'s have the same variance. Hence, sampling at random from the residuals results in a distribution that underestimates the variance of the errors by a factor $1-p/n$. The corresponding bootstrap confidence intervals are then naturally too small, and hence the error rate increases far from the nominal 5\% - as we observed in Figure \ref{subfig:basicCIError:L2}.

\paragraph{More general robust regression} The situation is much more complicated for general robust regression estimators.  One clear implication of Equations \eqref{eq:residual} and \eqref{eq:residual2ndPart} is that simply rescaling the residuals $\resid_i$ should not in general result in an estimated error distribution $\hat{\epsdist}$ that will have similar properties to those of $\epsdist$. 
The relationship between the residuals and the errors is very non-linear in high-dimensions. This is why in what follows we will propose to work with leave-one-out predicted errors $\prederror$ instead of the residuals $\resid_i$.

\paragraph{The classical case of $\bm{p/n\rightarrow 0}$:} In this setting, $c_i\rightarrow 0$ and therefore Equation \eqref{eq:residual} shows that the residuals $\resid_i$ are approximately equal in distribution to the predicted errors, $\prederror$. Similarly, $\betaHat_\rho$ is $L_2$ consistent when $p/n\rightarrow 0$, so $\norm{\betaHat_{\rho(i)}-\beta}_2^2\rightarrow 0$ and Equation \eqref{eq:residual2ndPart} gives $\prederror\simeq \eps_i$. Hence, the residuals should be fairly close to the true errors in the model when $p/n$ is small. This dimensionality assumption is key to many theoretical analyses of robust regression, and underlies the derivation of corrected residuals $\stresid_i$ of \cite{McKean:1993gz} given in Equation \eqref{eq:robustCorrect} above. 

\subsection{Alternative residual bootstrap procedures}\label{subsec:LeaveOut}
We propose two methods for improving the performance of confidence intervals obtained through the residual bootstrap. Both do so by providing alternative estimates of $\hat{\epsdist}$ from which bootstrap errors $\epsilon^*_i$ can be drawn. They estimate a $\hat{\epsdist}$ appropriate for the setting of high-dimensional data by accounting for relationship of the distribution of $\epsilon$ and $\prederror$. 

\paragraph{Method 1: Deconvolution} The relationship in Equation \eqref{eq:residual} says that the distribution of  $\prederror$ is a convolution of the correct $\epsdist$ distribution and a Normal distribution. This suggests applying techniques for deconvolving a signal from gaussian noise. Specifically, we propose the following bootstrap procedure: \begin{enumerate*}[label=\textbf{\arabic*)}]
\item calculate the predicted errors, $\prederror$; 
\item estimate the variance of the normal (i.e. $|\lambda_i|\norm{\betaHat_{\rho(i)}-\beta}_2^2$);
\item deconvolve in $\prederror$ the error term $\epsilon_i$ from the normal term; 
\item Use the resulting estimate $\hat{\epsdist}$ to draw errors $\eps^*_i$  for residual bootstrapping.
\end{enumerate*}

Deconvolution problems are known to be very difficult (see \cite{FanDeconvoluionAoS91}, Theorem 1 p. 1260, that gives $1/\log(n)^{\alpha}$ rates of convergence when convolving with a  Gaussian distribution). The resulting deconvolved errors are likely to be quite noisy estimates of $\epsilon_i$. However, it is possible that while individual estimates are poor, the distribution of the deconvolved errors is estimated well enough to form a reasonable $\hat{\epsdist}$ for the bootstrap procedure. 

We used the deconvolution algorithm in the \texttt{decon} package in R \citep{Wang2011Decon} to estimate the distribution of $\epsilon_i$.  The deconvolution algorithm requires knowledge of the variance of the Gaussian that is convolved with the $\epsilon_i$, i.e. estimation of $|\lambda_i|\norm{\betaHat_{\rho(i)}-\beta}_2$ term. In what follows, we assume a Gaussian design, i.e. $\lambda_i=1$, so that we need to estimate only the term $\norm{\betaHat_{\rho(i)}-\beta}^2_2.$ An estimation strategy for the more general setting of $|\lambda_i|\neq 1$ is presented in \SM Section \ref{supp:sec:estimationEllipParam}. We use the fact that $\norm{\betaHat_{\rho(i)}-\beta}^2_2\simeq \norm{\betaHat_{\rho}-\beta}^2_2$ for all $i$ 
and estimate $\norm{\betaHat_{\rho(i)}-\beta}_2$ as $\widehat{var}(\prederror)-\hat{\sigma}^2_{\eps},$ where $\widehat{var}(\prederror)$ is the empirical variance of the $\prederror$ and $\hat{\sigma}^2_{\eps}$ is an estimate of the variance of $\epsdist,$ which we discuss below. We note that the deconvolution strategy we employ makes assumptions of homoskedastic errors $\eps_i$'s, which is true in our simulations but may not be true in practice.  See \SM Section \ref{supp:deconvolution} for details regarding the implementation of Method 1.

\paragraph{Method 2: Bootstrapping from standardized $\prederror$ } 
A simpler alternative is bootstrapping from the predicted error terms, $\prederror$, without deconvolution. Specifically, we propose to bootstrap from a scaled version of $\prederror$,
\begin{equation}\label{eq:DefScaledConvoErrorDist}
\stprederror=\frac{\hat{\sigma}_{\epsilon}}{\sqrt{\widehat{var}(\prederror)}}\prederror,
\end{equation}
where $\widehat{var}(\prederror)$ is the standard estimate of the variance of $\prederror$ and $\hat{\sigma}_\eps$ is an estimate of $\sigma_\eps$. This scaling aligns the first two moments of $\prederror$ with those of $\epsilon_i$.
On the face of it, resampling from $\stprederror$ seems problematic, since Equation \eqref{eq:residual} demonstrates that $\prederror$ does not have the same distribution as $\eps_i$, even if the first two moments are the same. However, as we demonstrate in simulations, this distributional mismatch appears to have limited practical effect on our bootstrap confidence intervals.

\paragraph{Estimation of $\sigma_\eps^2$}  Both methods described above require an estimator of $\sigma_\eps$ that is consistent regardless of dimension and error distribution. As we have explained earlier, for general $\rho$ we cannot rely on the observed residuals $\resid_i$ nor on  $\prederror$ for estimating $\sigma_\eps$ (see Equations \eqref{eq:residual} and \eqref{eq:residual2ndPart}). The exception is  the standard estimate of $\sigma^2_\eps$ from least-squares regression, i.e $\rho(x)=x^2$,
$$\widehat{\sigma}_{\eps,LS}^2=\frac{1}{n-p}\sum_i \resid_{i,L_2}^2.$$
 $\widehat{\sigma}_{\eps,LS}^2$ is a consistent estimator of $\sigma^2$, assuming i.i.d errors and mild moment requirements. In implementing the two alternative residual-bootstrap methods described above, we use $\widehat{\sigma}_{\eps,LS}$ as our estimate of $\sigma_\eps$.

\paragraph{Performance in bootstrap inference} \label{review:performanceBootstrap}In Figure \ref{fig:bootErrorLeaveOut} we show the error rate of confidence intervals based on the two residual-bootstrap methods we proposed above.  We see that both methods control the Type I error, unlike bootstrapping directly from the residuals, and that both methods are conservative. There is little difference between the two methods with this sample size ($n=500$), though with $n=100$, we observe the deconvolution performance to be worse in $L_1$ (data not shown). 

The deconvolution strategy, however, depends on the distribution of the design matrix, which in these simulations we assumed was Gaussian (so we did not have to estimate $\lambda_i$'s). For elliptical designs ($\lambda_i\neq 1$), the error rate of the deconvolution method described above, with no adaptation for the design, was similar to that of uncorrected residuals in high dimensions (i.e. $>0.25$ for $p/n=0.5$).  Individual estimates of $\lambda_i$ (see \SM Section \ref{supp:sec:estimationEllipParam}) might improve the deconvolution strategy, but this problem points to the general reliance of the deconvolution method on precise knowledge about the design matrix. The bootstrap using standardized predicted errors, on the other hand, had a Type I error for an elliptical design only slightly higher than the target 0.05 (around $0.07$, data not shown), suggesting that it might be less sensitive to the properties of the design matrix.

\begin{figure}[t]
\centering
\centering
\subfloat[][$L_1$ loss]{\includegraphics[width=.3\textwidth]{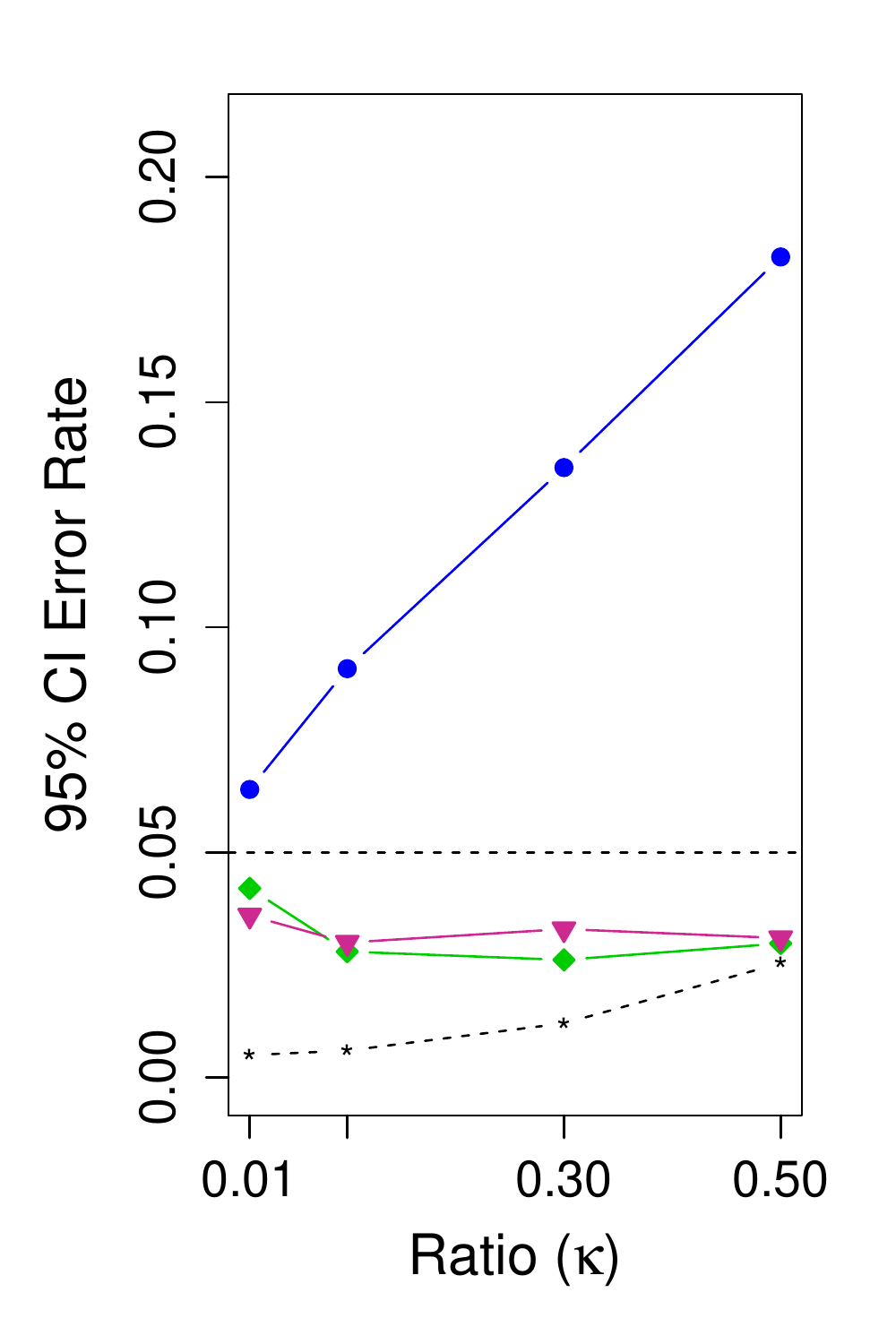}\label{subfig:bootErrorLeaveOut:L1}}
\subfloat[][Huber loss]{\includegraphics[width=.3\textwidth]{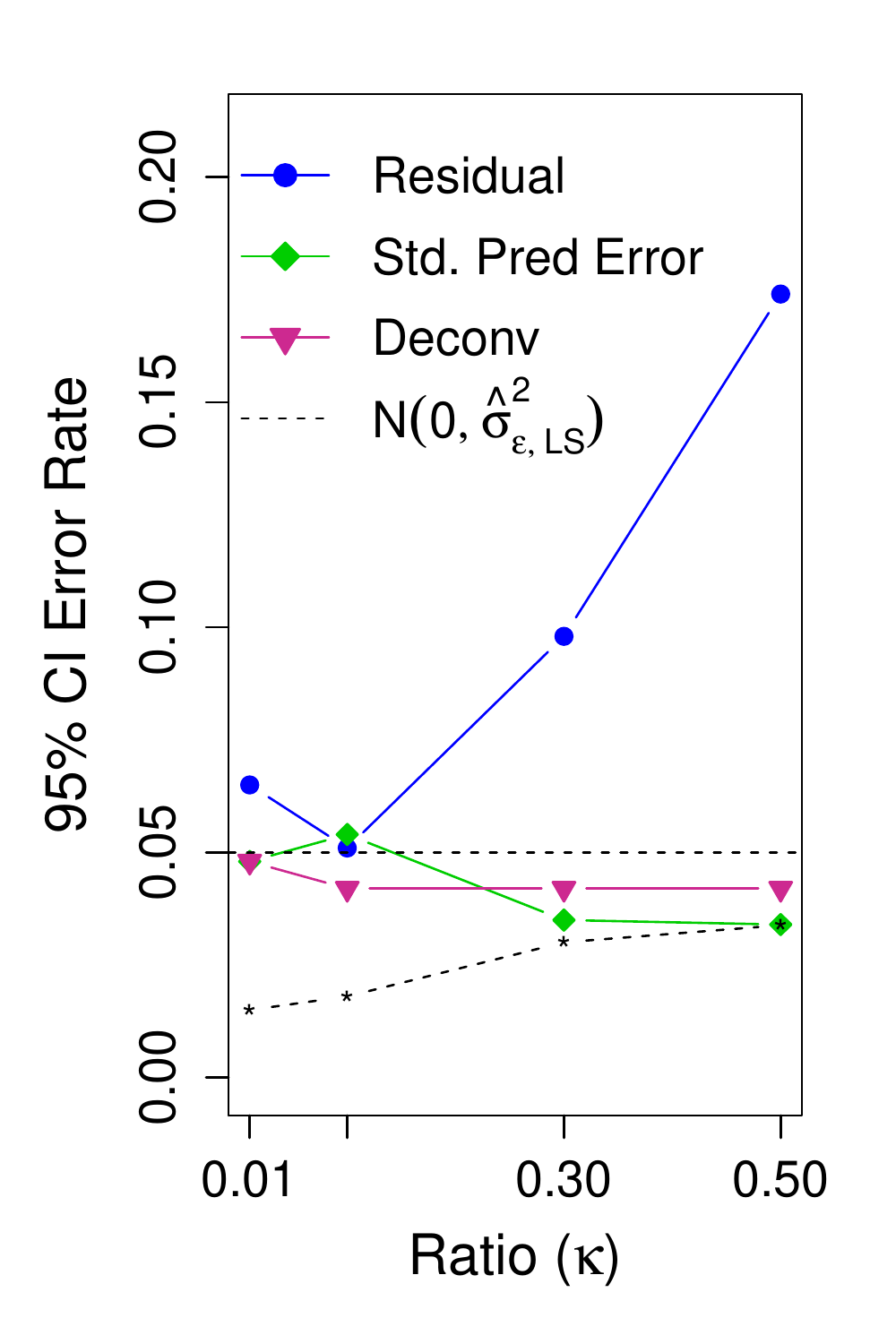} \label{subfig:bootErrorLeaveOut:Huber}}

\caption{\textbf{Bootstrap based on predicted errors:} We plotted the error rate of 95\% confidence intervals for the alternative bootstrap methods described in Section \ref{subsec:LeaveOut}: bootstrapping from standardized predicted errors (green) and from deconvolution of predicted error (magenta). We demonstrate its improvement over the standard residual bootstrap (blue) for \protect\subref{subfig:bootErrorLeaveOut:L1} $L_1$ loss and \protect\subref{subfig:bootErrorLeaveOut:Huber} Huber loss. The error distribution is double exponential and the design matrix $X$ is Gaussian, but otherwise the simulations parameters are as in Figure \ref{fig:basicCIError}. The error rates on confidence intervals based on bootstrapping from a $N(0,\widehat{\sigma}_{\eps,LS}^2)$ (dashed curve) are as a lower bound on the problem. For the precise error rates see Table \ref{tab:bootErrorLeaveOut}.}
\label{fig:bootErrorLeaveOut}
\end{figure}

Given our previous discussion of the behavior of $\prederror$, it is somewhat surprising that resampling from the distribution of $\stprederror$ performed well in our simulations. Clearly a few cases exist where $\stprederror$ should work well as an approximation of $\eps_i$. We have already noted that as $p/n\rightarrow 0$, the effect of the convolution with the Gaussian disappears since $\norm{\betaHat_\rho-\beta}\tendsto 0$; in this case both $\resid_i$ and $\stprederror$ should be good estimates of $\epsilon_i$. Similarly, in the case $\eps_i\sim N(0,\sigma^2)$, Equation \eqref{eq:residual} tells us that $\prederror$ are also asymptotically marginally normally distributed, so that correcting the variance should result in $\stprederror$ having the same distribution as $\epsilon_i$, at least when $X_{i,j}$ are i.i.d.

Surprisingly, for larger $p/n$ we do not see a deterioration of the performance of bootstrapping from $\stprederror$. This is unexpected, since as $p/n\rightarrow 1$ the risk $\norm{\betaHat_\rho-\beta}_2^2$ grows to be much larger than $\sigma_\eps^2$ (a claim we will make more precise in the next section); together with Equation \eqref{eq:residual}, this implies that $\stprederror$ is essentially distributed $N(0,\widehat{\sigma}_{\eps,LS}^2)$ as  $p/n\rightarrow 1$ regardless of the original distribution of $\epsilon_i$. This is confirmed in Figure  \ref{fig:bootErrorLeaveOut} where we  superimpose the results of bootstrap confidence intervals from when we simply estimate $\hat{\epsdist}$ with $N(0,\hat{\sigma}^2_{\eps,LS})$; we see the Type I error rate of the confidence intervals based on bootstrapping from $\stprederror$ do indeed approach that of $N(0,\hat{\sigma}_{\eps,LS}^2)$. Putting these two pieces of information together leads to the conclusion that as $p/n\rightarrow 1$ we can estimate $\hat{\epsdist}$ simply as $N(0,\hat{\sigma}_{\eps,LS})$ regardless of the actual distribution of $\epsilon$. 

In the next section we give some theoretical results that seek to understand this phenomenon.

\subsection{Behavior of the risk of $\betaHat$ when $\kappa\tendsto 1$}\label{subsec:PerfMethod2}

In the previous section we saw even if the distribution of the bootstrap errors $\eps^*_i$ given by $\hat{\epsdist}$, is not close to that of $\epsdist$, we can sometime get accurate bootstrap confidence intervals. For example, in least squares Equation \eqref{eq:residual} makes clear that even the standardized residuals, $\stresid_i$, do not have the same marginal distribution as $\epsilon_i$, yet they still provide accurate bootstrap confidence intervals in our simulations. We  would like to understand for what choice of distributions $\hat{\epsdist}$ will we see the same performance in our bootstrap confidence intervals of $\betaHat_1$?

When working conditional on $X$ as in residual resampling, the statistical properties of  $(\betaHat^*-\betaHat)$ differ from that of $(\betaHat-\beta)$ only because the errors are drawn from a different distribution -- $\hat{\epsdist}$ rather than $\epsdist$. Then to understand whether the distribution of  $\betaHat_1^*$ matches that of $\betaHat_1$ we can ask: what are the distributions of errors, $\epsdist$, that yield the same distribution for the resulting $\betaHat_1(\epsdist)$? In this section, we narrow our focus on understanding not the entire distribution of $\betaHat_1$, but only its variance. We do so because under assumptions on the design matrix $X$, $\betaHat_1$ is asymptotically normally distributed. This is true for both the classical setting of $\kappa=0$ and the high-dimensional setting of $\kappa\in(0,1)$ (see \SM Section \ref{supp:Reminders} for a review of these results and a more technical discussion). Our previous question is then reduced to understanding which distributions $\epsdist$ give the same  $\var{\betaHat_1(\epsdist)}$.

In the setting of least squares, it is clear that the only property of $\eps_i\iid\epsdist$ that matters for the variance of $\betaHat_{1,L_2}$ is $\sigma_\eps^2$, since $\var{\betaHat_{1,L_2}}=(X'X)^{-1}(1,1)\sigma^2_\eps$. For general $\rho$, if we assume $p/n\tendsto 0$, then $\var{\betaHat_{1,\rho}}$ will depend on features of $\epsdist$ beyond the first two moments (specifically through $\Exp{\psi^2(\eps)}/[\Exp{\psi'(\eps)}]^2$, \cite{HuberRobustRegressionAsymptoticsETCAoS73}). If we assume instead $p/n\tendsto \kappa\in (0,1),$ then $\var{\betaHat_{1,\rho}(\epsdist)}$ depends on $\epsdist$ via its influence on the squared risk of $\betaHat_\rho$, given by $\Exp{\norm{\betaHat_\rho(\epsdist)-\beta}_2^2}$ (see \SM Section \ref{supp:Reminders} for a review of these results). 

For this reason, in the setting of $p/n\tendsto \kappa\in(0,1)$, we need to characterize the risk of $\betaHat_\rho$ to understand when different distributions of $\epsilon$ result in the same variance of $\betaHat$. In what follows, we denote by $r_\rho(\kappa;\epsdist)$ the asymptotic risk of $\betaHat_{\rho}(\epsdist)$ as $p$ and $n$ tend to $\infty$. The dependence of  $r_\rho^2(\kappa;\epsdist)$  on $\epsilon$ is characterized by a system of two non-linear equations (given in \cite{NEKRobustPaperPNAS2013Published}, see \SM \ref{supp:Reminders}), and therefore it is difficult to characterize those distributions $\Gamma$ for which $r_\rho^2(\kappa;\epsdist)=r_\rho^2(\kappa;\Gamma)$. In the following theorem, however, we show that when $\kappa \tendsto 1$, the asymptotic squared risk $r_\rho^2(\kappa;\epsdist)$ converges to a constant that depends only on $\sigma_\eps^2$. This implies that when $\kappa \tendsto 1$, two different error distributions that have the same variances will result in estimators $\betaHat_{1,\rho}$ with the same variance. 

We now state the theorem formally; see Supplementary Text, Section \ref{supp:ResidProof} for the proof of this statement. 

\begin{theorem} \label{thm:asympPerfPNcloseto1}	
Suppose we are working with robust regression estimators, and $p/n\tendsto \kappa$. Assume that $X_{i,j}$ are i.i.d with mean 0 and variance 1, having Gaussian distribution or being bounded. Then, under the assumptions stated in \cite{NEKRobustRegressionRigorous2013} for $\rho$ and $\eps_i$'s, 
$$
r_\rho^2(\kappa;\epsdist)\sim_{\kappa\tendsto 1} \frac{\sigma^2_\eps}{1-\kappa}\;,
$$
provided $\rho$ is differentiable near 0 and $\rho'(x)\sim x$ near 0.
\end{theorem}
\noindent  Note that log-concave densities such as those corresponding to double exponential or Gaussian errors used in the current paper fall within the scope of this theorem. $\rho$ is required to be smooth and not grow too fast at infinity. So the theorem applies to the lest-squares problem, appropriately smoothed version of the $\ell_1$ or Huber losses, as well as the less well-known dimension-adaptive optimal loss functions described in \cite{NEKOptimalMEstimationPNASPublished2013}. We refer the reader to the \SM Section \ref{supp:Reminders} and \cite{NEKRobustRegressionRigorous2013} for details.

\paragraph{Implications for the Bootstrap} \label{review:implicationsBootstrap}
For the purposes of the residual-bootstrap, Theorem \ref{thm:asympPerfPNcloseto1} and our discussion in \SM Section \ref{supp:subsec:csqResidualBoot} imply that different methods of estimating the bootstrap distribution $\hat{\epsdist}$ will result in similar bootstrap confidence intervals as $p/n\tendsto 1$ if $\hat{\epsdist}$ has the same variance.  This agrees with our simulations, where both of our proposed  bootstrap strategies set the variance of $\hat{\epsdist}$ equal to $\widehat{\sigma}^2_{\eps,LS}$ and both had similar performance in our simulations for large $p/n$. Furthermore, as we noted, for $p/n$ closer to 1, they both had similar performance to a bootstrap procedure that simply sets $\hat{\epsdist}={\cal N}(0,\widehat{\sigma}^2_{\eps,LS})$ (Figure \ref{fig:bootErrorLeaveOut}). 

We return specifically to the bootstrap based on $\stprederror$, the standardized predicted errors. Equation \eqref{eq:residual} tells us that the marginal distribution of $\prederror$ is a convolution of the distribution of $\prederror$ and a normal, with the variance of the normal governed by the term $\norm{\betaHat_\rho-\beta}_2$. Theorem \ref{thm:asympPerfPNcloseto1} makes rigorous our previous assertion that as $p/n\tendsto 1,$ the normal term will dominate and the marginal distribution of $\prederror$ will approach normality, regardless of the distribution of $\epsilon$. However,  Theorem \ref{thm:asympPerfPNcloseto1} also implies that as $p/n\tendsto 1,$ inference for the coordinates of $\beta$ will be increasingly less reliant on features of the error distribution beyond the variance, implying that our standardized predited errors, $\stprederror$, will still result in an estimate $\hat{\epsdist}$ that will give accurate confidence intervals. Conversely, as $p/n \rightarrow 0$ classical theory tells us that the inference of $\beta$ relies heavily on the distribution $\epsdist$ beyond the first two moments, but in that case the distribution of $\stprederror$ approaches the correct distribution as we explained earlier. So bootstrapping from the marginal distribution of $\stprederror$ also makes sense when $p/n$ is small. 

For $\kappa$ between these two extremes it is difficult to theoretically predict the risk of $\betaHat_\rho(\hat{\epsdist})$ when the distribution $\hat{\epsdist}$ is given by resampling from the $\stprederror$. We turn to numerical simulations to evaluate this risk. 
 
Specifically, for $\epsilon_i\sim \epsdist$, we simulated data that is a convolution of $\epsdist$ and a normal with variance equal to $r^2_\rho(\kappa;\epsdist)$; we then scale this simulated data to have variance $\sigma_\eps^2$.  The scaled data are the $\eps_i^*$ and we refer to the distribution of $\eps_i^*$ as the convolution distribution, denoted $\epsdist_{conv}$. $\epsdist_{conv}$ is the asymptotic version of the marginal distribution of the standardized predicted errors, $\stprederror$, used in our bootstrap method proposed above. 

In Figure \ref{fig:relRisk} we plot for both Huber loss and $L_1$ loss the average risk $r_\rho(\kappa;\epsdist_{conv})$ (i.e errors given by $\epsdist_{conv}$) relative to the average risk $r_\rho(\kappa;\epsdist)$ (i.e errors distributed according to $\epsdist$), where $\epsdist$ has a double exponential distribution. We also plot the relative average risk
$r_\rho(\kappa;\epsdist_{norm})$, where $\epsdist_{norm}=N(0,\sigma_{\eps}^2)$. As predicted by Theorem \ref{thm:asympPerfPNcloseto1}, for $\kappa$ close to 1,  $r_\rho(\kappa;\epsdist_{conv})/r_\rho(\kappa;\epsdist)$ and $r_\rho(\kappa;\epsdist_{norm})/r_\rho(\kappa;\epsdist)$ converge to 1. Conversely, as $\kappa \rightarrow 0$, $r_\rho(\kappa;\epsdist_{norm})/r_\rho(\kappa;\epsdist)$ diverges dramatically from 1, while $r_\rho(\kappa;\epsdist_{conv})/r_\rho(\kappa;\epsdist)$ approaches 1, as expected. For Huber, the divergence of $r_\rho(\kappa;\epsdist_{conv})/r_\rho(\kappa;\epsdist)$ from 1 is at most 8\%, but the difference is larger for $L_1$ (12\%), probably due to the fact that the convolution with a normal error has a larger effect on the risk for $L_1$.

\begin{figure}[h]
\centering
\subfloat[][Relative risk of $\epsdist_{conv}$ to $\epsdist$ ]{\includegraphics[width=.4\textwidth]{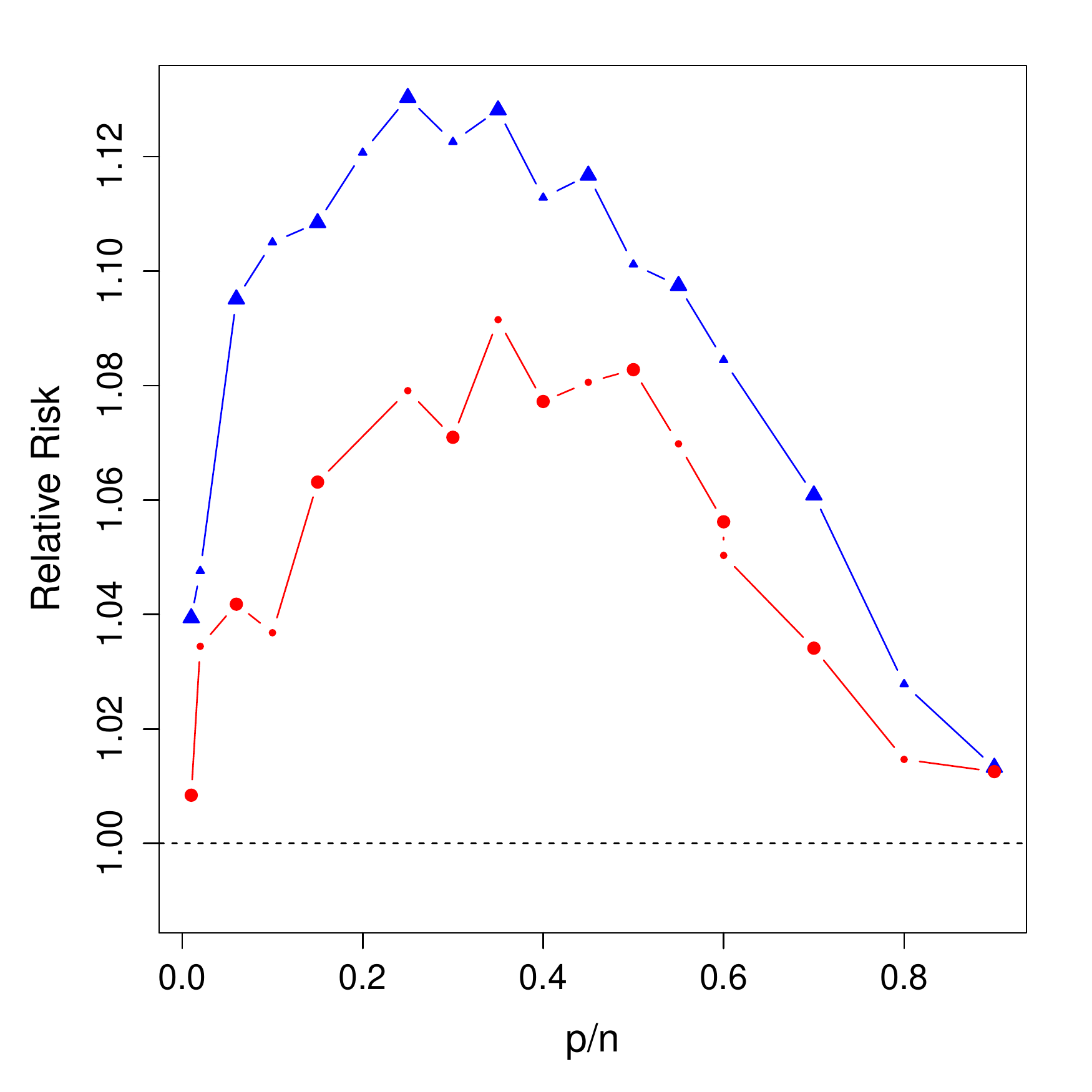}\label{subfig:relRisk:NoN}}
\subfloat[][Relative risk of $\epsdist_{conv}$ and $\epsdist_{norm}$ to $\epsdist$]{\includegraphics[width=.4\textwidth]{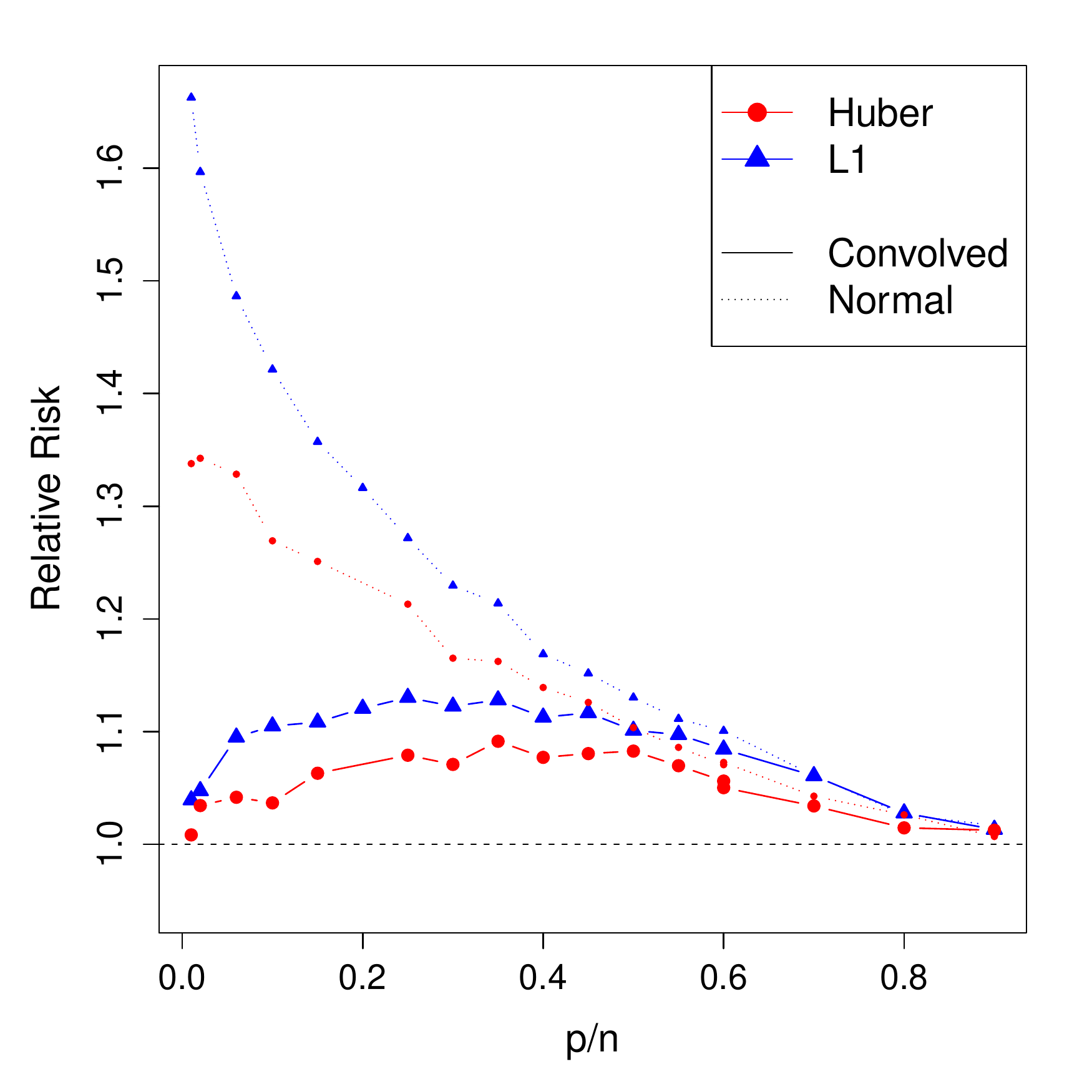}\label{subfig:relRisk:N}}
\caption{\textbf{Relative Risk of $\betaHat$ for scaled predicted errors vs original errors - population version:} \protect\subref{subfig:relRisk:NoN} Plotted with a solid lines are the ratios of the average risk of $\betaHat(\epsdist_{conv})$ to the average risk of $\betaHat(\epsdist)$ for Huber and $L_1$ loss.
\protect\subref{subfig:relRisk:N} shows the same plot, but with  the relative risk of $\betaHat(\rho)$ when the errors are distributed $\epsdist_{norm}={\cal N}(0,\sigma_\eps^2)$ added to the plot  (dotted lines). For both figures, the y-axis gives the relative risk, and the x-axis is the ratio $p/n$, with $n$ fixed at 500. Blue/triangle plotting symbols indicate $L_1$ loss; red/circle plotting symbols indicate Huber loss. The average risk is calculated over 500 simulations. The ``true" error distribution $\epsdist$ is the standard Laplacian distribution with  $\sigma^2_\eps=2$. Each simulation uses the standard estimate of $\sigma^2_\eps$ from the generated $\eps_i$'s. $r_\rho(\kappa;\epsdist)$ was computed using a first run of simulations with $\eps_i\iid \epsdist$. The Huber loss in this plot is $\text{Huber}_1$ and not the default $\text{Huber}_{1.345}$ of the \texttt{rlm} function.
}\label{fig:relRisk} 
\end{figure}

\section{Pairs Bootstrap} \label{sec:PairsBoot}
As described above, estimating the distribution $\hat{F}$ from the empirical distribution of $(y_i,X_i)$ (\emph{pairs bootstrapping}) is generally considered the most general and widely applicable method of bootstrapping, allowing for the linear model to be incorrectly specified (i.e $\Exp{y_i}$ is not a linear function of $X_i$). It is also considered to be slightly more conservative compared to bootstrapping from the residuals. In the case of random design, it makes also a lot of intuitive sense to use the pairs bootstrap, since resampling the predictors might be interpreted as mimicking the data generating process.

However, as in residual bootstrap, it is clear that the pairs bootstrap will have problems, at least in quite high  dimensions. In fact, when resampling the $X_i$'s from $\hat{F}$, the number of times a certain vector $X_{i_0}$ is picked has asymptotically $\textrm{Poisson}(1)$ distribution. So the expected number of different vectors appearing in the  bootstrapped design matrix $X^*$ is $n (1-1/e)$. When $p/n$ is large, with increasingly high probability the bootstrapped design matrix $X^*$ will no longer be of full rank. For example, if $p/n>(1-1/e)\approx 0.63$ then with probability tending to one as $n\tendsto \infty$,  the bootstrapped design matrix $X^*$ is singular, even when the original design matrix $X$ is of rank $p<n$. Bootstrapping the pairs in that situation makes little statistical sense. 

For smaller ratios of $p/n$, we evaluate the performance of pairs bootstrapping on simulated data. We see that the performance of the bootstrap for inference also declines dramatically as the dimension increases, becoming increasingly conservative (Figure \ref{fig:basicCIError}). In pairs bootstrapping, the error rates of 95\%-confidence-intervals drop far below the nominal 5\%, and are essentially zero for the ratio of $p/n=0.5$. Like residual bootstrap, this overall trend is seen for all the settings we simulated under (Supplemental Figures \ref{fig:basicCIErrorLap}, \ref{fig:basicCIErrorDesign}). For $L_1$ loss, even ratios as small as $0.1$ yield incredibly conservative bootstrap confidence intervals for $\betaHat_1$, with the error rate dropping to less than 0.01. For Huber and $L_2$ losses, the severe loss of power in our simulations starts for ratios of $0.3$ (see Tables \ref{tab:basicCIError},\ref{tab:basicCIErrorDesign}, \ref{tab:basicCIErrorLap}).

A minimal requirement for the distribution of the bootstrapped data to give reasonable inferences is that the variance of the bootstrap estimator $\betaHat^*_1$ needs to be a good estimate of the variance of $\betaHat_1$. This is not the case in high-dimensions. In Figure \ref{fig:bootOverEstFactor} we plot the ratio of the variance of $\betaHat^*_1$ to the variance of $\betaHat_1$ evaluated over simulations. We see that for $p/n=0.3$ and design matrices $X$ with i.i.d. ${\cal N}(0,1)$ entries, the  average variance of $\betaHat^*_1$ roughly overestimates the true variance of $\betaHat_1$ by a factor 1.3 in the case of least-squares; for Huber and $L_1$ the bootstrap estimate of variance is roughly twice as large as it should be (Table \ref{tab:bootOverEstFactor}).

In the case of least-squares, we can further quantify this loss in power by comparing the size of the bootstrap confidence intervals to the size of the correct confidence interval based on theoretical results (Figure \ref{fig:CIWidth}).  We see that even for ratios $\kappa$ as small as $0.1$, the confidence intervals for some design matrices $X$ were 15\% larger for pairs bootstrap than the correct size (e.g.  the case of elliptical distributions where $\lambda_i$ is exponential). For much higher dimensions of $\kappa=0.5$, the simple case of i.i.d normal entries for the design matrix gives intervals that are 80\% larger than needed; for the elliptical distributions we simulated, the width of the bootstrap confidence interval was as much as 3.5 times larger than that of the correct confidence interval. Furthermore, as we can see in Figure \ref{fig:basicCIError}, least-squares regression represents the best case scenario; $L_1$ and Huber will have even worse loss of power and at smaller values of $\kappa$.

\begin{SCfigure}
	\centering
	\includegraphics[width=.5\textwidth]{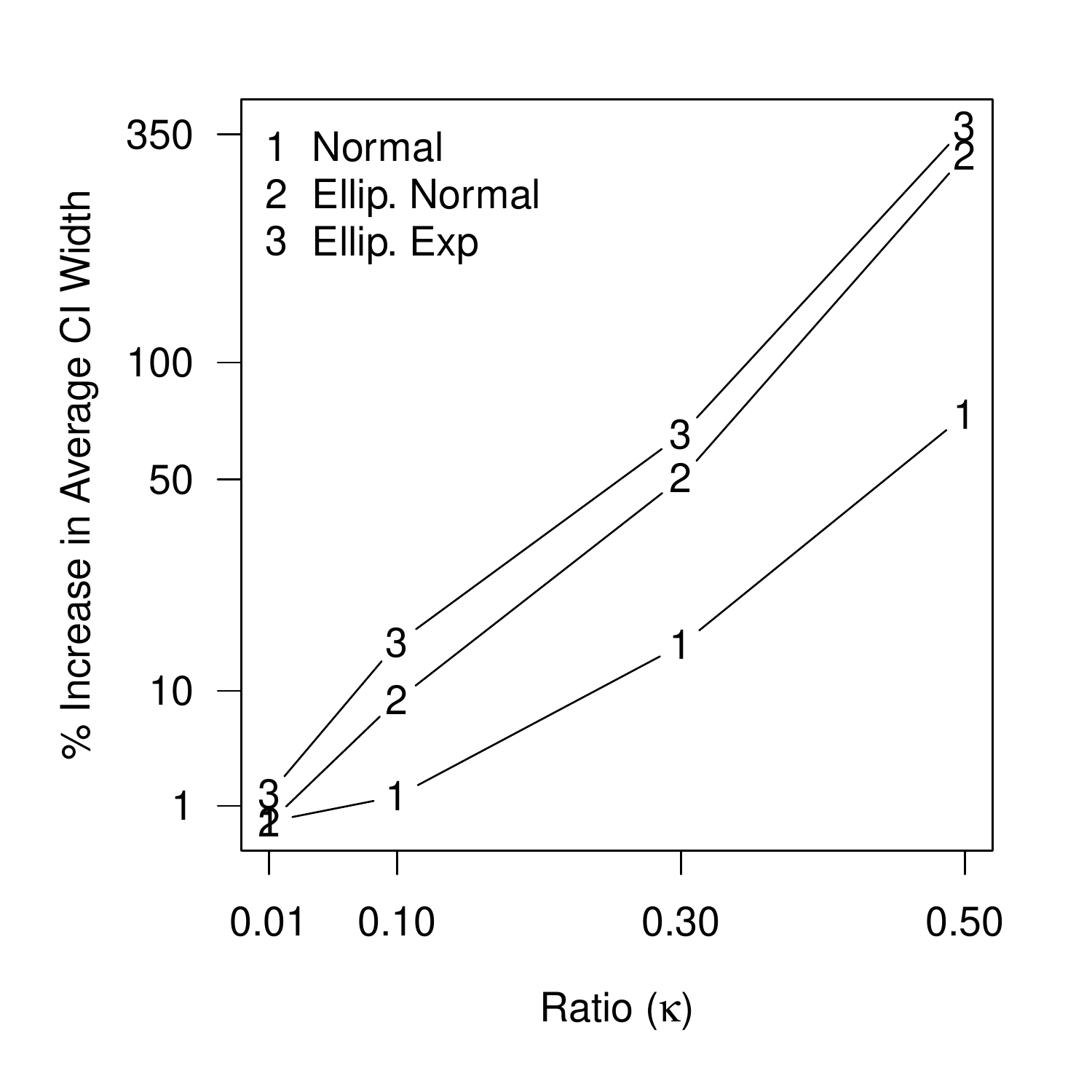} 
\caption{
 \textbf{Comparison of width of 95\% confidence intervals of $\beta_1$ for $L_2$ loss: }  Here we demonstrate the increase in the width of the confidence interval due to pairs bootstrapping. Shown on the y-axis is the percent increase of  the average confidence interval width based on simulation ($n=500$), as compared to the average for the standard confidence interval based on normal theory in $L_2$; the percent increase is plotted against the ratio $\kappa=p/n$ (x-axis).  Shown are three different choices in simulating the entries of the design matrix $X$:  (1) $X_{ij}\sim N(0,1)$ (2) elliptical $X_{ij}$ with $\lambda_i\sim N(0,1)$ and (3) elliptical $X_{ij}$ with $\lambda_i\sim Exp(\sqrt{2})$. The methods of simulation are the same as described in Figure \ref{fig:basicCIError}; exact values are given in Table \ref{tab:CIWidth}.
}\label{fig:CIWidth}
\end{SCfigure}

\subsection{Theoretical analysis for least-squares}\label{subsec:theormPairs}
In the setting of least-squares, we can for some distributions of the design matrix $X$ theoretically determine the asymptotic expectation of the variance of $v\trsp \betaHat^*$ and show that it is a severe over-estimate of the true variance of  $v\trsp \betaHat$.

We first setup some notation for the theorem that follows. Define $\betaHat_w$ as the result of regressing $y$ on $X$ with random weight $w_i$ for each observation $(y_i,X_i)$. In other words, 
$$
\betaHat_w=\argmin_{u\in \mathbb{R}^p} \sum_{i=1}^n w_i (y_i-X_i\trsp u)^2\;.
$$
We assume that the weights are independent of $\{y_i,X_i\}_{i=1}^n$ and define $\betaHat^*_w$ to be the random variable with distribution equal to that of $\betaHat_w$ conditional on the data $\{y_i,X_i\}_{i=1}^n$, i.e. $\betaHat^*_w\overset{\mathcal{L}}{=}\betaHat_w|\{y_i,X_i\}_{i=1}^n$. 
For the standard pairs bootstrap, the distribution of $\betaHat^*$ from resampling from the pairs $(y_i,X_i)$ is equivalent to the distribution of $\betaHat_w^*$, where $w$ is drawn from a multinomial distribution with expectation $1/n$ for each entry. In which case, the variance of $v\trsp \betaHat^*_w$ refers to the standard bootstrap estimate of variance given by the distribution of $v\trsp \betaHat^*$ over repeated resampling from the pairs $(y_i,X_i)$.

We have the following result for the expected value of the bootstrap variance of any contrast $v\trsp \betaHat^*_w$ where $v$ is deterministic, assuming independent weights with a Gaussian design matrix $X$ and some mild conditions on the distribution of the $w$'s.

\begin{theorem} \label{thm:bootVariance}
Let the weights $(w_i)_{i=1}^n$ be i.i.d.  and without loss of generality that $\Exp{w_i}=1$; we suppose that the $w_i$'s have 8 moments and for all $i$, $w_i>\eta>0$. Suppose $X_i$'s are i.i.d ${\cal N}(0,\Sigma)$, $\Sigma$ is positive definite and the vector $v$ is deterministic with $\norm{v}_2=1$.
	
Suppose $\betaHat$ is obtained by solving a least-squares problem and $y_i=X_i\trsp \beta+\eps_i$, $\eps_i$'s being i.i.d mean 0, with $\var{\eps_i}=\sigma^2_\eps$.
		
If $\lim p/n=\kappa<1$ then the expected variance of the bootstrap estimator, asymptotically as $n\tendsto \infty$, is given by
$$
p\frac{\Exp{\var{v\trsp \betaHat^*_w}}}{v\trsp \Sigma^{-1} v}=p\frac{\Exp{\var{v\trsp \betaHat_w|\{y_i,X_i\}_{i=1}^n}}}{v\trsp\Sigma^{-1}v}\tendsto \sigma^2_\eps \left[\kappa\frac{1}{1-\kappa-\Exp{\frac{1}{(1+cw_i)^2}}}-\frac{1}{1-\kappa}\right]\;,
$$
where $c$ is the unique solution of  $\Exp{\frac{1}{1+cw_i}}=1-\kappa$.
\end{theorem}
\noindent For a proof of this theorem and a consistent estimator of this limit, see Supplementary Text, Section \ref{supp:BootVarProof}. We note that $\Exp{\frac{1}{(1+cw_i)^2}}\geq \left[\Exp{\frac{1}{1+cw_i}}\right]^2=(1-\kappa)^2$ - where the first inequality comes from Jensen's inequality, and therefore the expression we give for the expected bootstrap variance is non-negative. 

In  \ref{subsubsec:extensionsThmBootVariance} below, we discuss possible extensions of this theorem, such as different design matrices. Before doing so, we first will discuss the implications of this result to pairs bootstrapping.

\paragraph{Multinomial Weights} In the standard pairs bootstrap, the weights are chosen according to a Multinomial$(n,1/n)$ distribution. This violates two conditions in the previous theorem: independence of $w_i$'s and the condition $w_i>0$.  In what follows, we use i.i.d $\textrm{Poisson}(1)$ weights as a proxy for the Multinomial$(n,1/n)$ to develop intuition about this latter case (see \SM Section \ref{supp:MultinomialVsPoisson} for technical details addressing these issues).

\paragraph{Implications for Pairs Bootstrap} \label{review:implicationsPairs} We can use the formula in Theorem \ref{thm:bootVariance} to explain why pairs bootstrap confidence intervals perform poorly in high-dimensions, at least for least squares regression with Gaussian design matrix. 

When $X_i\iid \mathcal{N}(0,\Sigma)$,  it is well known in the least-squares case that the quantity $p\ \var{v\trsp \betaHat}/v\trsp \Sigma^{-1} v$ converges asymptotically to $\kappa/(1-\kappa)\sigma^2_\eps$ (this can be shown through simple Wishart computations \citep{Haff79IdentityWishartDWithApps,mardiakentbibby}). If the variance of $v\trsp \betaHat_w^*$ converged to the variance of $v\trsp \betaHat$, we should be able to equate this latter quantity to the limit given in Theorem \ref{thm:bootVariance}, i.e.,
$$
\left[\kappa\frac{1}{1-\kappa-\Exp{\frac{1}{(1+cw_i)^2}}}-\frac{1}{1-\kappa}\right]=\frac{\kappa}{1-\kappa}\;,
$$
and hence should have
$$
\Exp{\frac{1}{(1+cw_i)^2}}=\frac{1-\kappa}{1+\kappa}\;.
$$

However, this relationship does not hold for most weight distributions.  In particular for weights following a  Poisson(1) distribution (which asymptotically corresponds to the standard pairs bootstrap), numerical calculations show that this relationship does not hold, and thus pairs bootstrap does not correctly estimate the variance of $v\trsp \betaHat$. In Figure \ref{subfig:bootOverEstFact:Theory} we calculate the theoretical predictions of $\Exp{\var{\betaHat_w^*}}$ given by Theorem \ref{thm:bootVariance} (using Poisson(1) weights and $\Sigma=\id_p$), and we compare them to the asymptotic variance of $\betaHat_1$ given by $\kappa/(1-\kappa)\sigma^2_\eps/p$. We see that Theorem \ref{thm:bootVariance} predicts that the pairs bootstrap overestimates the variance of the estimator by a factor that ranges from 1.2 to 3 as $\kappa$ varies between $0.3$ and $0.5$. These theoretical predictions correspond to the level of overestimation of the variance seen in our bootstrap simulations (Figure \ref{subfig:bootOverEstFact:Sim}).

\begin{figure}[t]
\centering
\subfloat[][$L_2$ (Theoretical)]{	\includegraphics[width=.4\textwidth, trim=0 -4.5cm 0 0]{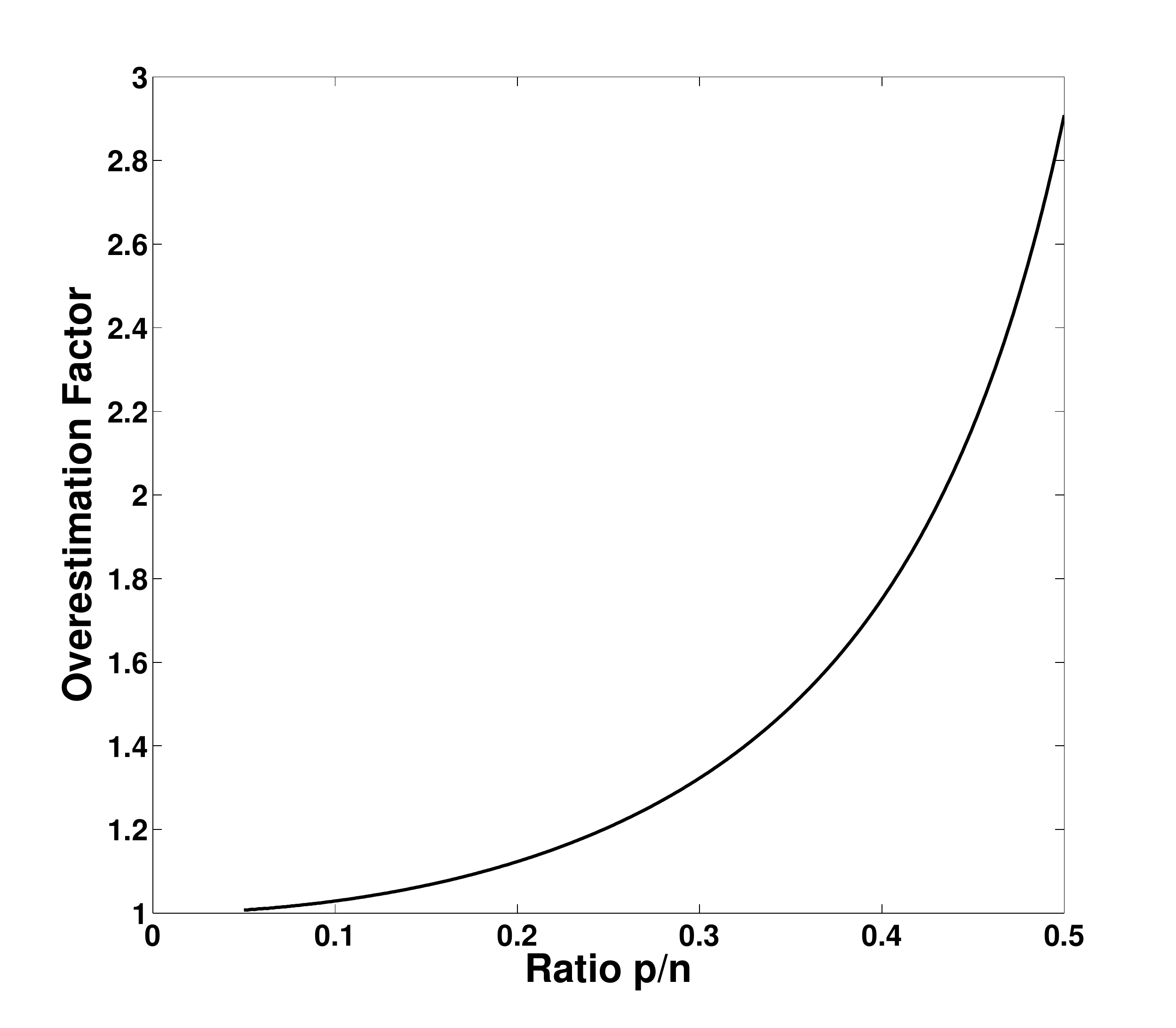}
\label{subfig:bootOverEstFact:Theory}}
\subfloat[][All (Simulated)]{
\includegraphics[width=.65\textwidth]{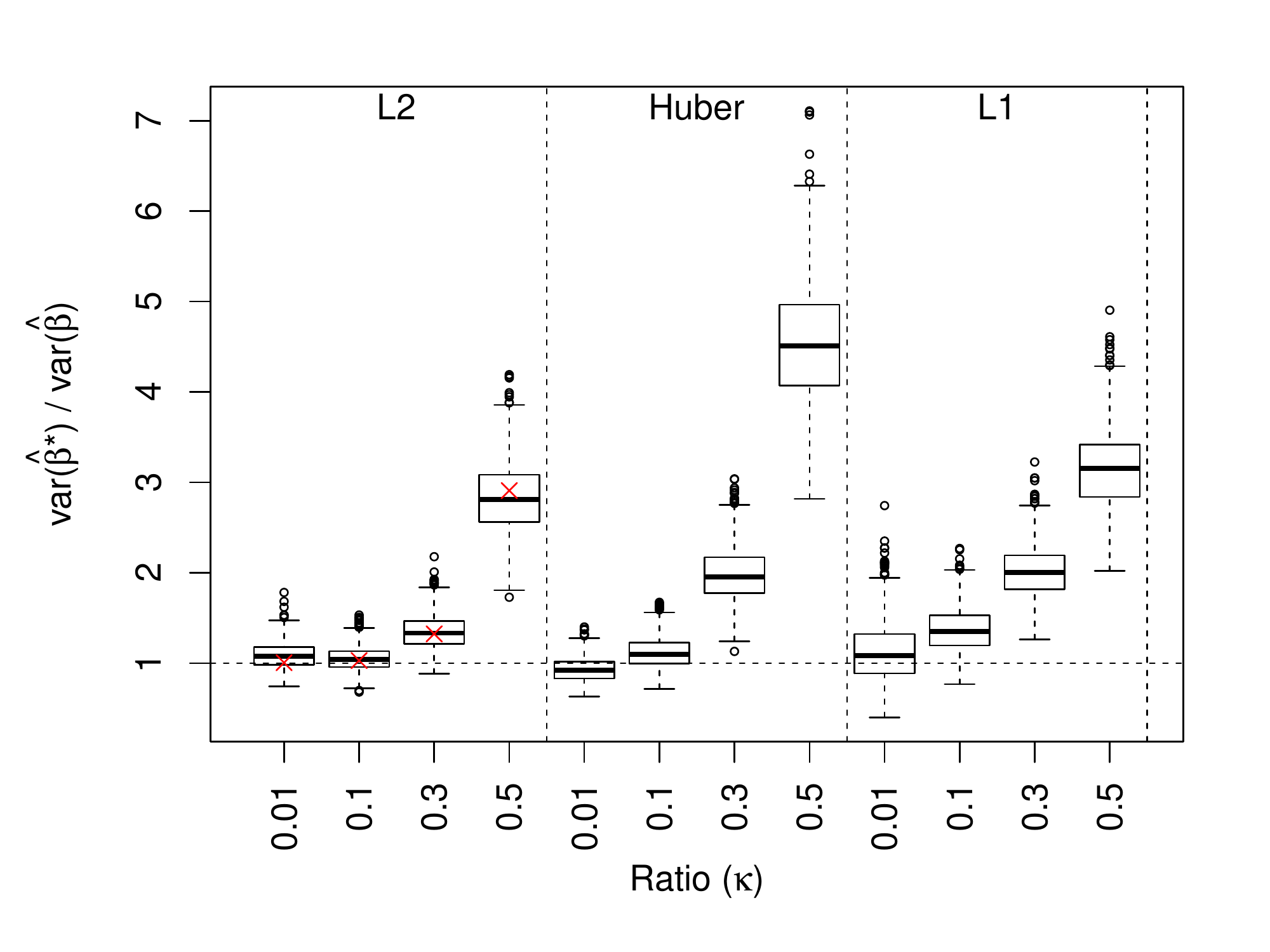}
\label{subfig:bootOverEstFact:Sim}
}
\caption{\textbf{Factor by which standard pairs bootstrap over-estimates the variance:} \protect\subref{subfig:bootOverEstFact:Theory} plotted is the ratio of the value of the expected bootstrap variance computed from Theorem \ref{thm:bootVariance} using Poisson(1) weights to the asymptotic variance $\kappa/(1-\kappa)\sigma^2_\eps$. 
\protect\subref{subfig:bootOverEstFact:Sim} boxplots of the ratio of the bootstrap variance of $\betaHat_1^*$ to the variance $\betaHat_1$, as calculated over 1000 simulations (i.e.  $\var{\betaHat}$ is estimated across simulated design matrices $X$, and not conditional on $X$). The theoretical prediction for the mean of the distribution from Theorem \ref{thm:bootVariance} is marked with a `X' for $L_2$ regression. Simulations were performed with normal design matrix $X$ and normal error $\epsilon_i$ with values of $n=500$.  For the median values of each boxplot, see Supplementary Table \ref{tab:bootOverEstFactor}. 
}\label{fig:bootOverEstFactor}
\end{figure}

\subsubsection{Extensions of Theorem \ref{thm:bootVariance}}\label{subsubsec:extensionsThmBootVariance}

\paragraph{Case of elliptical design} In light of previous work  on model robustness issues in high-dimensional statistics (see e.g \citep{DiaconisFreedmanProjPursuit84,HallMarronNeemanJRSSb05,nekCorrEllipD,nekMarkoRiskPub2010}), it is natural to ask whether the central results of Theorem \ref{thm:bootVariance} still apply when $X_i\equalInLaw \lambda_i Z_i$, with $\lambda_i$ a random variable independent of $Z_i$, and $Z_i\sim{\cal N}(0,\Sigma)$. We require $\Exp{\lambda_i^2}=1$ so that $\scov{X_i}=\Sigma$, as in the assumptions of Theorem \ref{thm:bootVariance}. The short answer is that the formula in Theorem \ref{thm:bootVariance} does not apply directly to this case. However, the proof given in \SM Section \ref{supp:BootVarProof} can be extended to that setting. We refer the interested reader to the Supplementary Text Section \ref{supp:PairsEllipticalDesign}  for more details.

\paragraph{Going beyond the Gaussian design} As explained in several papers 
in random matrix theory, a number of the quantities appearing in our theorems will converge to the same limit when i.i.d Gaussian predictors are replaced by i.i.d predictors with mean 0 and variance 1 and enough moments (an example being bounded random variables). Since our proof relies on random-matrix-theoretic arguments, the results we present here should be fairly robust to changing normality assumptions to i.i.d-ness assumptions for the entries of the design matrix $X$. The technical work necessary for making this rigorous, however, is beyond the scope of this paper.

\subsection{Alternative weight distributions for resampling}\label{subsec:newBootWeights}
The formula given in Theorem \ref{thm:bootVariance} suggests that resampling from a distribution $\hat{F}$ defined using weights other than i.i.d Poisson(1) (or, equivalently for our asymptotics, Multinominal(n,1/n)) should give us better bootstrap estimators than using the standard pairs bootstrap.  In fact, we should require, at least, that the bootstrap expected variance of these estimators match the correct variance $\var{v\trsp \betaHat}=\kappa/(1-\kappa)\sigma^2_\eps/p$ (for the Gaussian design, when $\Sigma=\id_p$).   We focus our discussion on the case $\Sigma=\id_p$; see \SM Section \ref{supp:subsec:acceptableWeightDistribution} for the case $\Sigma\neq \id_p$.

We note that if we use $w_i=1,\  \forall i$, the bootstrap variance will be 0, since with such a resampling scheme the resampled dataset is always the original dataset.  On the other hand, we have seen that with $w_i\sim \text{Poisson}(1)$, the expected bootstrap variance was too large compared to $\kappa/(1-\kappa)\sigma^2_\eps/p$. Hence, we tried to find alternative weights via calculating a parameter $\alpha$ such that if
\begin{equation}\label{eq:wtEq}
w_i\iid 1-\alpha+ \alpha \text{Poisson}(1)\;,
\end{equation}
the expected bootstrap variance would match the theoretical value of $\kappa/(1-\kappa)\sigma^2_\eps/p$.

We solved numerically this problem to find $\alpha(\kappa)$ (see Supplementary Table \ref{tab:GoodWeightProportions} and Supplementary Text, Subsection \ref{supp:subsec:acceptableWeightDistribution} for details of computation). 
We then used these values and performed bootstrap resampling using the weights defined in Equation \eqref{eq:wtEq}. We evaluated bootstrap estimate of $\var{\betaHat_1}$ as well as the confidence interval coverage of the true $\beta_1$. We find that this adjustment of the weights in estimating $\hat{F}$ results in accurate bootstrap  estimates of variance and appropriate levels of confidence interval coverage  (Table \ref{tab:SummaryCorrectedBootSimulations}).

However, small changes in the choice of $\alpha$ can result in fairly large changes in $\Exp{\var{v\trsp \betaHat_w|X,\eps}}$. For instance, for $\kappa=0.5$, using the value of $\alpha=0.95$ which is  close to the correct value of $\alpha(0.5)=0.92$ results in an expected bootstrap variance roughly 30\% larger than it should be.

\begin{table}
\begin{center}
\begin{tabular}{|c|c|c|c|c|c|c|c|c|c|c|}\hline
&\multicolumn{4}{c|}{$\kappa$}\\
 & \multicolumn{1}{c}{.1} & \multicolumn{1}{c}{.2} &   \multicolumn{1}{c}{.3} &  \multicolumn{1}{c|}{.5}\\ \hline
$\alpha$ & .9875& .9688  &  .9426  &  .9203  \\ 
Error Rate of 95\% CIs & 0.051 & 0.06 & 0.061 &  0.057 \\ 
Ratio of Variances & 1.0119 &   1.0236 & 0.9931 & 0.9992 \\ \hline
\end{tabular}
\end{center}
\caption{\textbf{Summary of weight-adjusted bootstrap simulations for $L_2$ : } Given are the results of performing bootstrap resampling for $n=500$ according to the estimate of $\hat{F}$ given by the weights in Equation \eqref{eq:wtEq}.  ``Error Rate of 95\% CIs" denotes the percent of bootstrap confidence intervals that did not containing the correct value of the parameter $\beta_1$. ``Ratio of Variances" gives the ratio of the empirical expected bootstrap variance over our simulations divided by the theoretical value $\sigma^2_\eps\kappa/(1-\kappa)$. Results are based on 1000 simulations, with a Gaussian random design and errors distributed as double exponential.}
\label{tab:SummaryCorrectedBootSimulations}
\end{table}

Moreover, this strategy for finding a good weight distribution requires knowing a great deal about the distribution of the design matrix. Hence the work we just presented on finding new weight distributions for bootstrapping is a proof of principle that alternative weighting schemes could be used for pairs bootstrapping in high-dimension, but important practical details would depend strongly on the statistical model that is assumed. This is in sharp contrast with the low-dimensional situation, where a unique and model-free bootstrap resampling technique works  in a broad variety of situations.

\section{The Jackknife} \label{sec:jackknife}
In the context we are investigating, where we know that the distribution of $\betaHat_1$ is asymptotically normal (see \SM \ref{supp:Reminders}), it is natural to ask whether we could simply use the jackknife to estimate the variance of $\betaHat_1$. The jackknife relies on leave-one-out procedures to estimate $\var{\betaHat_1}$. More specifically,  for a fixed vector $v$,  $\var{v\trsp \betaHat}$: 
\begin{equation}\label{eq:defJackVar}
\widehat{var}_{JACK}(v\trsp\betaHat)=\varJack=\frac{n-1}{n}\sum_{i=1}^n(v\trsp[\betaHat_{(i)}-\tilde{\beta}])^2
\end{equation}
where $\tilde{\beta}=\frac{1}{n}\sum_{i=1}^n \betaHat_{(i)}$. The case of $\betaHat_1$ corresponds to picking $v=e_1$, i.e the first canonical basis vector. The Efron-Stein inequality guarantees in general that the expectation of the jackknife estimate of variance gives an upper-bound on the variance of the statistic under consideration \citep{EfronStein81}.

Given the problems we just documented with the pairs bootstrap, it is natural to ask whether confidence intervals based on the jackknife estimate of variance  perform better than pairs bootstrap intervals in high-dimensions. The jackknife is known to have problems (\cite{EfronBook82} or \cite{KoenkerQuantileRegressionBook05}, p.105), 
 but the reliance of the  jackknife on leave-one-out estimates $\betaHat_{(i)}$ might suggest it could be more  robust to dimensionality issues than other methods. 

\paragraph{Empirical findings} As in the pairs bootstrap case, simulations show that confidence intervals based on the jackknife estimate of variance lead to extremely poor inference for $\beta_1$ (Figure \ref{fig:basicCIError}) and that the jackknife dramatically overestimates the variance of $\betaHat_1$ (Figure \ref{fig:bootOverEstFactorJack} and Supplementary table \ref{tab:bootOverEstFactor}). For $L_2$ and Huber loss, the jackknife estimate of variance is 10-15\% too large for $p/n=0.1$, and for $p/n=0.5$ the jackknife estimate of variance is 2-2.5 times larger than it should be. In the case of $L_1$ loss, the jackknife variance is completely erratic, even in low dimensions; this is not completely surprising given the known problems with the jackknife for the median \citep{KoenkerQuantileRegressionBook05}. Even for $p/n=0.01$, the estimate is not unbiased for $L_1$, with median estimates twice as large as they should be and enormous variance in the estimates of variance. Higher dimensions only worsen the behavior with jackknife estimates being 15 times larger than they should.  

\begin{SCfigure}
\centering
\includegraphics[width=.6\textwidth]{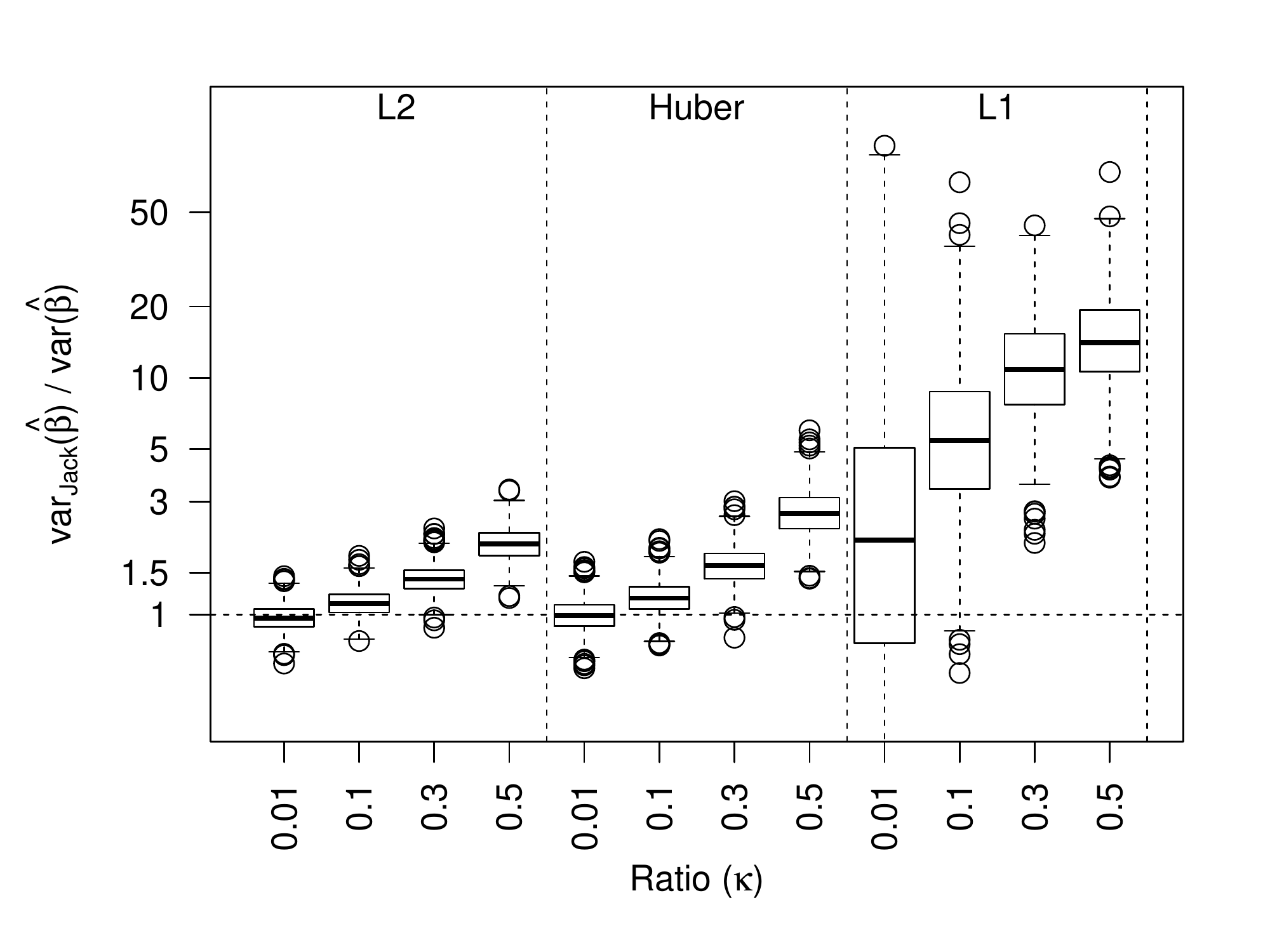}
\caption{\textbf{Factor by which jackknife over-estimates the variance:} boxplots of the ratio of the jackknife estimate of the variance $\betaHat_1$ to the variance of $\betaHat_1$ as calculated over 1000 simulations. Simulations were with normal design matrix $X$ and normal error $\epsilon_i$ with values of $n=500$. Note that because the $L_1$ jackknife estimates so wildly overestimate the variance, in order to put all the methods on the same plot the boxplot of ratio is on log-scale; y-axis labels give the corresponding ratio to which the log values correspond. For the median values of each boxplot, see Supplementary Table \ref{tab:bootOverEstFactor}.
}\label{fig:bootOverEstFactorJack}
\end{SCfigure}

\subsection{Theoretical results} 
Again, in the case of least-squares regression  with a Gaussian design matrix, we can theoretically evaluate the behavior of the jackknife. The proof of the following theorem is given in the \SM Section \ref{supp:JackknifeProof}.

\begin{theorem}\label{thm:Jackknife}
Let us call $\varJack$ the jackknife estimate of variance of $\betaHat_1$, the first coordinate of $\betaHat$. Suppose the design matrix $X$ is such that $X_{i}\iid {\cal N}(0,\Sigma)$.  Suppose $\betaHat$ is computed using least-squares and the errors $\eps$ have a variance. Then  we have, as $n,p\tendsto \infty$ and $p/n\tendsto \kappa<1$,  
$$
\frac{\Exp{\varJack}}{\var{\betaHat_1}}\tendsto \frac{1}{1-\kappa}\;.
$$
The same result is true for the jackknife estimate of variance of $v\trsp \betaHat$, where $v$ is any deterministic vector  with $\norm{v}_2=1$.
\end{theorem}

\paragraph{Correcting the Jackknife in Least Squares} Theorem \ref{thm:Jackknife} implies that scaling  
the jackknife estimate of variance by multiplying it by $1-p/n$ will result in an estimate of  $\var{\betaHat_1}$ with the correct expectation; simulations shown in Figure \ref{fig:jackCorrect} confirm that confidence intervals based on this corrected estimate of variance yield correct confidence intervals 
for least-squares estimates of $\betaHat$ when the design matrix $X$ is Gaussian. However this scaling factor is not robust to violations of these assumptions. In particular when the $X$ matrix follows an elliptical distribution the correction of $1-p/n$ from Theorem \ref{thm:Jackknife} gives little improvement even when the loss is still $L_2$ (Figure \ref{fig:jackCorrect}).

\newcommand{\weightedCovMatRobReg}{{\cal S}}
\paragraph{Corrections for more general settings} For the more general setting of an elliptical design matrix $X$  and loss function $\rho$, preliminary computations suggest an alternative result. Let $\weightedCovMatRobReg$ be the random matrix defined by 
$$
\weightedCovMatRobReg=\frac{1}{n}\sum_{i=1}^n \psi'(\resid_i) X_i X_i\trsp. 
$$
Then in our asymptotic regime, and when $\Sigma=\id_p$, preliminary heuristic calculations suggest that we can estimate the amount by which $\Exp{\varJack}$ overestimates the variance of $\betaHat_{1}$ by $\Exp{\hat{\gamma}}$, where
\begin{equation}\label{eq:heuristicJack}
\hat{\gamma}\triangleq \frac{\trace{\weightedCovMatRobReg^{-2}}/p}{\left[\trace{\weightedCovMatRobReg^{-1}}/p\right]^2}\;.
\end{equation}
Note that when applied to least-squares regression with $X\sim\mathcal{N}(0,Id_p)$ this conforms to our result in Theorem \ref{thm:Jackknife}. Theoretical considerations suggest that in our asymptotics, for smooth $\rho$, $\gammaHat\simeq \Exp{\gammaHat}$, which suggests a data-driven correction to the jackknife estimate of variance; however that correction depends having information about the distribution of the design matrix.

Equation \eqref{eq:heuristicJack} assumes that the loss function can be twice differentiated, which is not the case for either Huber or $L_1$ loss. In the case of non-differentiable $\rho$ and $\psi$, we can use appropriate regularizations to make sense of those functions. For $\rho=\text{Huber}_k$, i.e a Huber function that transitions from quadratic to linear at $|x|=k$, $\psi'$ should be understood as $\psi'(x)=1_{|x|\leq k}$. For $L_1$ loss, $\psi'$ should be understood as $\psi'(x)=1_{x=0}$. 

In Figure \ref{fig:jackCorrect} we show simulation results for confidence intervals created based on rescaling the jackknife estimate of variance by $\Exp{\hat{\gamma}}$  defined in Equation \eqref{eq:heuristicJack}. 
In the case of least-squares with an elliptical design matrix, this correction -- which directly uses the distribution of the observed $X$ matrix -- leads to a definite improvement in our jackknife confidence intervals. Similarly, for the Huber loss  we see a definite improvement as compared to the standard jackknife estimate, as well as an improvement over the simpler correction of $1-p/n$ that would be appropriate for squared error loss.

\begin{figure}[t]
\centering
\includegraphics[width=.85\textwidth]{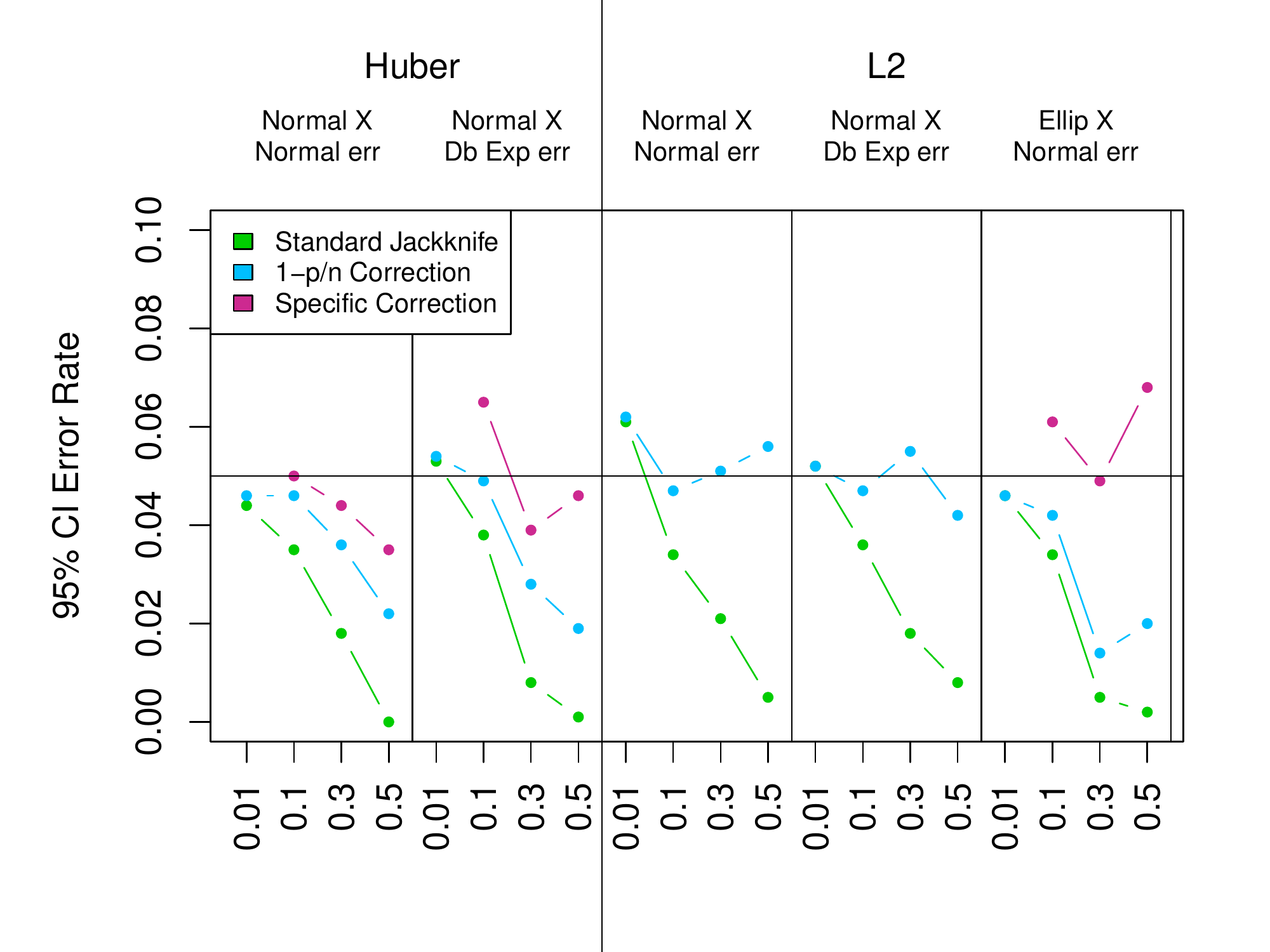}
\caption{\textbf{Rescaling jackknife estimate of variance:} Shown are the error rates for confidence intervals for different re-scalings of the jackknife estimate of variance: the  standard jackknife estimate (green); re-scaling using $1-p/n$ as given in Theorem \ref{thm:Jackknife} for the $L_2$ case with normal design matrix $X$ (blue); and re-scaling based on the heuristic in Equation \eqref{eq:heuristicJack} for those settings not covered by the assumptions of Theorem \ref{thm:Jackknife} (magenta). The Huber loss in this plot is $\text{Huber}_1$ rather than the default $\text{Huber}_{1.345}$;  $\text{Huber}_1$ is further from $L_2$ than $\text{Huber}_{1.345}$ and therefore better shows  the improvement gained by using the heuristic in Equation \eqref{eq:heuristicJack}.
}\label{fig:jackCorrect}
\end{figure}

It should be noted that the quality of this proposed correction seems to depend on how smooth $\psi$ is. In particular, even using the previous interpretations, the correction does not perform well for $L_1$ (at least for $n=1000$ and $\kappa=.1,.3,.5$, data not shown) - though as we mentioned Figure \ref{fig:bootOverEstFactorJack} shows that jackknifing in $L_1$-regression is probably not a good idea; see also \cite{KoenkerQuantileRegressionBook05}, Section 3.9. 
\subsection{Case of $\Sigma\neq \id_p$ and extensions of Theorem \ref{thm:Jackknife}}

The invariance arguments concerning $\scov{X_i}$ we give in the Supplementary Text Section \ref{sec:supp:CovAndInvarianceIssue} apply to all loss functions $\rho$ when $X_i$ has an elliptical distribution. In particular, if $\betaHat_\rho(\Sigma)$ denotes our estimator when $\scov{X_i}=\Sigma$, we have 
$$
\frac{\Exp{\varJack(v\trsp \betaHat_\rho(\Sigma))}}{\var{v\trsp \betaHat_\rho(\Sigma)}}=\frac{\Exp{\varJack(v\trsp \betaHat_\rho(\id_p))}}{\var{v\trsp \betaHat_\rho(\id_p)}}
$$
In Equation \eqref{eq:heuristicJack} we give a heuristically derived way to estimate this quantity. However, the $\gammaHat$ given there depends crucially on knowing that $\scov{X_i}=\id_p$ and cannot be used as-is when $\Sigma\neq \id_p$.

It is also natural to ask if Theorem \ref{thm:Jackknife} is likely to be true and can be extended to $X_{i,j}$'s being i.i.d with  mean 0 and variance 1. Since the proof of Theorem \ref{thm:Jackknife} is based on random matrix techniques, further technical work should allow such an extension, provided $X_{i,j}$'s have sufficiently many moments (see \SM \ref{supp:jackKnife:DesignIssue} for details).  

\section{Conclusion}
In this paper, we have studied various resampling plans in the high-dimensional setting where $p/n$ is not close to zero. One of our main findings is that the two most widely-used and advocated bootstraps will yield either highly conservative or highly anti-conservative confidence intervals. This is in sharp contrast to the low-dimensional setting where $p$ is fixed and $n\tendsto \infty$ or $p/n\tendsto 0$. Under various assumptions underlying our simulations, we explained theoretically the phenomena we were observing in our numerical work. 
We give improvements to these bootstrap methods that give confidence intervals with approximately correct coverage probability. 
However, these corrections were generally based on knowing or assuming certain non-trivial properties of the design matrix  - hence they violate the tenets of the bootstrap which promises a simple and universal numerical method to get accurate solutions to a broad class of problems. A possible exception is our proposal for resampling the standardized predicted errors. This bootstrap routine continued to perform reasonably well without distribution-specific corrections and has the potential to be a general-purpose bootstrap method for high dimensions. 
\color{black}

This work has focused on estimation of the linear model, where we can check the performance of the bootstrap against theoretical benchmarks.  The real practical power of the bootstrap lays in giving the ability to perform inference in complicated settings involving sophisticated statistical procedures for which we do not even begin to have theoretical results for the behavior of our estimators. Yet, even for the simple case of inference in the linear model and for the simplest inferential question, our work shows that the two most common and natural resampling techniques perform very poorly in only moderately high-dimensions. More importantly, these two   equally intuitive methods have completely divergent statistical  behavior with one being incredibly conservative and the other anti-conservative. This casts serious doubts about the reliability, interpretability and accuracy of inferential statements made through resampling methods in high dimensions, which is troubling for more complicated problems where resampling techniques are the only inference tools currently available. Our findings also raise many interesting new theoretical and methodological questions.

\bigskip
\begin{center}
{\large\bf SUPPLEMENTARY MATERIAL}
\end{center}

\begin{description}

\item[Supplementary Text] More detailed description of simulations and proofs of the theorems stated in main text  (see below; see also authors' website for different formatting)

\item[Supplementary Figures] Supplementary Figures referenced in the main text (pdf; see also authors' website for different formatting)

\item[Supplementary Tables] Supplementary Tables referenced in the main text (pdf; see also authors' website for different formatting)

\end{description}

\clearpage	
\renewcommand{\lowerBoundWs}{\eta}
\renewcommand{\bHat}{\widehat{b}}
\renewcommand{\gHat}{\widehat{g}}
\renewcommand{\resid}{e}
\renewcommand{\prederror}{\tilde{\resid}_{i(i)}}
\newcommand{\prederrorOnePredictorOut}{\tilde{\resid}_{i(i),[p]}}
\renewcommand{\prederrorAtj}{\tilde{\resid}_{j(i)}}
\renewcommand{\stresid}{r}
\renewcommand{\stprederror}{\tilde{\stresid}_{i(i)}}
\renewcommand{\stprederrorAtj}{\tilde{\stresid}_{j(i)}}
\renewcommand{\rTilde}{\widetilde{r}}
\renewcommand{\varJack}{\textrm{varJACK}}
\renewcommand{\epsdist}{G}
\renewcommand{\SM}{Supplementary Text, }

\appendix
\begin{center}
\textbf{\textsc{APPENDIX}}
\end{center}
\vspace{1cm}
\renewcommand{\thesection}{S\arabic{section}}
\renewcommand{\thelemma}{\arabic{section}-\arabic{lemma}}
\renewcommand{\thesubsection}{S\arabic{section}.\arabic{subsection}}
\renewcommand{\thesubsubsection}{S\arabic{section}.\arabic{subsection}.\arabic{subsubsection}}
\setcounter{equation}{0}  
\setcounter{lemma}{0}
\setcounter{corollary}{0}
\setcounter{section}{0}

\renewcommand{\thetable}{S\arabic{table}}   
\renewcommand{\thefigure}{S\arabic{figure}}
\setcounter{table}{0}
\setcounter{figure}{0}

\textbf{Notations :} in this appendix, we use $\resid_i$ to denote the $i$-th residual, i.e $\resid_i=y_i-X_i\trsp \betaHat$. We use $\prederror$ to denote the $i$-th prediction error, i.e $\prederror=y_i-X_i\trsp \betaHat_{(i)}$, where $\betaHat_{(i)}$ is the estimate of $\betaHat$ with the $i$-th pair $(y_i,X_i)$ left out. We assume that the linear model holds so that $y_i=X_i\trsp \beta+\eps_i$. We assume that the errors $\eps_i$ are i.i.d with mean 0.

\section{Deconvolution Bootstrap}\label{supp:deconvolution}
In the main text, we considered situations where our predictors $X_i$ are i.i.d with an elliptical distribution and assume for instance that $X_i=\lambda_i \xi_i$, where $\xi_i\sim {\cal N}(0,\Sigma)$ and $\lambda_i$ are i.i.d scalar random variables with $\Exp{\lambda_i^2}=1$. As described in the main text, if $X$ is elliptical, $\prederror$ is a convolution of the   correct $\epsdist$ distribution and a Normal distribution,
$$\prederror \simeq \epsilon_i + \tilde{Z}_i,$$
where $$\tilde{Z}_i\iid \mathcal{N}(0,\lambda_i^2\norm{\betaHat_{\rho(i)}-\beta}_2^2)$$ and are independent of $\epsilon_i$. 

We proposed in Section \ref{subsec:LeaveOut} of the main text an alternative bootstrap method based on using deconvolution techniques to estimate $\epsdist$ (Method 1). Specifically, we proposed the following bootstrap procedure: 
\begin{enumerate}
\item Calculate the predicted errors, $\prederror$ 
\item Estimate $|\lambda_i|\norm{\betaHat_{\rho(i)}-\beta}_2$ (the standard deviation of the $\tilde{Z}_i$)
\item Deconvolve in $\prederror$ the error term $\epsilon_i$ from the $\tilde{Z}_i$ term ; 
\item Use the resulting estimates of $\epsdist$ as the estimate of $\hat{\epsdist}$ in residual bootstrapping.
\end{enumerate} 

\subsection{Estimating $\bm{\norm{\betaHat_\rho-\beta}}$ and the variance of the $Z_i$}
\label{suppsec:estimateDeconVar} 
Deconvolution methods that deconvolve $\eps$ from the $\tilde{Z}_i$ require an estimate of the variance of the $\tilde{Z}_i$. Equation \eqref{eq:residual} gives the variance as 
$\lambda_i^2\norm{\betaHat_{\rho(i)}-\beta}_2^2,$ and we need to estimate this quantity from the data. We use the approximation $$\norm{\betaHat_{\rho(i)}-\beta}_2\simeq \norm{\betaHat_{\rho}-\beta}_2.$$ 
See \SM \ref{supp:Reminders} and references therein for justification of this approximation.

Furthermore, as we note in the main text, in our implementation of this deconvolution in simulations we assume $X\sim \mathcal{N}(0,Id_p)$ so that $\lambda_i=1$ (see Section \ref{supp:sec:estimationEllipParam} below for estimating $\lambda_i$ in the elliptical case). This means we are estimating the variance of $\tilde{Z}_i$ as $\norm{\betaHat-\beta}_2^2$ for all $i$. We estimate this as 
$$\widehat{var}(\tilde{Z}_i)=\widehat{var}(\prederror)-\hat{\sigma}^2_{\eps},$$
where $\widehat{var}(\prederror)$ is the standard estimate of variance
and $\hat{\sigma}^2_{\eps}$ is the estimate of variance from the least squares fit, $\widehat{\sigma}_{\eps,LS}^2$, defined in the main text.

In the case where $\widehat{var}(\prederror)\leq \hat{\sigma}^2_{\eps}$, we do not do a deconvolution, but simply bootstrap from the $\prederror$. This is generally only the case when $p/n$ is quite small. 

\subsection{Estimating $\hat{\epsdist}$}\label{supp:sec:bandwidth}
We used the deconvolution algorithm in the \texttt{decon} package in R \citep{Wang2011Decon} to estimate the distribution of $\epsilon_i$.  Deconvolution algorithms require selection of a bandwidth in the kernels that make up the functional basis of the estimate. The appropriate bandwidth parameter in deconvolution problems is tied intrinsically to the use of the estimate, with optimal bandwidths depending on what functional of the distribution is wanted (e.g. the pdf versus the cdf). Moreover, the optimal bandwidth depends on the distribution of $\tilde{Z}_i$ with which the signal is being convolved. Ultimately, our procedure resamples from the distribution $\hat{\epsdist}$, requiring estimates of $\epsdist^{-1}(y)$, and the distribution of $\tilde{Z}_i$ is Gaussian. There is no specific theory for the optimal bandwidth in this setting (though see the work of \cite{Hall2008} for optimal bandwidth selection for estimations of the quantiles of $\hat{\epsdist}$ if the $\tilde{Z}_i$ are distributed according to a distribution whose characteristic function decays polynomially at infinity - see Assumption (A.11) on p.2133; this is clearly violated in our case where $\tilde{Z}_i$ are normally distributed.)

We used the bandwidth estimation procedure \texttt{bw.dboot2} provided in the package \texttt{decon}. \cite{delaigleProblemDecon2014} outlines problems in the estimation of bandwidth parameter in \texttt{decon}; specifically that the  implementation in \texttt{decon} of existing bandwidth estimation procedures does not match their published descriptions. \texttt{bw.dboot2} was not one of the bandwidth procedures with these discrepancies. However, we also compared our results with a bandwidth selected via the bandwidth selection method of \cite{Delaigle2002,Delaigle2004} and used the R code implementation provided by the authors on \url{http://www.ms.unimelb.edu.au/~aurored/links.html#Code}. The two different choices in bandwidth, however, had little effect on the coverage of the bootstrap confidence intervals (Supplemental Figure \ref{supfig:bandwidth}). The results in Figure \ref{fig:bootErrorLeaveOut} in the main text make use of the bandwidth parameter of \cite{Delaigle2002,Delaigle2004}.

For both bandwidth selections, we  estimated the cdf using the function \texttt{DeconCdf} provided in the \texttt{decon} package and provided the bandwidth parameters described above. We specified the error distribution as `Normal' and set the variance of $\tilde{Z}_i$ as described above in Section \ref{suppsec:estimateDeconVar}. The number of grid points for evaluating the cdf (the `ngrid' argument) was set to be the number needed to get a space of 0.01 across the range of observed $\prederror$, with a lower bound of 512 grid points (the default of `ngrid' given by the \texttt{DeconCdf} function). Other options were set to the default of \texttt{DeconCdf}.

\subsection{Random draws from $\hat{\epsdist}$}

The end result of the \texttt{DeconCdf} function was values of the $\hat{\epsdist}$ evaluated at specific grid points $x$. The resulting $\hat{\epsdist}(x)$ was not always guaranteed to be $\leq 1$ nor monotonically decreasing; this is likely due to the fact that use of higher-order kernels estimates (which is standard practice in deconvolution literature) does not constrain the estimate be a proper density. Furthermore, the tail ends of the cdf are based on little data and unlikely to reliable, as well as having problems either non-monotonicity or extending beyond the boundaries of $(0,1)$. We truncated the left tail of $\hat{\epsdist}(x)$ to be within $0.001$ by finding the largest such $x_0$ such that $\hat{\epsdist}(x_0)\leq 0.001$ and setting $\hat{\epsdist}(x)=0.001$ for $x\leq x_0$; and we similarly trimmed the right tail based on $1-0.001$. We then calculated the differences $d_i=\hat{\epsdist}(x_i)-\hat{\epsdist}(x_{i-1})$ and for $d_i<0$ set $d_i=0$. We then defined a monotone cdf based on the cumulative sum of the $d_i,$
$$C(x_j)=\sum_{i=1}^j d_i.$$
We then renormalized the values $C(x_j)$ so that they extend from $0$ to $1$, giving the final monotone estimate of $\hat{G}(x_j)$ as
$$\hat{\epsdist}(x_j)=\frac{C(x_j)-\min_iC(x_i)}{\max_i C(x_i)-\min_i C(x_i)}$$
To randomly sample from $\hat{\epsdist}$, we needed to be able to evaluate $\hat{\epsdist}$ for all $x$. We did this by linearly interpolating between the $\hat{\epsdist}(x_j)$ values. In what follows, we consider the values $\hat{\epsdist}(x)$ based on this smoothed and monotone version of the original output of the \texttt{DeconCdf} function. 

We create random draws from $\hat{\epsdist}$ by drawing random variables $U_i$ from a $Unif(0,1)$ and calculating $E_i=\hat{\epsdist}(U_i)$. We further centered and standardized the draws $E_j$ from $\hat{\epsdist}$ to get  $$\eps^*_j=\left(E_j-mean_i(E_j)\right)\frac{\widehat{\sigma}_{\eps,LS}}{\sqrt{\var{E_i}}}$$ so that the resulting $\eps^*_j$ have mean zero and variance $\widehat{\sigma}_{\eps,LS}^2.$ This was done because the variance of $\hat{\epsdist}$ was not guaranteed to have the correct variance, dispite the fact we prespecify the variance in the deconvolution call. Ensuring the correct moments of $\eps^*_j$ was a critical component for reasonable coverage of the bootstrap confidence intervals. When we did not standardize the results and just took the draws from $E_j$, the resulting bootstrap confidence intervals became more and more conservative as $p/n$ grew. This again highlights the results of Theorem \ref{thm:asympPerfPNcloseto1} -- the variance of $\hat{\epsdist}$ is the most important feature of the distribution in order to have accurate confidence intervals.

\subsection{Bootstrap estimates $\betaHat^*$ from $\hat{\epsdist}$} \label{subsec:deconvbootmethods}
\label{supp:sec:Deconbootstrapdraws}
We used $\hat{\epsdist}$ to create bootstrap errors, $\{\eps^*_i\}_{i=1}^{n}$ in two ways. For the first method we estimated $\{\eps^*_i\}_{i=1}^{n}$ as a i.i.d draws from $\hat{\epsdist}$, and repeatedly drew such samples from $\hat{\epsdist}$, $B$ times. In the second method, we drew one single estimate $\{\hat{\eps}_i\}_{i=1}^{n}$ as i.i.d draws from $\hat{\epsdist}$ and then created $\{\eps^*_i\}_{i=1}^{n}$ from resampling from the empirical distribution of the $\{\hat{\eps}_i\}_{i=1}^{n}$, and repeated this resampling from the empirical distribution of $\{\hat{\eps}_i\}_{i=1}^{n}$ $B$ times. For both methods, we then calculated $\betaHat^{*}$ from the data $(X_i,y_i^*)$ where $y_i^*=X_i'\betaHat+\eps^{*}_i$, as in the standard residual bootstrap.
The first method seems to do slightly better in simulations, see Supplemental Figure \ref{supfig:bandwidth}.

\subsection{Estimation of $\lambda_i^2$}\label{supp:sec:estimationEllipParam}

To extend the deconvolution bootstrapping method to the elliptical case when $p/n\tendsto \kappa \in (0,1)$, one needs to be able to estimate $\lambda_i$, at least up to sign. In which case, one could estimate individually the variance of $\tilde{Z}_i$ and feed these individual estimates into the deconvolution method described above. 

We recall a simple proposal from the paper \citep{nekMarkoRiskPub2010} to solve this problem. Specifically, the author proposes to use
$$
\widehat{\lambda_i}^2=\frac{\norm{X_i^2}_2/p}{\frac{1}{p}\trace{\SigmaHat}}=\frac{\norm{X_i}_2^2}{\frac{1}{n}\sum_{i=1}^n \norm{X_i}_2^2}\;,
$$
where $\SigmaHat=\frac{1}{n}\sum_{i=1}^n X_i X_i\trsp$. Under mild conditions on $\Sigma$ and $\lambda_i$, it can be shown that when $n\tendsto \infty$ and $p/n\tendsto \kappa \in (0,\infty)$
$$
\sup_{1\leq i\leq n}|\lambda_i^2-\widehat{\lambda_i^2}|\tendsto 0 \text{ in probability}.  
$$
The intuition and proof are as follows. Concentration of measure arguments \citep{ledoux2001} show that $\norm{\xi_i}^2/p\simeq \trace{\Sigma}/p$ and hence $\norm{X_i}^2/p\simeq \lambda_i^2 \trace{\Sigma}/p$.  The law of large numbers and a little bit of further technical work then imply that $\frac{1}{n}\sum_{i=1}^n \norm{X_i}^2/p \simeq \Exp{\lambda_i^2} \trace{\Sigma}/p=\trace{\Sigma}/p$.

\color{black}

\section{Description of Simulations and other Numerics} \label{supp:Simulations}
In the simulations described in the paper, we explored variations in the distribution of the design matrix $X$, the error distribution, the loss function, the sample size ($n$), and the ratio of $\kappa=p/n$, detailed below. 

All results in the paper were based upon $1,000$ replications of our simulation routine for each combination of these values. Each simulation consisted of
\begin{enumerate}
	\item Simulation of data matrix $X$, $\{\epsilon_i\}_{i=1}^n$ and construction of data $y_i=X'\beta+\epsilon_i$. However, for our simulations, $\beta=0$ (without loss of generality for the results, which are shift equivariant), so $y_i=\epsilon_i$. 
	\item Estimate $\hat{\beta}$ using the corresponding loss function. For $L_2$ this was via the \texttt{lm} command in R, for Huber via the \texttt{rlm} command in the \texttt{MASS}  package  with default settings ($k=1.345$) \citep{MASS}, and for $L_1$ via an internal program making use of MOSEK optimization package and accessed in R using the \texttt{Rmosek} package \citep{Rmosek}. The internal $L_1$ program was checked to give the same  results as the \texttt{rq} function that is part of the R package \texttt{quantreg} \citep{quantreg}, but was much faster for simulations. 
	\item Bootstrapping according to the relevant bootstrap procedure (using the \texttt{boot} package) and estimating $\hat{\beta}^*$ for each bootstrap sample. Each bootstrap resampling consisted of $R=1,000$ bootstrap samples, the minimum generally suggested for 95\% confidence intervals \citep{DavisonHinkley97}. For jackknife resampling and for calculating leave-one-out prediction errors $\prederror$, we wrote an internal function that left out each observation in turn and recalculated $\hat{\beta}_{(i)}$. 
	\item Construction of confidence intervals for $\hat{\beta}_1$. For bootstrap resampling, we used the function \texttt{boot.ci} in the \texttt{boot} package to calculate confidence intervals. We calculated ``basic", ``percentile", ``normal", and ``BCA" confidence intervals (see help of \texttt{boot.ci} and \cite{DavisonHinkley97} for details about each of these), but all results shown in the manuscript rely on only the percentile method. The percentile method calculates the boundaries of the confidence intervals as the estimates of 2.5\% and 97.5\% percentiles of $\hat{\beta}^*_1$ (note that the estimate is not exactly the \emph{observed} 2.5\% and 97.5\% of $\hat{\beta}^*_1$, since there is a correction term for estimating the percentile, again see \cite{DavisonHinkley97}). For the jackknife confidence intervals, the confidence interval calculated was a standard normal confidence interval ($\pm 1.96 \sqrt{\widehat{var}_{Jack}(\hat{\beta_1})}$)
\end{enumerate}	

\subsection{Values of parameters}

\paragraph{Design Matrix} For the design matrix $X$, we considered the following designs for the distribution of an element $X_{ij}$ of the matrix $X$
\begin{itemize}
\item Normal: $X_{ij}$ are i.i.d $N(0,1)$
\item Double Exp: $X_{ij}$ are i.i.d. double exponential with variance $1$.
\item Elliptical: $X_{ij}\sim \lambda_i Z_{ij}$ where the $Z_{ij}$ are i.i.d $N(0,1)$ and the $\lambda_i$ are i.i.d according to 
\begin{itemize}
	\item $\lambda_i \sim Exp(\sqrt{2})$ (i.e. mean $1/\sqrt{2}$)
	\item $\lambda_i \sim N(0,1)$
	\item $\lambda_i \sim Unif(0.5,1.5)$
\end{itemize}
\end{itemize}

\paragraph{Error Distribution} We used two different distributions for the i.i.d errors $\epsilon_i$: $N(0,1)$ and standard double exponential (with variance $2$). 

\paragraph{Dimensions} We simulated from $n=100$, $500$, and $1,000$ though we showed only $n=500$ in our results for simplicity. Except where noted, no significant difference in the results was seen for varying sample size. The ratio $\kappa$ was simulated at $0.01, 0.1, 0.3, 0.5$.

\subsection{Correction factors for Jackknife}
We computed these quantities using the formula we mentioned in the text and Matlab. We solve the associated regression problems with $\texttt{cvx}$ \citep{cvx,cvxgb08}, running $\texttt{Mosek}$ \citep{mosek2015} as our optimization engine. We used $n=500$ and $1,000$ simulations to compute the mean of the quantities we were interested in. 

\subsection{Plotting of Figure \ref{subfig:relRisk:NoN}}
This figure was generated with Matlab, using cvx and Mosek, as described above. We picked $n=500$ and did 500 simulations. $p$ was taken in (5, 10, 30, 50, 75, 100, 125, 150, 175, 200, 225, 250, 275, 300, 350, 400, 450).
We used our simulations for the case of the original errors to estimate $\Exp{\norm{\betaHat-\beta}_2}$. We used this estimate in our simulation under the convolved error distribution. The Gaussian error simulations were made with ${\cal N}(0,2)$ to match the variance of the double exponential distribution.

\section{Technical background on results for robust regression} \label{supp:Reminders}
Recall that we consider 
$$
\betaHat_\rho=\argmin_{u\in\mathbb{R}^p} \sum_{i=1}^n \rho(y_i-X_i\trsp u)\;, \text{ where } y_i=\eps_i+X_i\trsp \beta\;.
$$
The $\eps_i$'s are assumed to be i.i.d with mean 0 here. 

\subsection{Classical results and asymptotic normality}

\paragraph{Least-squares} In this case $\rho(x)=x^2/2$ and we have of course 
$$
\betaHat_{LS}-\beta=(X\trsp X)^{-1}X\trsp \eps\;.
$$
Hence, 
$$
\scov{\betaHat_{LS}}=(X\trsp X)^{-1} \var{\eps}\;.
$$

\paragraph{Robust regression} We recall the classic result of Huber \citep{HuberRobustRegressionAsymptoticsETCAoS73} and \citep{HuberRonchettiRobustStatistics09}, Chapter 7: when $p$ is fixed and $n\tendsto \infty$, 
$$
\scov{\betaHat_\rho}=\frac{1}{n}\left(\frac{X\trsp X}{n}\right)^{-1} \frac{\Exp{\psi^2(\eps)}}{[\Exp{\psi'(\eps)}]^2}\;.
$$
See also the papers \cite{PortnoyMestLargishPNConsistencyAoS84,PortnoyMestLargishPNCLTAoS85,PortnoyCLTRobustRegressionJMVA87,MammenRobustRegressionAos89} for the situation where $p\tendsto\infty$ and $p/n\tendsto 0$ at various rates. 

\paragraph{Asymptotic normality questions and impact on confidence intervals: $\bm{p/n\tendsto 0}$} In the case of least-squares, the Lindeberg-Feller theorem \citep{StroockProbaBook} guarantees that under mild conditions on the $p\times n$ matrix $X$, the coordinates of $\betaHat_{LS}$ are asymptotically Normal. Similarly if the $1\times n$ vector $v\trsp (X\trsp X)^{-1}X\trsp $ satisfy the conditions of the the Lindeberg-Feller theorem, then $v\trsp (\betaHat_{LS}-\beta)$ is asymptotically normal. Similarly, under mild conditions on $X$, the classic papers mentioned above guarantee asymptotic normality of the coordinates of $\betaHat_\rho$ when $p/n\tendsto 0$. In these cases, the width of confidence intervals for the coordinates of $\beta$ are hence only dependent asymptotically on the variance of the coordinates of $\betaHat_\rho$.

\subsection{Summary of recent results on high-dimensional robust regression}
We summarize in this section the key results we use from the recent papers \cite{NEKetAlRobustRegressionTechReport11,NEKRobustPaperPNAS2013Published,NEKRobustRegressionRigorous2013}. The third paper is a completely rigorous version of the heuristic arguments of the first two; the first paper is the long-form version of the second one. Those papers are concerned with the asymptotic properties of $\betaHat_\rho$ when $p/n\tendsto\kappa\in (0,1)$. The predictor vectors $X_i$'s are assumed to be random and independent, with $X_i=\lambda_i \Sigma^{1/2}\tilde{X}_i$, where $\tilde{X}_i$ has i.i.d (not necessarily Gaussian) entries with mean 0 and variance 1. $\lambda_i$'s are independent random variables with $\Exp{\lambda_i^2}=1$. (The $n\times p$ design matrix $X$ is full rank with probability 1. $\Sigma$ has only positive eigenvalues.) $X_i$'s are independent of $\eps_i$'s. 

\paragraph{Role of $\scov{X_i}=\Sigma$} It is shown that in these papers, if $\betaHat(\beta;\Sigma)$ is the regression vector corresponding to the situation where $y_i=X_i\trsp \beta+\eps_i$ and $\scov{X_i}=\Sigma$ for all $i$,
$$
\betaHat_\rho(\beta;\Sigma)=\beta+\Sigma^{-1/2} \betaHat(0;\id_p)\;.
$$ 
This follows from a simple change of variable. It also means that to understand the properties of $\betaHat_\rho(\beta;\Sigma)$, it is enough to understand the ``null case'' $\beta=0$ and $\Sigma=\id_p$. 

\paragraph{Consequence for leave-one-out-predicted errors} The result we just mentioned has an important consequence for our leave-one-out predicted error, i.e $\prederror=y_i-X_i\trsp \betaHat_{(i)}$: $\prederror(\beta;\Sigma)=\prederror(0;\id_p)$. In other words, we can assume without loss of generality that $\beta=0$ and $\Sigma=\id_p$ when working with leave-one-out-predicted errors. 

\newcommand{\myGaussian}{\upsilon}
\paragraph{A non-asymptotic and exact stochastic representation in the elliptical case} When $X_i\iid \lambda_i \myGaussian_i$, where $\myGaussian_i\sim{\cal N}(0,\Sigma)$ and $\lambda_i$ is a random variable independent of $\myGaussian_i$, it is shown that 
$$
\betaHat_\rho(\beta;\Sigma)\equalInLaw \beta + \norm{\betaHat_\rho(0;\id_p)}_2 \Sigma^{-1/2} u\;,
$$
where $u$ is uniformly distributed on the unit sphere in $\mathbb{R}^p$ and $\norm{\betaHat_\rho(0;\id_p)}_2$ is independent of $u$. $\norm{\betaHat_\rho(0;\id_p)}_2$ is simply the norm of $\betaHat_\rho$ when $\beta=0$ and $\scov{X_i}=\id_p$. Note that $u$ has the stochastic representation $u\equalInLaw Z_p/\norm{Z_p}_2$, where $Z_p\sim {\cal N}(0,\id_p)$. 

\paragraph{Consequence of the previous representation for large $p$} Since $\norm{Z_p}_2$ has $\chi_p$ distribution, it is clear that as $p\tendsto \infty$, if  $v$ is a deterministic vector, 
$$
\sqrt{p}\frac{v\trsp (\betaHat_\rho(\beta;\Sigma)-\beta)}{\norm{\betaHat_\rho(0;\id_p)}_2} \weakCv {\cal N}(0,v\trsp \Sigma^{-1} v)\;,
$$
where $\weakCv$ denotes weak convergence of distributions. Hence, provided $\norm{\betaHat_\rho(0;\id_p)}_2$ and $v\trsp \Sigma^{-1}v$ remain bounded, $v\trsp \betaHat_\rho(\beta;\Sigma)$ is $\sqrt{p}$-consistent for $v\trsp\beta$. 

\paragraph{Properties of $\norm{\betaHat_\rho(0;\id_p)}_2$} It is shown, under various technical assumptions, that as $p$ and $n$ tend to infinity with $p/n\tendsto \kappa$, the variance of the random variable $\norm{\betaHat_\rho(0;\id_p)}_2$ goes to zero. Hence, for practical matters, $\norm{\betaHat_\rho(0;\id_p)}_2$ can be considered non-random. In particular, that implies that 
$$
\sqrt{p}v\trsp (\betaHat_\rho(\beta;\Sigma)-\beta) \text{ is approximately Normal as }  p/n\tendsto \kappa\;.
$$
Of great importance is the characterization of $\norm{\betaHat_\rho(0;\id_p)}_2$, since it will affect the width of confidence intervals. It can be characterized, in the case where $\lambda_i=1$ (see the papers for the case $\lambda_i\neq 1$) in the following way:
$\norm{\betaHat_\rho(0;\id_p)}_2\tendsto r_\rho(\kappa)$. The non-random scalar $r_\rho(\kappa)$  can be characterized through a system of two non-linear equations, involving another constant, $c$. The pair of positive and deterministic scalars $(c,r_\rho(\kappa))$ satisfy: if $\hat{z}_\eps=\eps+r_\rho(\kappa)Z$, where $Z\sim {\cal N}(0,1)$ is independent of $\eps$, and $\eps$ has the same distribution as $\eps_i$'s: 
$$
\left\{
\begin{array}{rl}
\Exp{(\prox(c\rho))'(\hat{z}_\eps)}&=1-\kappa\;,\\
\kappa r_\rho^2(\kappa)&=\Exp{[\hat{z}_\eps-\prox(c\rho)(\hat{z}_\eps)]^2}\;.
\end{array}
\right.
$$
In this system, $\prox(c\rho)$ refers to Moreau's proximal mapping of the convex function $c\rho$ - see \cite{MoreauProxPaper65} or \cite{HiriartLemarechalConvexAnalysisAbridged2001}. (The system is rigorously shown in \cite{NEKRobustRegressionRigorous2013} under the assumption that the $X_i$'s have i.i.d entries with mean 0 and variance 1, as well as a few other minor requirements; these assumptions are satisfied when $X_{i,j}$ have a Gaussian distribution, or are bounded, or do not have heavy tails, the latter requiring appeal to various truncation arguments. Another proof of the validity of this system, which first appeared in \cite{NEKetAlRobustRegressionTechReport11}, can be found in \cite{DonohoMontanariRobustArxiv13}. That proof is limited to the case of $X_i$'s having i.i.d Gaussian entries.) The assumptions on $\eps_i$'s and $\rho$ are relatively mild. See \cite{NEKRegressionMEstimateRigorousElliptical2015} for the latest, handling the situation where $\eps_i$'s have for instance a Cauchy distribution. We note that some of the results in \cite{NEKRobustRegressionRigorous2013} are stated with $\rho$ strongly convex (and $\eps_i$'s having many moments). While the proof in that paper suggests several ways of removing this assumption, it is also possible to change $\rho$ in to $\rho+\eta x^2/2$ with $\eta$ very small (e.g $\eta=10^{-100}$) to satisfy this technical assumption and change essentially nothing to the statistical problem at hand. 

\paragraph{Consequences for the distribution of $\betaHat_1$ or other contrasts of interest} In our simulation setup, the previous results imply that the distribution of $\betaHat_1$ (or any other coordinates or contrasts $v\trsp \betaHat$ for $v$ deterministic) is asymptotically normal. In the case where $\Sigma=\id_p$, the variance of $\sqrt{p}(\betaHat_1-\beta_1)$ is roughly ${\cal N}(0,r_\rho^2(\kappa))$.  See \cite{NEKOptimalMEstimationPNASPublished2013} and its supplementary material for a longer discussion and questions related to building confidence intervals.

\paragraph{Asymptotic normality questions and impact on confidence intervals: $\bm{p/n\tendsto \kappa \in (0,1)}$} Because we know that, in the Gaussian design case, the coordinates of $\betaHat_\rho$ are asymptotically normal, the width of these intervals is completely determined by the variance of the coordinates of $\betaHat_\rho$. We explain above how these variances depend on the distribution of $\eps$ and the loss function $\rho$: basically through $\norm{\betaHat(\rho;\id)}_2$ and hence $r_\rho(\kappa)$. Therefore, as was the case in the low-dimensional situation, the variance of the coordinates of $\betaHat_\rho$ can be used as a proxy for the width of the confidence interval in the high-dimensional case where $p/n\tendsto \kappa$, $0<\kappa<1$.

In \citep{NEKOptimalMEstimationPNASPublished2013}, these asymptotic normality results are used to create confidence intervals for $v\trsp\beta$ in the Gaussian design case: if $z_{1-\alpha/2}$ is the $(1-\alpha/2)$-quantile of the Gaussian distribution a $100(1-\alpha)\%$ confidence interval for $v\trsp\beta$ is 
$$
v\trsp \betaHat\pm \frac{z_{1-\alpha/2}}{\sqrt{p}} \hat{r} \sqrt{(1-p/n)v\trsp \SigmaHat^{-1} v}\;,
$$ 
where $\hat{r}$ is a consistent estimator of $\norm{\betaHat_\rho(0;\id_p)}_2$. In \citep{NEKOptimalMEstimationPNASPublished2013}, it is said without more precision that leave-one-techniques can be used to come up with $\hat{r}$; we propose in the current paper estimates $\hat{r}$ based on leave-one-out predicted errors that can therefore be used for the purpose of building those confidence intervals. (See Section \ref{subsec:LeaveOut} in the main paper)

\paragraph{Leave-one-out approximations for $\betaHat$} It is shown in the aforementioned papers that 
$$
\betaHat\simeq\betaHat_{(i)}+\frac{1}{n}S_i^{-1}X_i \psi(\resid_i)\;,
$$
where $\simeq$ means that we are neglecting a quantity that is negligible for all our mathematical and statistical purposes  (see the papers for very precise bounds on the quantity we are neglecting). This approximation is the key to the approximations in Equations \eqref{eq:residual} and \eqref{eq:residual2ndPart} which we use in the main paper.  Recall that $S_i=\frac{1}{n}\sum_{j\neq i} \psi'(\prederrorAtj) X_j X_j\trsp$.

\subsection{Consequences for the residual bootstrap} \label{supp:subsec:csqResidualBoot}
We call $\{\eps_i^*\}_{i=1}^n$ the estimated errors used in the residual bootstrap. When doing a residual bootstrap, we are effectively sampling from a model with fixed design $X$, ``true $\beta$'' taken to be equal to $\betaHat_\rho$ and i.i.d errors sampled according to the empirical distribution of the $\{\eps_i^*\}_{i=1}^n$. As a shortcut, we call this distribution $\eps^*$ in what follows. We call $\betaHat_\rho^*$ the bootstrapped version of $\betaHat$.

\paragraph{Case $\bm{p/n\tendsto 0}$} Naturally, the classic results mentioned above imply that the distribution of $v\trsp (\betaHat_\rho^*-\betaHat_\rho)$ is going to be asymptotically normal (under mild conditions on $X$ that are satisfied in our simulations); the  variance of the coordinates of $\betaHat_\rho^*$, on the other hand depends on $\frac{\Exp{\psi^2(\eps^*)}}{[\Exp{\psi'(\eps^*)}]^2}$. Hence, even if the distribution of the estimated errors $\eps^*$ is very different from that of the ``true'' errors, $\eps$, the residual bootstrap may work very well: indeed, if $\eps$ and $\eps^*$ have two very different distribution but 
$$
\frac{\Exp{\psi^2(\eps^*)}}{[\Exp{\psi'(\eps^*)}]^2}=\frac{\Exp{\psi^2(\eps)}}{[\Exp{\psi'(\eps)}]^2}\;,
$$
using a residual bootstrap with ``the wrong error distribution'', $\eps^*$, will give us bootstrap confidence intervals of the right width.  An important question then becomes, when $p/n$ is small: what class of distributions $\eps^*$ is such that $\frac{\Exp{\psi^2(\eps^*)}}{[\Exp{\psi'(\eps^*)}]^2}=\frac{\Exp{\psi^2(\eps)}}{[\Exp{\psi'(\eps)}]^2}$, as this class defines all acceptable error distributions from the point of view of our residual bootstrap. 

\paragraph{Case $\bm{p/n\tendsto \kappa \in (0,1)}$} We note that at this point in the case $p/n\tendsto \kappa \in (0,1)$ we are not aware of central limit theorems for the coordinates of $\betaHat$ that are valid conditional on the design matrix $X$. However, it is expected that such theorems will hold if the design matrix results from a draw of a random design matrix similar to the ones we consider (with very high-probability with respect to the sampling of the design matrix). The discussions above make then clear that the key quantity to describe the width of the residual bootstrap confidence intervals becomes the risk $\norm{\betaHat_\rho(0;\id_p;\eps^*)}_2$, i.e the risk $\norm{\betaHat_\rho(0;\id_p)}_2$ when the error distribution is $\eps^*$. A ``good'' error distribution is therefore one for which $r_\rho(\kappa;\eps^*)\simeq r_\rho(\kappa;\eps)$. (We used the notation $r_\rho(\kappa;\eps)=\lim_{n\tendsto \infty} \norm{\betaHat_\rho(0;\id_p;\eps)}$, when $p/n\tendsto \kappa$.)

\paragraph{The case of least squares} Let us call $\hat{\epsdist}_{n,p}$ the distribution of the errors we use in our residual bootstrap. We assume that $\hat{\epsdist}_{n,p}$ has mean 0. Let us call $w\trsp = v\trsp (X\trsp X)^{-1} X\trsp$ - where we choose to not index $v$ and $w$ by $p$ for the sake of clarity. $v$ is a deterministic sequence of $p$-dimensional vectors.  Assume that $w$ and $\hat{\epsdist}_{n,p}$ satisfy the conditions of the Linderberg-Feller theorem for triangular arrays, and that $\lim_{n\tendsto \infty} \var{\hat{\epsdist}_{n,p}}=\sigma^2_\eps$. Then the Lindeberg-Feller theorem guarantees that  
$$
\frac{v\trsp (\betaHat^*-\betaHat)}{\norm{w}}\weakCv {\cal N}(0,\sigma^2_\eps)\;.
$$
Note that it also guarantees, under the same assumptions on $w$ that 
$$
\frac{v\trsp (\betaHat-\beta)}{\norm{w}}\weakCv {\cal N}(0,\sigma^2_\eps)\;.
$$
These results do not depend on the size of $\kappa$, the limit of the ratio $p/n$.

Informally, what this means is that provided that the entries of $w$ are all relatively small, that $\hat{\epsdist}_{n,p}$ has mean 0 and $\var{\hat{\epsdist}_{n,p}}$ is close to $\sigma^2_\eps$, then bootstrapping from the residuals in least-squares works for approximating the distribution $v\trsp (\betaHat-\beta)$. 

\paragraph{Conclusion for the purposes of the main paper} In our discussions we use $\norm{\betaHat_\rho(0;\id_p;\eps^*)}$ and its closeness to its value under the correct error distribution, $\norm{\betaHat_\rho(0;\id_p;\eps)}$, as a proxy to understand a priori the quality of residual bootstrap confidence intervals when using $\eps^*$ to sample the errors instead of $\eps$. The previous discussion explains why we do so. Our numerical work in Section \ref{subsec:LeaveOut} of the main text shows numerically that this yields valuable insights. This is why our discussion in Section \ref{subsec:PerfMethod2} is focused on understanding $\norm{\betaHat(0;\id_p;\eps)}_2$ for various error distributions. In particular, Theorem \ref{thm:asympPerfPNcloseto1} shows that when $p/n$ is close to 1, if $\eps^*$ has approximately the same two first moments as $\eps$, $\norm{\betaHat(0;\id_p;\eps^*)}/\norm{\betaHat(0;\id_p;\eps)}\simeq 1$. This explains why the scaled $\stprederror$ is probably a good error distribution $\eps^*$ to use in the residual bootstrap when $\kappa$ is close to 0 or 1. We note that when $\kappa$ is close to 1, $\stprederror$ gives an error distribution that is in general very different from the distribution of $\eps$. Our numerical work of Section \ref{subsec:LeaveOut} shows that it is nonetheless a good error distribution from the point of view of the residual bootstraps we consider.

\section{Residual bootstrap ($p/n$ close to 1)}\label{supp:ResidProof}
We analyze the problem when $p/n$ is close to 1 and prove Theorem \ref{thm:asympPerfPNcloseto1}.

\begin{proof}[Proof of Theorem \ref{thm:asympPerfPNcloseto1}]
	
Recall the system describing the asymptotic limit of $\norm{\betaHat_\rho-\beta}$ when $p/n\tendsto \kappa$ and the design matrix has i.i.d mean 0, variance 1 entries, is, under some conditions on $\eps_i$'s and some mild further conditions on the design (see Section \ref{supp:Reminders} above): $\norm{\betaHat_\rho-\beta}\tendsto r_\rho(\kappa)$ and the pair of positive and deterministic scalars $(c,r_\rho(\kappa))$ satisfy: if $\hat{z}_\eps=\eps+r_\rho(\kappa)Z$, where $Z\sim {\cal N}(0,1)$ is independent of $\eps$, and $\eps$ has the same distribution as $\eps_i$'s: 
$$
\left\{
\begin{array}{rl}
\Exp{(\prox(c\rho))'(\hat{z}_\eps)}&=1-\kappa\;,\\
\kappa r_\rho^2(\kappa)&=\Exp{[\hat{z}_\eps-\prox(c\rho)(\hat{z}_\eps)]^2}\;.
\end{array}
\right.
$$
In this system, $\prox(c\rho)$ refers to Moreau's proximal mapping of the convex function $c\rho$ - see \cite{MoreauProxPaper65} or \cite{HiriartLemarechalConvexAnalysisAbridged2001}.
	
We first give an informal argument to ``guess'' the correct values of various quantities of interest, namely $c$ and of course, $r_\rho(\kappa)$. 

Note that when $|x|\ll c$, and when $\psi(x)\sim x$ at 0, $\text{prox}(c\rho)(x)\simeq \frac{x}{1+c}$. Hence, $x-\prox(c\rho)(x)\simeq x c/(1+c)$. (Note that as long as $\psi(x)$ is linear near 0, we can assume that $\psi(x)\sim x$, since the scaling of $\rho$ by a constant does not affect the performance of the estimators.)

We see that $1-\kappa\simeq 1/(1+c)$, so that $c\simeq \kappa/(1-\kappa)$ - assuming for a moment that we can apply the previous approximations in the system . Hence, we have 
$$
\kappa r_\rho(\kappa)^2 \simeq (c/(1+c))^2 [r_\rho(\kappa)^2+\sigma^2_\eps]\simeq \kappa^2 [r_\rho(\kappa)^2+\sigma^2_\eps]\;.
$$
We can therefore conclude (informally at this point) that 
$$
r_\rho(\kappa)^2 \sim \frac{\sigma^2_\eps \kappa}{1-\kappa}\sim \frac{\sigma^2_\eps}{1-\kappa}\;.
$$

Once these values are guessed, it is easy to verify that $r_\rho(\kappa)\ll c$ and hence all the manipulations above are valid if we plug these two expressions in the system driving the performance of robust regression estimators described above. We note that our argument is not circular: we just described a way to guess the correct result. Once this has been done, we have to make a verification argument to show that our guess was correct. 

In this particular case, the verification is done as follows: we can rewrite the expectations as integrals and split the domain of integration into $(-\infty,-s_\kappa)$, $(-s_\kappa,s_\kappa)$, $(s_\kappa,\infty)$, with $s_\kappa=(1-\kappa)^{-3/4}$. Using our candidate values for $c$ and $r_\rho(\kappa)$, we see that the corresponding $\zHat_\eps$ has extremely low probability of falling outside the interval $(-s_\kappa,s_\kappa)$ - recall that $1-\kappa\tendsto 0$. Coarse bounding of the integrands outside this interval shows the corresponding contributions to the expectations are negligible at the scales we consider. On the interval $(-s_\kappa,s_\kappa)$, we can on the other hand make the approximations for $\prox(c\rho)(x)$ we discussed above and integrate them. That gives us the verification argument we need, after somewhat tedious but simple technical arguments.  (Note that the method of propagation of errors in analysis described in \citep{MillerAppliedAsymptoticAnalysis06} works essentially in a similar a-posteriori-verification fashion. Also, $s_\kappa$ could be picked as $(1-\kappa)^{-(1/2+\delta)}$ for any $\delta\in (0,1/2)$ and the arguments would still go through.)
\end{proof}

\section{On the expected Variance of the bootstrap estimator (Proof of Theorem \ref{thm:bootVariance})} \label{supp:BootVarProof}

In this section, we compute the expected variance of the bootstrap estimator.

We recall that for random variables $T,\Gamma$, we have 
$$
\var{T}=\var{\Exp{T|\Gamma}}+\Exp{\var{T|\Gamma}}\;.
$$
In our case, $T=v\trsp \betaHat_w$, the projection of the regression estimator $\betaHat_w$ obtained using the random weights $w$ on the contrast vector $v$. $\Gamma$ represents both the design matrix and the errors. We assume without loss of generality that $\norm{v}_2=1$. 

Hence, 
$$
\var{v\trsp\betaHat_w}=\var{v\trsp\Exp{\betaHat_w|\Gamma}}+\Exp{\var{v\trsp \betaHat_w|\Gamma}}\;.
$$

In plain English, the variance of $v\trsp \betaHat_w$ is equal to the variance of the bagged estimator plus the expectation of the variance of the bootstrap estimator (where we randomly weight observation $(y_i,X_i)$ with weight $w_i$).

As explained in Section \ref{sec:supp:CovAndInvarianceIssue}, we can study without loss of generality the case where $\Sigma=\id_p$ and $\beta=0$. This is what we do in this proof. Further the rotational invariance arguments we give in Section \ref{sec:supp:CovAndInvarianceIssue} mean that we can focus on the case $v=e_p$,the $p$-th canonical basis vector, without loss of generality. 

We consider the case where $X_i\iid {\cal N}(0,\id_p)$. This allows us to work with results in \cite{NEKetAlRobustRegressionTechReport11,NEKRobustPaperPNAS2013Published}, \cite{NEKRobustRegressionRigorous2013}. 

\paragraph{Notational simplification} To make the notation lighter, in what follows in this proof we use the notation $\betaHat$ for $\betaHat_w$. There are no ambiguities that we are always using a weighted version of the estimator and hence this simplification should not create any confusion. 

In particular, we have, using the derivation of Equation (9) in \cite{NEKRobustPaperPNAS2013Published} and noting that in the least-squares case all approximations in that paper are actually exact equalities,
$$
\betaHat_p=\hat{c}\frac{\sum_{i=1}^n w_i X_i(p) \resid_{i,[p]}}{p}\;.
$$
$\resid_{i,[p]}$ here are the residuals based on the first $p-1$ predictors, when $\beta=0$. We note that, under our assumptions on $X_i$'s and $w_i$'s, $\hat{c}= \frac{1}{n}\trace{S_w^{-1}}+\lo_{L_2}(1)$, where $S_w=\frac{1}{n}\sum_{i=1}^n w_i X_i X_i\trsp$. It is known from work in random matrix theory (see e.g \cite{nekCorrEllipD}) that $\frac{1}{n}\trace{S_w^{-1}}$ is asymptotically deterministic in the situation under investigation with our assumptions on $w$ and $X$, i.e $\frac{1}{n}\trace{S_w^{-1}}=c+\lo_{L_2}(1)$, where $c=\Exp{\frac{1}{n}\trace{S_w^{-1}}}$. 

We also recall the residuals representation from \cite{NEKRobustPaperPNAS2013Published}, which are exact in the case of least-squares~: namely here,
$$
\betaHat-\betaHat_{(i)}=\frac{w_i}{n} S_i^{-1}X_i \psi(\resid_i)\;,
$$
which implies that, with $S_i=\frac{1}{n}\sum_{j\neq i}w_j X_j X_j\trsp$,  
$$
\prederror=\resid_i+w_i\frac{X_i\trsp S_i^{-1}X_i}{n}\psi(\resid_i)\;.
$$
In the case of least-squares, $\psi(x)=x$, so that 
$$
\resid_i=\frac{\prederror}{1+w_i c_i}\;,
$$
where 
$$
c_i=\frac{X_i\trsp S_i^{-1}X_i}{n}\;.
$$
These equalities also follow from simple linear algebra since we are in the least-squares case. We note that $c_i=c+\lo_P(1)$, where $c$ is deterministic, as explained in e.g \cite{nekMarkoRiskPub2010}, \cite{NEKRobustRegressionRigorous2013}. Furthermore, here the approximation holds in $L_2$ because of our assumptions on $w$'s and existence of moments for the inverse Wishart distribution - see e.g \cite{Haff79IdentityWishartDWithApps}. As explained in \cite{NEKRobustRegressionRigorous2013}, the same is true for $c_{i,[p]}$ which is the same quantity computed using the first $(p-1)$ coordinates of $X_i$, vectors we denote generically by $V_i$. We can rewrite 
$$
\betaHat_p=\hat{c} \frac{\sum_{i=1}^n w_i X_i(p) \frac{\prederrorOnePredictorOut}{1+w_i c_{i,[p]}}}{p}\;.
$$
Let us call $\bHat$ the bagged estimate. We note that $\prederrorOnePredictorOut$ is independent of $w_i$ and so is $c_{i,[p]}$. We have already seen that $\hat{c}$ is close to a constant, $c$. So taking expectation with respect to the weights,  we have, if $w_{(i)}$ denotes $\{w_j\}_{j\neq i}$, and using independence of the weights,  
$$
\bHat_p=\frac{1}{p}\sum_{i=1}^n \myExp_{w_i}\left(\frac{cw_i}{1+cw_i}\right) X_i(p) \myExp_{w_{(i)}}\left(\prederrorOnePredictorOut\right) [1+\lo_{L_2}(1)]\;.
$$
Now the last term is of course the prediction error for the bagged problem, i.e 
$$
\myExp_{w_{(i)}}\left(\prederrorOnePredictorOut\right)=\eps_i-V_i\trsp(\gHat_{(i)}-\gamma)\;
$$
\renewcommand{\gammaHat}{\widehat{\gamma}}
where $\gHat_{(i)}$ is the bagged estimate of $\gammaHat$ and $\gammaHat$ is the regression vector obtained by regressing $y_i$ on the first $p-1$ coordinates of $X_i$. (Recall that in these theoretical considerations we are assuming that $\beta=0$, without loss of generality.)  

So we have, since we can work in the null case where $\gamma=0$ (without loss of generality),
$$
\bHat_p= \frac{1}{p}\sum_{i=1}^n \myExp_{w_i}\left(\frac{cw_i}{1+cw_i}\right) X_i(p) \left[\eps_i-V_i\trsp\gHat_{(i)}\right](1+\lo_{L_2}(1))\;.
$$
Hence, 
$$
\Exp{p\bHat_p^2}= \frac{1}{p} \sum_{i=1}^n \left[\myExp_{w_i}\left(\frac{cw_i}{1+cw_i}\right)\right]^2 (\sigma^2_\eps+\Exp{\norm{\gHat_{(i)}}_2^2})(1+\lo(1))\;.
$$
Now, in expectation, using e.g \cite{NEKRobustRegressionRigorous2013},  $\Exp{\norm{\gHat_{(i)}}_2^2}(1+\lo(1))= \Exp{\norm{\bHat}^2_2}=p\Exp{\bHat_p^2}$. The last equality comes from the fact that all coordinates play a symmetric role in this problem, so they are all equal in law.

Now, recall that according to e.g \cite{NEKRobustPaperPNAS2013Published}, top-right equation on p. 14562, or \cite{nekMarkoRiskPub2010}
$$
\frac{1}{n}\sum_{i=1}^n \frac{1}{1+c w_i}=1-\frac{p}{n}+\lo_{L_2}(1)\;,
$$
since the previous expression effectively relates $\trace{D_w X (X\trsp D_w X)^{-1}X\trsp}$ to $n-p$, the rank of the corresponding ``hat matrix''.

Since $\frac{cw_i}{1+cw_i}=1-\frac{1}{1+cw_i}$, we see that 
$$
\myExp_{w_i}\left(\frac{cw_i}{1+cw_i}\right)=\frac{p}{n}+\lo(1)\;.
$$

Hence, for the bagged estimate, we have the equation 
$$
\Exp{\norm{\bHat}_2^2}=\frac{p}{n}\left(\sigma^2+\Exp{\norm{\bHat}_2^2}\right)(1+\lo(1))\;.
$$
We conclude that 
$$
\Exp{\norm{\bHat}_2^2}= (1+\lo(1))\frac{\kappa}{1-\kappa}\sigma^2\;.
$$
Note that $\frac{\kappa}{1-\kappa}\sigma^2=\Exp{\norm{\betaHat_{sLS}}_2^2}$, where the latter is the standard (i.e non-weighted) least squares estimator. 

We note that the rotational invariance argument given in \cite{NEKetAlRobustRegressionTechReport11,NEKRobustPaperPNAS2013Published} still apply here, so that we have the
$$
\bHat-\beta\equalInLaw \norm{\bHat-\beta}u\;,
$$
where $u$ is uniform on the sphere and independent of $\norm{\bHat-\beta}$ (recall that this simply comes from the fact that if $X_i$ is changed into $OX_i$, where $O$ is orthogonal, $\bHat$ is changed into $O\bHat$ - and we then apply invariance arguments coming from rotational invariance of the distribution of $X_i$). Therefore, 
$$
\var{v\trsp (\bHat-\beta)}=\frac{\norm{v}^2}{p}\Exp{\norm{\bHat-\beta}_2^2}\;.
$$

So we conclude that 
$$
p\Exp{\var{v\trsp \betaHat_w|\Gamma}}=p\var{v\trsp \betaHat_w}-\frac{\kappa}{1-\kappa}\sigma^2\norm{v}_2^2+\lo(1)\;.
$$

Now, the quantity $\var{v\trsp \betaHat_w}$ is well understood. The rotational invariance arguments we mentioned before give that 
$$
\var{v\trsp \betaHat_w}=\frac{\norm{v}_2^2}{p}\Exp{\norm{\betaHat_w-\beta}_2^2}\;.
$$
In fact, using the notation $D_w$ for the diagonal matrix with $D_w(i,i)=w_i$, since 
$$
\betaHat_w-\beta=(X\trsp D_w X)^{-1}X\trsp D_w \eps\;,
$$
we see that 
$$
\Exp{\norm{\betaHat_w-\beta}_2^2}=\sigma^2_\eps\Exp{\trace{(X\trsp D_w X)^{-2}X\trsp D_{w^2} X}}\;.
$$
(Note that under mild conditions on $\eps$, $X$ and $w$, we also have $\norm{\betaHat_w-\beta}^2_2=\Exp{\norm{\betaHat_w-\beta}_2^2}+\lo_{L_2}(1)$ - owing to concentration results for quadratic forms of vectors with independent entries; see \cite{ledoux2001}.)

We now need to simplify this quantity.

\textbf{Analytical simplification of $\bm{\trace{(X\trsp D_w X)^{-2}X\trsp D_{w^2} X}}$} 
Of course, 
$$
\trace{(X\trsp D_w X)^{-2}X\trsp D_{w^2} X}=\trace{D_{w} X (X\trsp D_w X)^{-2}X\trsp D_{w} }=\sum_{i=1}^n w_i^2 X_i\trsp (X\trsp D_w X)^{-2}X_i\;.
$$
Hence, if $\SigmaHat_w=\frac{1}{n}\sum_{i=1}^n w_i X_i X_i\trsp\triangleq \frac{w_i}{n}X_i X_i\trsp +\SigmaHat_{(i)}$, we have
$$
\trace{(X\trsp D_w X)^{-2}X\trsp D_{w^2} X}=\frac{1}{n}\sum_{i=1}^n w_i^2 \frac{X_i\trsp \SigmaHat^{-2} X_i}{n}\;.
$$

Call $\SigmaHat(z)=\SigmaHat-z\id_p$. Using the identity 
$$
(\SigmaHat-z\id_p)(\SigmaHat-z\id_p)^{-1}=\id_p\;,
$$
we see, after taking traces, that (\cite{silverstein95}) 
$$
\frac{1}{n}\sum_{i=1}^n w_i X_i\trsp (\SigmaHat-z\id_p)^{-1} X_i-z \trace{(\SigmaHat-z\id_p)^{-1}}=p\;.
$$
We call, for $z\in \mathbb{C}$,  $c(z)=\frac{1}{n}\trace{(\SigmaHat-z\id_p)^{-1}}$ and $c_i(z)=X_i\trsp (\SigmaHat_{(i)}-z\id_p)^{-1}X_i$, provided $z$ is not an eigenvalue of $\SigmaHat$. 

Differentiating with respect to $z$ and taking $z=0$ (we know here that $\SigmaHat$ is non-singular with probability 1, so this does not create a problem), we have 
$$
\frac{1}{n}\sum_{i=1}^n w_i X_i\trsp\SigmaHat^{-2} X_i-\trace{\SigmaHat^{-1}}=0\;.
$$
Also, since, by the Sherman-Morrison-Woodbury formula (\cite{hj}), 
$$
X_i\trsp \SigmaHat(z)^{-1}X_i=\frac{X_i\trsp \SigmaHat_{(i)}(z)^{-1}X_i}{1+w_i\frac{1}{n}X_i\trsp \SigmaHat_{(i)}(z)^{-1}X_i}\;,
$$
we have, after differentiating,
$$
\frac{1}{n}X_i\trsp \SigmaHat^{-2}X_i=\frac{c_i'(0)}{[1+w_ic_i(0)]^2}\;,
$$
where of course $c_i'(0)=X_i\trsp \SigmaHat_{(i)}^{-2}X_i$. Hence, 
$$
\frac{1}{n}\sum_{i=1}^n w_i^2\frac{1}{n}X_i\trsp \SigmaHat^{-2}X_i=\frac{1}{n}\sum_{i=1}^n w_i^2 \frac{c_i'(0)}{[1+w_i c_i(0)]^2}= c'(0) \frac{1}{n}\sum_{i=1}^n \frac{w_i^2}{[1+w_i c(0)]^2}\;.
$$
(Note that the arguments given in e.g \cite{nekMarkoRiskPub2010} or \cite{NEKHolgerShrinkage11} for why $c_i(z)=c(z)(1+\lo_P(1))$ extend easily to $c_i'$ and $c'$ given our assumptions on $w$'s and the fact that these functions have simple interpretations in terms of traces of powers of inverses of certain well-behaved - under our assumptions - matrices.) 

Going back to 
$$
\frac{1}{n}\sum_{i=1}^n w_i X_i\trsp (\SigmaHat-z\id_p)^{-1} X_i-z \trace{(\SigmaHat-z\id_p)^{-1}}=p\;,
$$
and using the previously discussed identity 
$$
\frac{w_i}{n}X_i\trsp (\SigmaHat-z\id_p)^{-1} X_i = 1-\frac{1}{1+w_i c_i(z)}\;,
$$
we have
$$
n-\sum_{i=1}^n \frac{1}{1+w_i c_i(z)}-zn c(z)=p\;.
$$
In other words, 
$$
1-\kappa=\frac{1}{n}\sum_{i=1}^n \frac{1}{1+w_i c_i(z)}+z c(z)\;.
$$
Now, 
\begin{align*}
c(z)\frac{1}{n}\sum_{i=1}^n \frac{w_i}{1+w_i c(z)}&=\frac{1}{n}\sum_{i=1}^n (1-\frac{1}{1+w_i c(z)})\\
&=\frac{1}{n}\sum_{i=1}^n (1-\frac{1}{1+w_i c_i(z)})+\eta(z)\\
&=\kappa+zc(z)+\eta(z)\;,
\end{align*}
where $\eta(z)$ is such that $\eta(z)=\lo_P(1)$ and $\eta'(z)=\lo_P(1)$ ($\eta$ has an explicit expression which allows us to verify these claims). Therefore, by differentiation, and after simplifications, 
$$
\frac{1}{n}\sum \left[\frac{w_i}{1+w_i c(0)}\right]^2 c'(0)=\kappa \frac{c'(0)}{[c(0)]^2}-1+\lo_P(1)\;.
$$

Hence, 
$$
\trace{(X\trsp D_w X)^{-2}X\trsp D_{w^2} X}=\left[\kappa \frac{\trace{\SigmaHat_w^{-2}}/n}{[\trace{\SigmaHat_w^{-1}}/n]^2}-1\right]+\lo_P(1)\;.
$$
The fact that we can take expectations on both sides of this equation and that $\lo_P(1)$ is in fact $\lo_{L_2}(1)$ come from our assumptions about $w_i$'s - especially the fact that they are independent and bounded away from 0 - and properties of the inverse Wishart distribution. 

\textbf{Conclusion} We can now conclude that a consistent estimator of the expected variance of the bootstrap estimator is 
$$
\frac{\norm{v}_2^2}{p}\sigma^2_\eps \left[\kappa \frac{\trace{\SigmaHat_w^{-2}}/n}{[\trace{\SigmaHat_w^{-1}}/n]^2}-\frac{1}{1-\kappa}\right]\;.
$$

Using the fact that 
$$
1-\kappa=\frac{1}{n}\sum_{i=1}^n \frac{1}{1+w_i c(z)}+zc(z)\;,
$$
we see that, since $\frac{1}{n}\trace{\SigmaHat_w^{-2}}=c'(0)$,
$$
\frac{1}{n}\trace{\SigmaHat_w^{-2}}=\frac{c(0)}{\frac{1}{n}\sum_{i=1}^{n}w_i/(1+w_ic(0))^2}\;.
$$

We further note that asymptotically, when $w_i$ are i.i.d and satisfy our assumptions, $c(0)\tendsto c$, which solves:
$$
\myExp_{w_i}\left[\frac{1}{1+w_ic}\right]=1-\kappa\;.
$$
Hence, asymptotically, when $w_i$'s are i.i.d and satisfy our assumptions, we have
$$
\frac{\trace{\SigmaHat_w^{-2}}/n}{[\trace{\SigmaHat_w^{-1}}/n]^2}\tendsto \frac{1}{c\myExp_{w_i}[w_i/(1+w_i c)^2]}\;.
$$
Since $cw_i/(1+cw_i^2)=1/(1+cw_i)-1/(1+cw_i)^2$, we finally see that 
\begin{align*}
c\myExp_{w_i}\left[\frac{w_i}{(1+w_i c)^2}\right]&=
\myExp_{w_i}\left[\frac{1}{1+cw_i}\right]-\myExp_{w_i}\left[\frac{1}{(1+cw_i)^2}\right]\;,\\
&=1-\kappa-\myExp_{w_i}\left[\frac{1}{(1+cw_i)^2}\right]\;.
\end{align*}

So asymptotically, the expected bootstrap variance is equivalent to, when $\norm{v}_2=1$, 
$$
\frac{\sigma^2_\eps}{p} \left[\kappa\frac{1}{1-\kappa-\Exp{\frac{1}{(1+cw_i)^2}}}-\frac{1}{1-\kappa}\right]\;,
$$
where $\Exp{\frac{1}{1+cw_i}}=1-\kappa$.

In particular, when $w_i=1$, we see, unsurprisingly that the above quantity is 0, as it should, given that the bootstrapped estimate does not change when resampling.

We finally make note of a technical point, that is addressed in papers such as \cite{nekMarkoRiskPub2010,NEKRobustRegressionRigorous2013} and on which we rely here by using those papers. Essentially, theoretical considerations regarding quantities such as $\frac{1}{p}\trace{\SigmaHat_w^{-k}}$ are easier to handle by working rather with  $\frac{1}{p}\trace{(\SigmaHat_w+\tau\id_p)^{-k}}$, for some $\tau>0$. In the present context, it is easy to show (and done in those papers) that this approximation allows us to take the limit - even in expectation - for $\tau\tendsto 0$ in all the expressions we get for $\tau>0$ and that that limit is indeed $\Exp{\frac{1}{p}\trace{\SigmaHat_w^{-k}}}$. Technical details rely on using the first resolvent identity \citep{KatoPerturbTheory}, using moment properties of inverse Wishart distributions and using the fact that $w_i$'s are bounded below. 

\subsection{On acceptable weight distributions}\label{supp:subsec:acceptableWeightDistribution}

An acceptable weight distribution is such that the variance of the resampled estimator is equal to the variance of the sampling distribution of the original estimator, i.e the least-squares one in the case we are considering. Here, this variance is asymptotically $\kappa/(1-\kappa)\sigma^2_\eps/p$, in the case where $\Sigma=\id_p$. 

Recall that in the main text, we proposed to use 
$$
w_i\iid 1-\alpha+\alpha \textrm{Poisson}(1)
$$

To determine $\alpha$ numerically so that 
$$
\left[\kappa\frac{1}{1-\kappa-\myExp_{w_i}\left[\frac{1}{(1+cw_i)^2}\right]}-\frac{1}{1-\kappa}\right]=\frac{\kappa}{1-\kappa}\;,
$$
we performed a simple dichotomous search for $\alpha$ over the interval $[0,1]$. Our initial $\alpha$ was .95. We specified a tolerance of $10^{-2}$ for the results reported in the paper in Table \ref{tab:GoodWeightProportions}. This means that we stopped the algorithm when the ratio of the two terms in the previous display was within 1\% of 1. We used a sample size of $10^6$ to estimate all the expectations. 

\paragraph{Case $\Sigma\neq \id_p$} In the case where $\Sigma\neq \id$, both $\Exp{\var{v\trsp \betaHat^*}}$ and $\var{v\trsp \betaHat}$ 
depend on $v\trsp \Sigma^{-1} v$. It is therefore natural to ask how we could estimate this quantity. If we are able to do so, it is clear that we could follow the same strategy as above to find $\alpha$ from the data. Standard Wishart results (\cite{mardiakentbibby}, Theorem 3.4.7) give that 
$$
\frac{v\trsp \Sigma^{-1} v}{v\trsp \SigmaHat^{-1}v}\sim \frac{\chi^2_{n-p}}{n-1}\tendsto (1-\kappa) \text{in probability}.
$$
This of course suggests using $(1-p/n)v\trsp \SigmaHat^{-1}v$ as an estimator of $v\trsp \Sigma^{-1} v$ and solves the question we were discussing above. 

However, we note that since 
$$
\frac{\Exp{\var{v\trsp \betaHat^*}}}{\var{v\trsp \betaHat}}
$$
does not depend on $\Sigma$ when the design is Gaussian or Elliptical, the same $\alpha$ should work regardless of $\Sigma$, provided it is positive definite. In particular, an acceptable weight distribution for resampling as defined above could be computed by assuming $\Sigma=\id_p$ and would work for any positive definite $\Sigma$.

\subsection{Numerics for Figure \ref{subfig:bootOverEstFact:Theory}}
This figure, related to the current discussion was generated by assuming Poisson(1) weights and computing deterministically the expectations of interest. This was easy since if $W\sim \text{Poisson}(1)$, $P(W=k)=\frac{\exp(-1)}{k!}$. 

We truncated the expansion of the expectation at $K=100$, so we neglected terms of order $1/100!$ or lower only. The constant $c$ was found by dichotomous search, with tolerance $10^{-6}$ for matching the equation $\Exp{1/(1+Wc)}=1-p/n$. Once $c$ was found, we approximated the expectation in Theorem \ref{thm:bootVariance} in the same fashion as we just described. 

Once we had computed the quantity appearing in Theorem \ref{thm:bootVariance}, we divided it by $\kappa/(1-\kappa)$. We repeated these computations for $\kappa=.05$ to $\kappa=.5$ by increments of $10^{-3}$ to produce our figure.

\subsection{Extensions of Theorem \ref{thm:bootVariance}}\label{supp:subsec:ExtensionsBootVarThm}
\subsubsection{Elliptical Design}\label{supp:PairsEllipticalDesign}
In this case, we have $\tilde{X}_i=\lambda_i X_i$, where $X_i\sim{\cal N}(0,\id_p)$ and $y_i=\eps_i+\tilde{X}_i\trsp \beta$. We assume $\lambda_i\neq 0$ for all $i$, $\Exp{\lambda_i^2}=1$, $\lambda_i$'s are i.i.d and bounded away from 0.

We can go through the proof of Theorem \ref{thm:bootVariance} and make necessary adjustments. 

Of course, we have 
$$
\Exp{\norm{\betaHat_w-\beta}_2^2}=\sigma^2_\eps \Exp{\trace{(\tilde{X}\trsp D_w \tilde{X})^{-2}\tilde{X}\trsp D_{w^2}\tilde{X}}}\;.
$$
If we reformulate this expression in terms of $X$ we get 
$$
\Exp{\norm{\betaHat_w-\beta}_2^2}=\sigma^2_\eps \Exp{\trace{(X\trsp D_{\lambda^2 w} X)^{-2}X\trsp D_{\lambda^2w^2}X}}\;.
$$
So this quantity is affected by the distribution of $\lambda_i$'s; hence the risk of $\betaHat_w$ is different in the Gaussian and elliptical design case. 

The other important part of the proof is the computation of the risk of the bagged estimator. In this case, earlier work in random matrix theory (e.g \cite{nekCorrEllipD,NEKHolgerShrinkage11}) shows that we can use the approximations
$$
c_i\simeq \lambda_i^2 c \;,
$$
where $c=\lim_{n,p\tendsto \infty}\Exp{\frac{1}{n}\trace{S^{-1}}}$, where $S=\frac{1}{n}\sum_{i=1}^n \lambda_i^2 w_i X_i X_i\trsp$, i.e $S=\frac{1}{n}X\trsp D_{\lambda^2 w} X$.

If we call 
$$
g(\lambda_i^2)=\myExp_{w_i}\left(\frac{c\lambda_i^2}{1+c\lambda_i^2 w_i}\right)\;,
$$
we see by keeping track of changes in the earlier proof that we have asymptotically
$$
\Exp{\norm{\hat{b}}_2^2}=\frac{\myExp_{\lambda_i}(\lambda_i^2 g^2(\lambda_i^2))}{\kappa-\myExp_{\lambda_i}(\lambda_i^4 g^2(\lambda_i^2))} \sigma^2_\eps\;.
$$
The same arguments we used before give that 
$$
\Exp{g(\lambda_i^2)}=\kappa\;.
$$
Based on this information, we can compute $\Exp{\var{v\trsp \betaHat^*_w}}$ as we had in the proof of Theorem \ref{thm:bootVariance} and compare it to $\var{v\trsp \betaHat}$. The expressions do not seem to simplify much further however in this case, by contrast to the Gaussian design case where $\lambda_i=1$ for all $i$. (For instance, when $\lambda_i$=1 for all $i$'s, $g(\lambda_i)=g(1)=\kappa$ and we recover the results of Theorem \ref{thm:bootVariance}.) 

Importantly, the characteristics of the distribution of $\lambda_i$ that affect $\Exp{\var{v\trsp \betaHat^*_w}}$ go beyond $\Exp{\lambda_i^2}$. And hence the expression we gave in Theorem \ref{thm:bootVariance} won't apply directly to the elliptical case. 

\subsubsection{Multinomial($n,1/n$) weights}\label{supp:MultinomialVsPoisson}
A natural question is whether the computations we have made can be extended to $w_i$'s that are i.i.d $Poisson(1)$ and/or Multinomial($n,1/n$), as in the standard bootstrap. 

In both cases, technical issues arise because with asymptotically negligible but non-zero probability, the matrix $X\trsp D_w X$ may be of rank less than $p$. This can handled in several ways. A simple one is to replace the weights $w_i$ by $w_i(\tau)=\tau+(1-\tau) w_i$ and study the problem when $\tau\tendsto 0$. 

Beyond that technicality, an important question is whether one can handle the fact that the weights are dependent in the multinomial case. For quantities of the type $\frac{1}{n}\trace{(X\trsp D_w X)^{-1}}$, it was argued in \cite{nekMarkoRiskPub2010} that one could ignore the dependency issue and treat the problem as if the weights where i.i.d $\Poisson(1)$. This type of arguments would be easy to extend where we need them here, for instance in quantities that arise in the computation of $\Exp{\norm{\betaHat_w-\beta}_2^2}$ or to show that we can write $c_i=c+\lo_P(1)$, where $c$ is deterministic.  

The remaining question is therefore the characterization of the risk of the bagged estimator. We have, with a slight modification with respect to the case of independent weights,
$$
\hat{b}_p=\frac{1}{p}\sum_{i=1}^n X_i(p)\left[\myExp_{w_i}\left(\frac{cw_i}{1+cw_i}\right) \eps_i -V_i\trsp \myExp_{w}\left(\hat{\gamma}_{(i)}\frac{cw_i}{1+cw_i}\right)\right] (1+\lo_P(1))\;.
$$
As before, $\myExp_{w_i}\left(\frac{cw_i}{1+cw_i}\right)=p/n+\lo(1)$. The problem is the dependence between $\hat{\gamma}_{(i)}$ and $w_i$. The rotational invariance arguments we invoked before still hold, so that $\hat{\gamma}_i=\norm{\hat{\gamma}_{(i)}}_2 u$, where $u$ is uniform on the unit sphere and independent of $\norm{\hat{\gamma}_i}_2$. It is also independent of $V_i$, since $\hat{\gamma}_{(i)}$ is the leave-one-out estimate of $\gamma$. The same rotational invariance arguments hold for the bagged estimate $\myExp_w{\hat{\gamma}_{(i)} \frac{cw_i}{1+cw_i}}$. Hence, after a little bit of work we see that 
$$
\Exp{[V_i\trsp \myExp_{w}\left(\hat{\gamma}_{(i)}\frac{cw_i}{1+cw_i}\right)]^2}=\Exp{\norm{\myExp_{w}\left(\gammaHat_{(i)}\frac{cw_i}{1+cw_i}\right)}_2^2}\;.
$$
Using the fact that $w_{(i)}|w_i\sim \text{Multinomial}(n-w_i,1/(n-1))$, the only real technical hurdle is to show that $\myExp_{w_i}{\hat{\gamma}_{(i)}}$ is asymptotically deterministic and independent of $w_i$. A strategy for this is to create a coupling: one can compare $\hat{\gamma}_{(i)}$ to $\hat{\mathfrak{g}}_{(i)}$, where $\hat{\mathfrak{g}}_{(i)}$ is computed using a $n-1$ dimensional vector of weights with distribution $\text{Multinomial}(n-1,1/(n-1))$ - i.e running $w_i-1$ multinomial trials after having obtained $w_{(i)}$ (the case $w_i=0$ is easy to handle separately). Clearly, the distribution of $\hat{\mathfrak{g}}_{(i)}$ is independent of $w_i$, by construction. On the other hand, a bit of work on top of the leave-one-observation-out expansions show that $\norm{\hat{\mathfrak{g}}_{(i)}-\hat{\gamma}_{(i)}}_2^2$ is roughly of size at most $w_i^2/n\tendsto 0$. Furthermore, $\norm{\myExp_w(\hat{\mathfrak{g}}_{(i)})-\myExp_w(\gammaHat_{(i)})}_2\tendsto 0$ for the same reason. This suggests that further technical work along those lines will give that 
$$
\myExp_{w}\left(\gammaHat_{(i)}\frac{cw_i}{1+cw_i}\right)\simeq \myExp_{w}\left(\gammaHat_{(i)}\right)\myExp_{w}\left(\frac{cw_i}{1+cw_i}\right)\;,
$$
where $\simeq$ means that the approximation is valid in Euclidean norm. The same coupling arguments will give that
$$
\norm{\myExp_{w}\left(\gammaHat_{(i)}\right)}\simeq \norm{\hat{b}}\;,
$$
where $\hat{b}$ is the bagged estimator. This will yield the same results as in the i.i.d $\Poisson(1)$ case. 

\paragraph{Numerical results}
We verified that our theoretical results (i.e Theorem \ref{thm:bootVariance} hold for Poisson(1) weights in limited simulations (note that in this case $w_i=0$ is possible).  For Gaussian design matrix, double exponential errors, and ratios $\kappa=.1, .3, .5$ we found that the ratio of the observed bootstrap expected variance of $\betaHat_1^*$ to our theoretical prediction using Poisson(1) weights was 1.0027, 1.0148, and 1.0252, respectively (here $n=500$, and there were $R=1000$ bootstrap resamples for each of $1000$ simulations).

\section{Jackknife Variance (Proof of Theorem \ref{thm:Jackknife})}\label{supp:JackknifeProof}
As explained in Section \ref{sec:supp:CovAndInvarianceIssue}, we can study without loss of generality the case where $\Sigma=\id_p$ and $\beta=0$. This is what we do in this proof.

We study it in details in the least-squares case, and postpone a detailed analysis of the robust regression case to future studies. 

According to the approximations in \cite{NEKRobustPaperPNAS2013Published}, which are exact for least squares, or classic results \cite{WeisbergLinearRegressionBook14} we have:
$$
\betaHat-\betaHat_{(i)}=\frac{1}{n}\SigmaHat_{(i)}^{-1}X_i \resid_i\;.
$$ 
Recall also that 
$$
\resid_i=\frac{\prederror}{1+\frac{1}{n}X_i\trsp \SigmaHat_{(i)}^{-1} X_i}\;.
$$
Hence,
$$
v\trsp(\betaHat-\betaHat_{(i)})=\frac{1}{n}v\trsp\SigmaHat_{(i)}^{-1}X_i \frac{\prederror}{1+\frac{1}{n}X_i\trsp \SigmaHat_{(i)}^{-1} X_i}\;.
$$

Hence, 
$$
n\sum_{i=1}^n [v\trsp(\betaHat-\betaHat_{(i)})]^2=\frac{1}{n}\sum_{i=1}^n \frac{[v\trsp\SigmaHat_{(i)}^{-1}X_i \prederror]^2}{[1+\frac{1}{n}X_i\trsp \SigmaHat_{(i)}^{-1} X_i]^2}\;.
$$

Note that at the denominator, we have 
\begin{align*}
1+\frac{1}{n}X_i\trsp \SigmaHat_{(i)}^{-1} X_i&= 1+\frac{1}{n}\trace{\SigmaHat^{-1}}+\lo_P(1)\;,\\ &=1+\frac{p}{n}\frac{1}{1-p/n}+\lo_P(1)=\frac{1}{1-p/n}+\lo_P(1)\;.
\end{align*}
by appealing to standard results about concentration of high-dimensional Gaussian random variables, and standard results in random matrix theory and classical multivariate statistics (see \cite{mardiakentbibby,Haff79IdentityWishartDWithApps}). By the same arguments, this approximation works not only for each $i$ but for all $1\leq i \leq n$ at once. The approximation is also valid in expectation, using results concerning Wishart matrices found for instance in \cite{mardiakentbibby}. 

For the numerator, we see that 
$$
T_i=v\trsp\SigmaHat_{(i)}^{-1}X_i \prederror=v\trsp\SigmaHat_{(i)}^{-1}X_i (\eps_i-X_i\trsp(\betaHat_{(i)}-\beta))\;.
$$
Since $\eps_i$ is independent of $X_i$ and $\SigmaHat_{(i)}$, we see that
$$
\Exp{T_i^2}=\Exp{\eps_i^2}\Exp{(v\trsp \SigmaHat_{(i)}^{-1}X_i)^2}+\Exp{[X_i\trsp (\betaHat_{(i)}-\beta)]^2 [v\trsp \SigmaHat_{(i)}^{-1}X_i]^2}\;.
$$
If $\alpha$ and $\beta$ are fixed vectors, $\alpha\trsp X_i$ and $\beta\trsp X_i$ are Gaussian random variables with covariance $\alpha\trsp \beta$, since we are working under the assumption that $X_i\sim {\cal N}(0,\id_p)$. It is easy to check that if $Z_1$ and $Z_2$ are two Gaussian random variables with covariance $\gamma$ and respective variances $\sigma_1^2$ and $\sigma_2^2$, we have
$$
\Exp{(Z_1Z_2)^2}=\sigma_1^2 \sigma_2^2+2\gamma^2\;.
$$

We conclude that 
$$
\Exp{(a\trsp X_i)^2(b\trsp X_i)^2}=\norm{a}_2^2 \norm{b}_2^2+2 (a\trsp b)^2\;.
$$

We note that 
$$
\Exp{[v\trsp \SigmaHat_{(i)}^{-1}X_i]^2}=\Exp{v\trsp \SigmaHat_{(i)}^{-2}v}\;.
$$
Classic Wishart computations give (\cite{Haff79IdentityWishartDWithApps}, p.536 (iii)) that as $n,p\tendsto \infty$, 
$$
\Exp{\SigmaHat_{(i)}^{-2}}= (\frac{1}{(1-p/n)^3}+\lo(1)) \id_p\;.
$$
Hence, in our asymptotics,
$$
\Exp{(v\trsp \SigmaHat_{(i)}^{-1}X_i)^2}\tendsto \frac{1}{(1-p/n)^3} \norm{v}_2^2\;.
$$
We also note that 
$$
\myExp_{\eps}\left[(v\trsp \SigmaHat_{(i)}^{-1}\betaHat_{(i)})^2\right]=\frac{1}{n} v\trsp \SigmaHat_{(i)}^{-3} v\;. 
$$
Hence, 
$$
\Exp{(v\trsp \SigmaHat_{(i)}^{-1}\betaHat_{(i)})^2}=\lo(1) \text{in our asymptotics}\;.
$$
Therefore, 
$$
\Exp{T_1^2}=\frac{1}{(1-p/n)^3} \norm{v}_2^2 \sigma^2_\eps(1+\frac{p/n}{1-p/n})+\lo(1)\;
$$
since $\Exp{\norm{\betaHat_{(i)}-\beta}_2^2}=\sigma^2_\eps\frac{p/n}{1-p/n}+\lo(1)$.

When $v=e_1$, we therefore have 
$$
\Exp{T_1^2}=\sigma^2_\eps \frac{1}{(1-p/n)^4}+\lo(1)\;.
$$

Therefore, in that situation,
$$
\Exp{n\sum_{i=1}^n (v\trsp (\betaHat_{(i)}-\betaHat)^2)}=\sigma^2_\eps \frac{1}{(1-p/n)^2}+\lo(1)\;.
$$
In other words,
$$
\Exp{\sum_{i=1}^n (v\trsp (\betaHat_{(i)}-\betaHat)^2)}=\left[\frac{1}{1-p/n}+\lo(1)\right] \var{\betaHat_1}
$$
\subsection{Dealing with the centering issue}
Let us call $\betaHat_{(\cdot)}=\frac{1}{n}\sum_{i=1}^n \betaHat_{(i)}$. 
We have previously studied the properties of $\sum_{i=1}^n([v\trsp (\betaHat-\betaHat_{(i)})]^2)$ and now need to show that the same results apply to 
$\sum_{i=1}^n([v\trsp (\betaHat_{(\cdot)}-\betaHat_{(i)})]^2)$.

To show that replacing $\betaHat$ by $\betaHat_{(\cdot)}$ does not affect the result, we consider the quantity
$$
n^2 [v\trsp (\betaHat-\betaHat_{(\cdot)})]^2\;.
$$
Since $\betaHat-\betaHat_{(i)}=\frac{1}{n}\SigmaHat_{(i)}^{-1}X_i \resid_i$, we have 
$$
\betaHat-\betaHat_{(\cdot)}=\frac{1}{n^2}\sum_{i=1}^n \SigmaHat_{(i)}^{-1}X_i \resid_i\;.
$$
Hence, 
$$
n^2[v\trsp (\betaHat-\betaHat_{(\cdot)})]^2=\left[\frac{1}{n}\sum_{i=1}^n v\trsp \SigmaHat_{(i)}^{-1}X_i (\eps_i-X_i\trsp (\betaHat-\beta))\right]^2\;.
$$
A simple variance computation gives that $\frac{1}{n}\sum_{i=1}^n v\trsp \SigmaHat_{(i)}^{-1}X_i \eps_i\tendsto 0$ in $L^2$, since each term has mean 0 and the variance of the sum goes to 0.

Recall now that 
$$
\SigmaHat^{-1}X_i=\frac{\SigmaHat_{(i)}^{-1}X_i}{1+c_i}\;,
$$
where all $c_i$'s are equal to $p/n/(1-p/n)+\lo_P(1)$. Let us call $\mathsf{c}=p/n/(1-p/n).$

We conclude that 
$$
\frac{1}{n}\sum_{i=1}^nv\trsp \SigmaHat_{(i)}^{-1}X_i X_i\trsp (\betaHat-\beta) =v\trsp(\betaHat-\beta)(1+\mathsf{c}+\lo(1))\;.
$$
When $v$ is given, we clearly have $v\trsp (\betaHat-\beta)=\lo_P(p^{-1/2})$, given the distribution of $\betaHat-\beta$ under our assumptions on $X_i$'s and $\eps_i$'s. So we conclude that 
$$
n^2 [v\trsp (\betaHat-\betaHat_{(\cdot)})]^2\tendsto 0 \text{ in probability}\;.
$$
Because we have enough moments, the previous result is also true in expectation.
\subsection{Putting everything together}
The jackknife estimate of variance of $v\trsp \betaHat$ is up to a factor going to 1
\begin{align*}
\frac{n}{n-1}\text{JACK}(\var{v\trsp \betaHat})&=\sum_{i=1}^n [(v\trsp \betaHat_{(i)}-\betaHat_{(\cdot)})]^2\\
&=\sum_{i=1}^n [(v\trsp \betaHat_{(i)}-\betaHat)]^2+n[v\trsp (\betaHat-\betaHat_{(\cdot)})]^2\;.
\end{align*}
Our previous analyses therefore imply (using $v=e_1$) that 
$$
\frac{n}{n-1}\Exp{\text{JACK}(\var{\betaHat_1})}=\left[\frac{1}{1-p/n}+\lo(1)\right] \var{\betaHat_1}\;.
$$

This completes the proof of Theorem \ref{thm:Jackknife}
\subsection{Extension to more involved design and different loss functions}\label{supp:jackKnife:DesignIssue}
Our approach could be used to analyze similar problems in the case of elliptical designs. However, in that case, it seems that the factor that will appear in quantifying the amount by which the variance is misestimated will depend in general on the ellipticity parameters. We refer to \cite{NEKRealizedRiskMarkoPublished2013} for computations of quantities such as $v\trsp \SigmaHat^{-2} v$ in that case, which are of course essential to measuring mis-estimation.

We obtained the possible correction we mentioned in the paper for these more general settings following the ideas used in the rigorous proof we just gave, as well as approximation arguments given in \cite{NEKRobustPaperPNAS2013Published} and justified rigorously in \cite{NEKRobustRegressionRigorous2013}. Checking fully rigorously all the approximations we made in this Jackknife computation would require a very large amount of technical work, and since this is tangential to our main interests in this paper, we postpone that to a future work of a more technical nature. 

It is also clear, since all these results and the proof we just gave rely on random matrix techniques, that a similar analysis could be carried out in the case where $X_{i,j}$ are i.i.d with a non-Gaussian distribution, provided that distribution has enough moments (see e.g \cite{PajorPasturPub09} or \cite{NEKHolgerShrinkage11} for examples of such techniques, actually going beyond the case of i.i.d entries for the design matrix). The main issues in carrying out this program seem to be technical and not conceptual at this point, so we leave this problem to possible future work.

\section{More details on going from $\Sigma=\id_p$ to $\Sigma\neq \id_p$}\label{sec:supp:CovAndInvarianceIssue}
As discussed in Section \ref{supp:Reminders}, we have 
$$
\betaHat_\rho(y_i;X_i;\eps_i)-\beta=\Sigma^{-1/2} \betaHat_\rho(\eps_i;\Sigma^{-1/2} X_i;\eps_i)\;,
$$
In other words, $\betaHat_\rho(\tilde{y}_i;\Sigma^{-1/2} X_i;\eps_i)$ is the robust regression estimator in the null case where $\beta=0$ and $X_i$ is replaced by $\tilde{X}_i=\Sigma^{-1/2}X_i$. Of course, if $\scov{X_i}=\Sigma$, $\scov{\tilde{X}_i}=\id_p$.

\subsection{Consequences for the Jackknife}
Naturally the same equality applies to leave-one-out estimators. So, with the notations of Equation \eqref{eq:defJackVar} in the main text, we have, when $\textrm{span}(\{X_i\}_{i=1}^n)=\mathbb{R}^p$ and $\Sigma$ is positive definite, 
$$
(v\trsp[\betaHat_{(i)}-\tilde{\beta}])^2=(v\trsp \Sigma^{-1/2}[\betaHat_{(i)}(\eps_i;\Sigma^{-1/2} X_i;\eps_i)-\tilde{\beta}(\eps_i;\Sigma^{-1/2} X_i;\eps_i)])^2\;.
$$
Let us call $\betaHat_\rho(\Sigma;\beta)$ our robust regression estimator when $\scov{X_i}=\Sigma$ and $\Exp{y_i|X_i}=X_i\trsp \beta$. 
It is  clear from the previous display that the properties of $\varJack(v\trsp\betaHat_\rho(\Sigma;\beta))$ are the same as those of $\varJack(v\trsp \Sigma^{-1/2} \betaHat_\rho(\id_p;0))$. So understanding the null case is enough to understand the general case, which is why we focus on the null case in our computations. 

Furthermore, by the same arguments, we have 
$$
\var{v\trsp \betaHat_\rho(y_i;X_i;\eps_i)} = \var{v\trsp \Sigma^{-1/2} \betaHat_\rho(\id_p;0)}. 
$$

So we have 
$$
\frac{\varJack(v\trsp\betaHat_\rho(\Sigma;\beta))}{\var{v\trsp \betaHat_\rho(\Sigma;\beta)}}=\frac{\varJack(v\trsp\Sigma^{-1/2}\betaHat_\rho(\id_p;0))}{\var{v\trsp \Sigma^{-1/2} \betaHat_\rho(\id_p;0)}}\;.
$$
Calling $u_1=\Sigma^{-1/2} v/\norm{\Sigma^{-1/2}v}$, we see that $u_1$ is a unit vector. And we finally see that 
$$
\frac{\varJack(v\trsp\betaHat_\rho(\Sigma;\beta))}{\var{v\trsp \betaHat_\rho(\Sigma;\beta)}}=\frac{\varJack(u_1\trsp \betaHat_\rho(\id_p;0))}{\var{u_1\trsp \betaHat_\rho(\id_p;0)}}\;.
$$

Hence, characterizing $\frac{\varJack(v\trsp \betaHat_\rho(\id_p;0))}{\var{v\trsp \betaHat_\rho(\id_p;0)}}$ for all fixed unit vectors $v$ characterizes 
$$
\frac{\varJack(v\trsp\betaHat_\rho(\Sigma;\beta))}{\var{v\trsp \betaHat_\rho(\Sigma;\beta)}}
$$
for all $\beta$ and invertible $\Sigma$. This is why our proof is focused on the null case $\Sigma=\id_p$ and $\beta=0$.

\subsection{Consequences for the pairs bootstrap}
Let us call $D_w$ the diagonal matrix with $(i,i)$-entry $D(i,i)=w_i$. 
We consider only the case where $w_i>0$, so we do not have to consider the case where fewer than $p$ $X_i$'s are assigned positive weights - which would result in $\betaHat_\rho$ being ill-defined (since infinitely many solutions would then be feasible). 

In particular, for least squares, we have in our setting
$$
\betaHat_w=(X\trsp D_w X)^{-1}X\trsp D_w Y=\beta+(X\trsp D_w X)^{-1}X\trsp D_w \eps\;.
$$

More generally, by a simple change of variables, since $w_i>0$ and $\textrm{span}(\{X_i\}_{i=1}^n)=\mathbb{R}^p$, when $\Sigma$ is invertible, 
$$
\betaHat_{w,\rho}(y_i;\{X_i\}_{i=1}^n;\eps_i)-\beta=\Sigma^{-1/2} \betaHat_{w,\rho}(\eps_i;\Sigma^{-1/2} X_i;\eps_i)\;.
$$
If $b_\rho$ is the corresponding bagged estimate, obtained by averaging $\betaHat_{w,\rho}$ over $w$'s, we also have 
$$
b_\rho(y_i;\{X_i\}_{i=1}^n;\eps_i)-\beta=\Sigma^{-1/2 }b_{\rho}(\eps_i;\Sigma^{-1/2} X_i;\eps_i)\;.
$$
Hence, we also have 
$$
\betaHat_{w,\rho}(y_i;\{X_i\}_{i=1}^n;\eps_i)-b_\rho(y_i;\{X_i\}_{i=1}^n;\eps_i)=\Sigma^{-1/2} \left[\betaHat_{w,\rho}(\eps_i;\Sigma^{-1/2} X_i;\eps_i)-b_{\rho}(\eps_i;\Sigma^{-1/2} X_i;\eps_i)\right]
$$
We further note that since $y_i=\eps_i+X_i\trsp \beta$, $y_i=\eps_i+(\Sigma^{-1/2}X_i)\trsp \Sigma^{1/2}\beta$ and hence 
$$
\betaHat_{w,\rho}(y_i;\Sigma^{-1/2} X_i;\eps_i)=\Sigma^{1/2} \beta+\betaHat_{w,\rho}(\eps_i;\Sigma^{-1/2} X_i;\eps_i)\;.
$$

The previous equation clearly implies that, if $v$ is a fixed vector and $u_1=\Sigma^{-1/2} v$
\begin{gather*}
v\trsp (\betaHat^*_{\rho}(y_i;\{X_i\}_{i=1}^n;\beta)-b_{\rho}(y_i;\{X_i\}_{i=1}^n;\beta))\\
=u_1 \trsp \left[\betaHat^*_{\rho}(y_i;\Sigma^{-1/2} X_i;\eps_i)-b_{\rho}(y_i;\Sigma^{-1/2} X_i;\eps_i)\right]\;, \\
=u_1 \trsp \left[\betaHat^*_{\rho}(\eps_i;\Sigma^{-1/2} X_i;\eps_i)-b_{\rho}(\eps_i;\Sigma^{-1/2} X_i;\eps_i)\right]\;. 
\end{gather*}

We note that if $\scov{X_i}=\Sigma$, the last line in the previous display corresponds to the bootstrap distribution of our estimator in the null case where $\beta=0$ and $\Sigma=\id_p$, but $v$ has been replaced by $u_1=\Sigma^{-1/2} v$. 
This shows that understanding the bootstrap properties of $v\trsp (\betaHat^*_\rho-b_\rho)$ in the null case $\scov{X_i}=\Sigma$ and $\beta=0$ gives the result we seek in the general case of $\Sigma\neq \id_p$ and $\beta\neq 0$. (Here we centered our estimator around the bagged estimator, because it is natural when computing bootstrap variances. The arguments above show that many other centering choices are possible, however.)

The last small issue that one needs to handle is the fact that our computations are done for $v$ with unit norm and $u_1$ may not have unit norm. This is easily handled by simply scaling by the deterministic $\norm{u_1}$. In particular, it is easy to see through simple scaling arguments that 
$$
\frac{\Exp{\var{v\trsp \betaHat^*_\rho(\Sigma;\beta)}}}{\var{v\trsp \betaHat_\rho(\Sigma;\beta)}}=\frac{\Exp{\var{\tilde{u_1}\trsp \betaHat^*_\rho(\id_p;0)}}}{\var{\tilde{u}_1\trsp \betaHat_\rho(\id_p;0)}}\;,
$$
where $\tilde{u}_1=u_1/\norm{u_1}$ has unit norm.

\subsection{Rotational invariance arguments and consequences}
Motivated by the arguments in the previous two subsections, we now consider the null case where $\beta=0$ and $\scov{X_i}=\id_p$. Note that then $y_i=\eps_i$. 
Also, if $X_i$ is replaced by $OX_i$, where $O$ is an orthogonal matrix, and $\betaHat$ is replaced by $O\betaHat$. In other words, 
$$
\betaHat_\rho(\eps_i;\{OX_i\}_{i=1}^n;\eps_i)=O \betaHat_\rho(\eps_i;\{X_i\}_{i=1}^n;\eps_i)\;,.
$$
Note that when the design matrix is such that $OX_i\equalInLaw X_i$ for all $i$ (i.e the distribution of $X_i$'s is invariant by rotation), 
$$
\betaHat_\rho(\eps_i;\{OX_i\}_{i=1}^n;\eps_i)\equalInLaw \betaHat_\rho(\eps_i;\{X_i\}_{i=1}^n;\eps_i)\;.
$$

When $w_i>0$ for all $i$, we see that exactly the same arguments apply to $\betaHat_{w,\rho}(\eps_i;\{X_i\}_{i=1}^n;\eps_i)$ and hence $\betaHat^*_\rho(\eps_i;\{X_i\}_{i=1}^n;\eps_i)$. 
In particular, for any orthogonal matrix $O$, since $X_i\equalInLaw OX_i$,
\begin{align*}
\Exp{\var{v\trsp \betaHat^*_\rho(\eps_i;\{X_i\}_{i=1}^n;\eps_i)}}&=\Exp{\var{v\trsp \betaHat^*_\rho(\eps_i;\{OX_i\}_{i=1}^n;\eps_i)}}\\
&=\Exp{\var{v\trsp O\betaHat^*_\rho(\eps_i;\{X_i\}_{i=1}^n;\eps_i)}}\;.
\end{align*}
This implies that for any unit vector $v$, we have, if $e_1$ is the first canonical basis vector,  
$$
\Exp{\var{v\trsp \betaHat^*_\rho(\eps_i;X_i;\eps_i)}}=\Exp{\var{e_1\trsp \betaHat^*_\rho(\eps_i;X_i;\eps_i)}}\;.
$$
Indeed, we just need to take $O$ to be such that $O\trsp v=e_1$ to prove the above result. 

In the case where $X_i$'s are i.i.d ${\cal N}(0,\id_p)$, we do have $X_i\equalInLaw OX_i$, so the arguments above apply. Therefore, to understand $\Exp{\var{v\trsp \betaHat^*_\rho}}$ in this case it is sufficient to understand $\Exp{\var{e_1\trsp \betaHat^*_\rho}}$. This latter case is the case tackled in the proof of Theorem \ref{thm:bootVariance}. (These rotational invariance arguments are closely related to those in \cite{NEKRobustPaperPNAS2013Published}.)

\begin{center}
\textbf{\textsc{Supplementary Figures}}
\end{center}

\begin{figure}[h]
	\centering
	\subfloat[][$L_1$ loss]{\includegraphics[width=.3\textwidth]{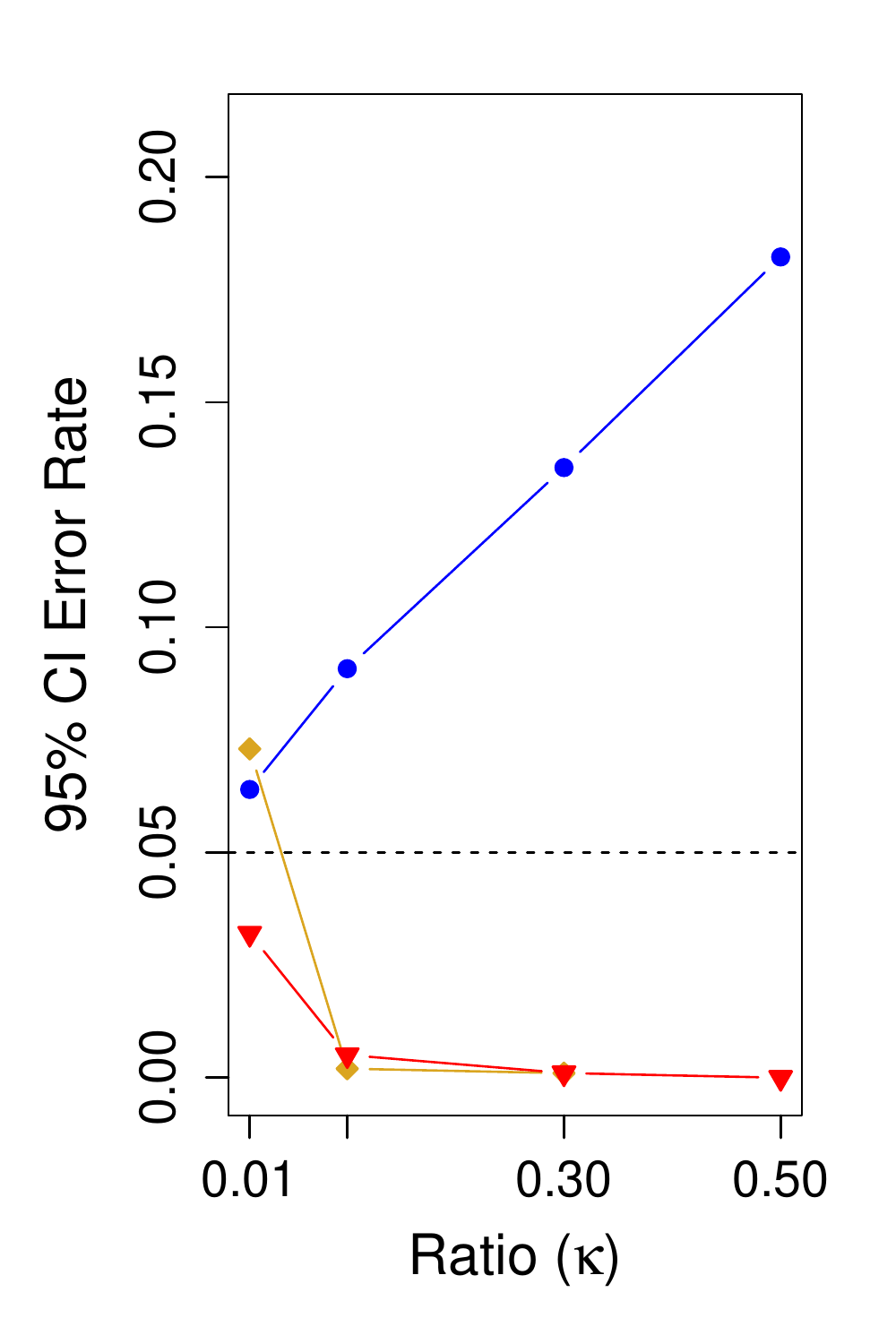} \label{subfig:basicCIErrorLap:L1} }
	\subfloat[][Huber loss]{\includegraphics[width=.3\textwidth]{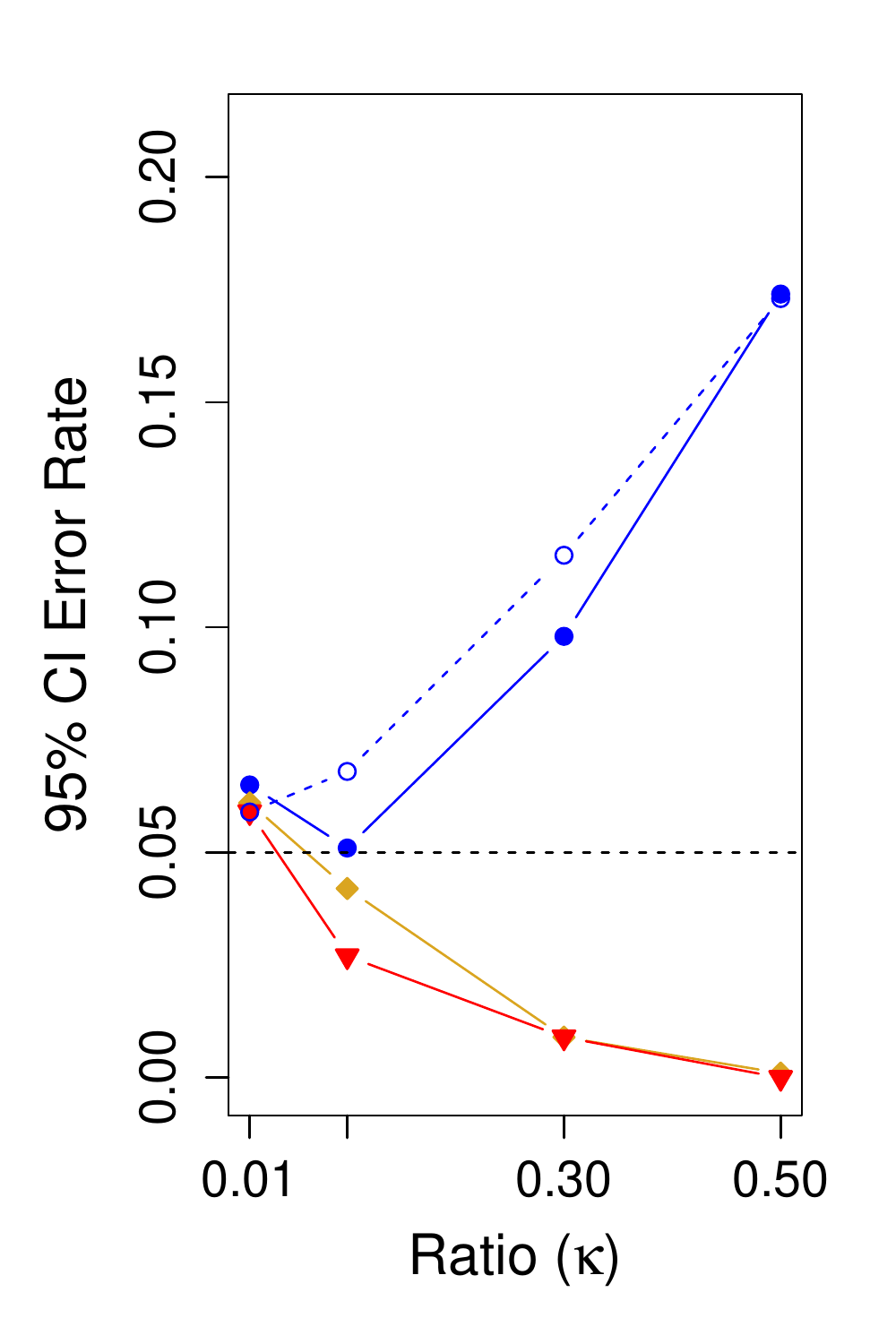} \label{subfig:basicCIErrorLap:Huber} }\\
	\subfloat[][$L_2$ loss]{\includegraphics[width=.3\textwidth]{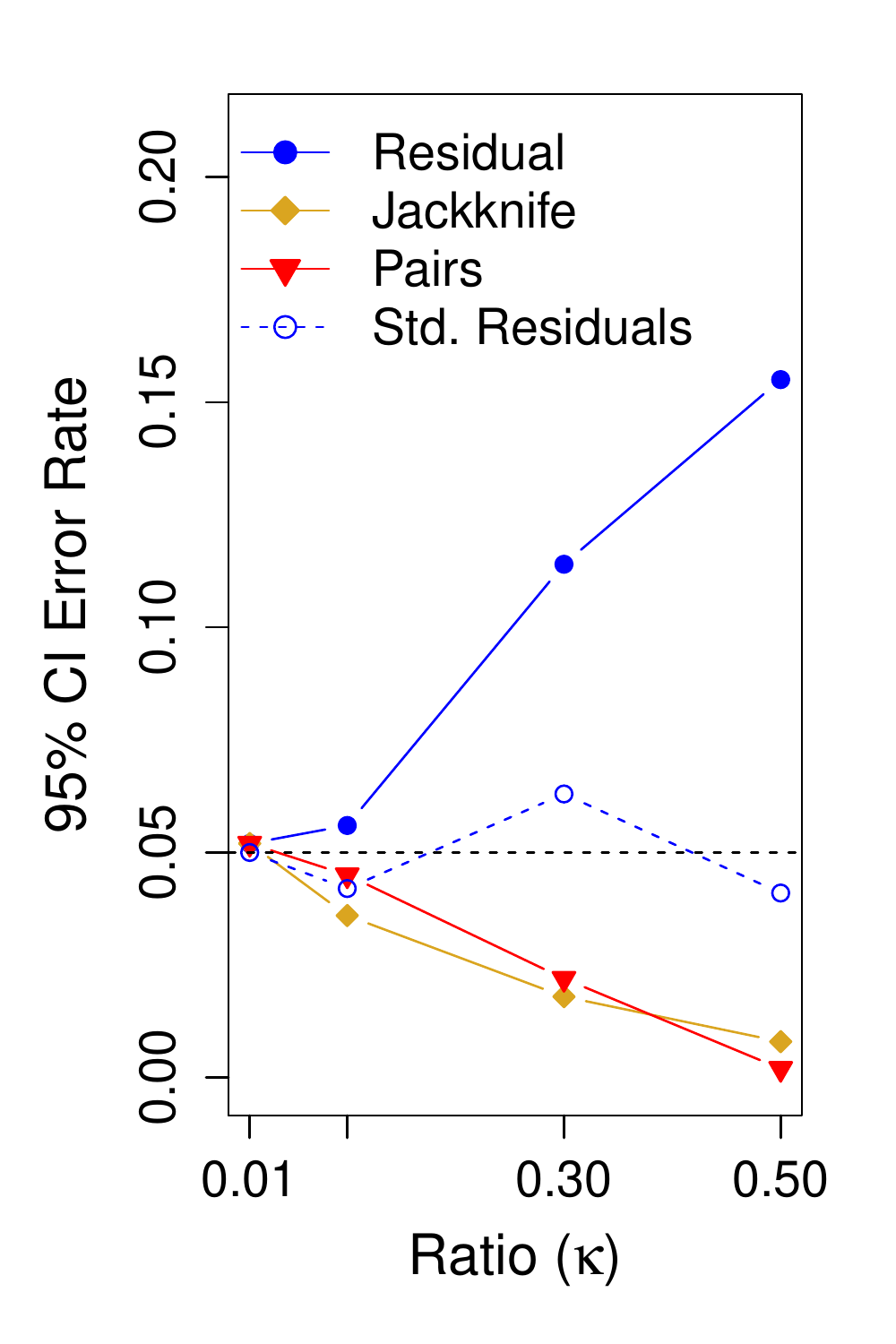}
	 \label{subfig:basicCIErrorLap:L2} }
\caption{
 \textbf{Performance of 95\% confidence intervals of $\beta_1$ (double exponential error): }  Here we show the coverage error rates for 95\% confidence intervals for $n=500$ with the error distribution being double exponential (with $\sigma^2=2$) and i.i.d. normal entries of $X$. See the caption of Figure \ref{fig:basicCIError} for more details. 
}\label{fig:basicCIErrorLap}
\end{figure}

\begin{figure}[h]
	\centering
	\subfloat[][Normal $X$]{\includegraphics[width=.3\textwidth]{errorBasicL2} \label{subfig:basicCIErrorDesign:N} }
	\subfloat[][Ellip. $X$, Unif]{\includegraphics[width=.3\textwidth]{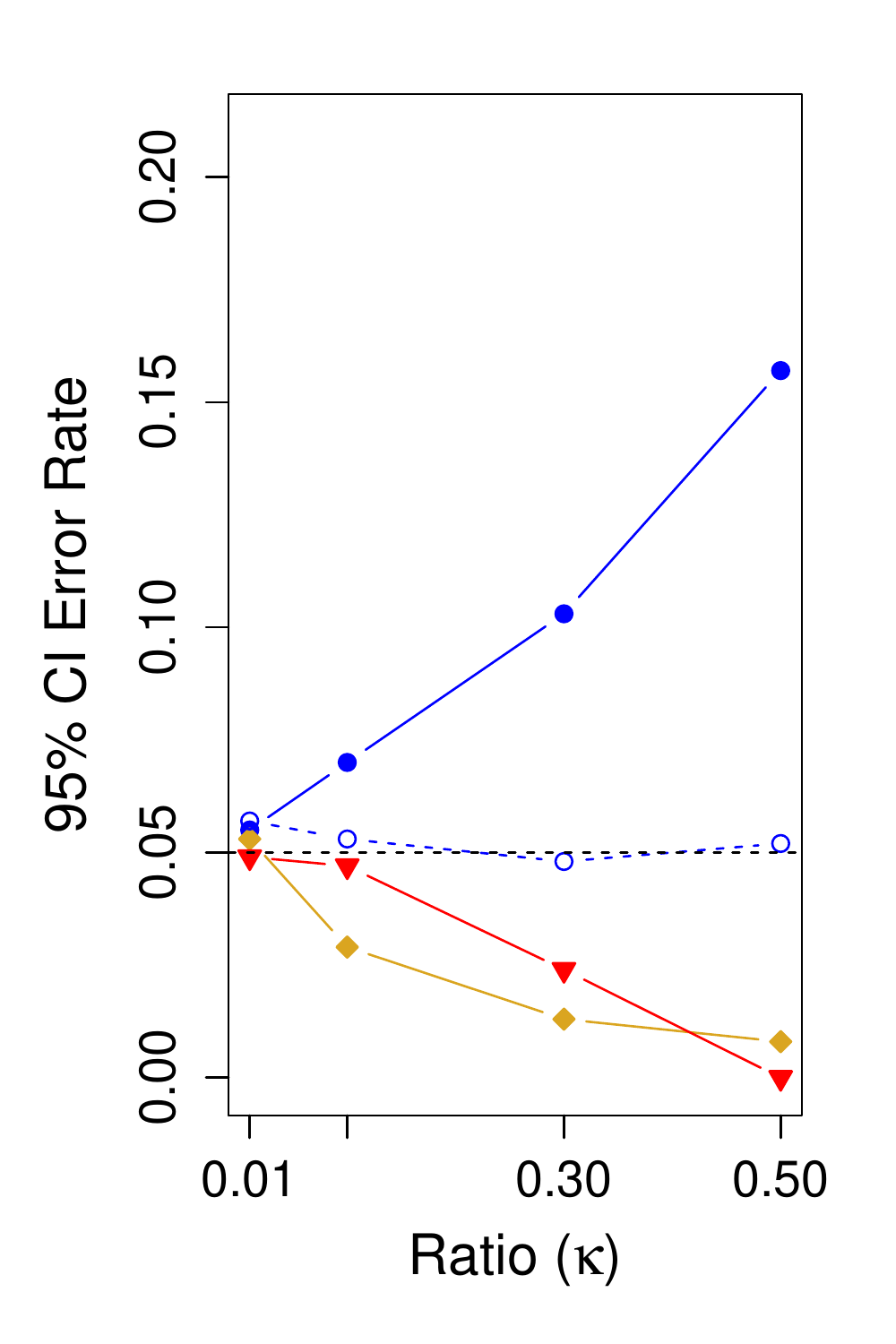} \label{subfig:basicCIErrorDesign:EU} }\\
	\subfloat[][Ellip. $X$, $N(0,1)$]{\includegraphics[width=.3\textwidth]{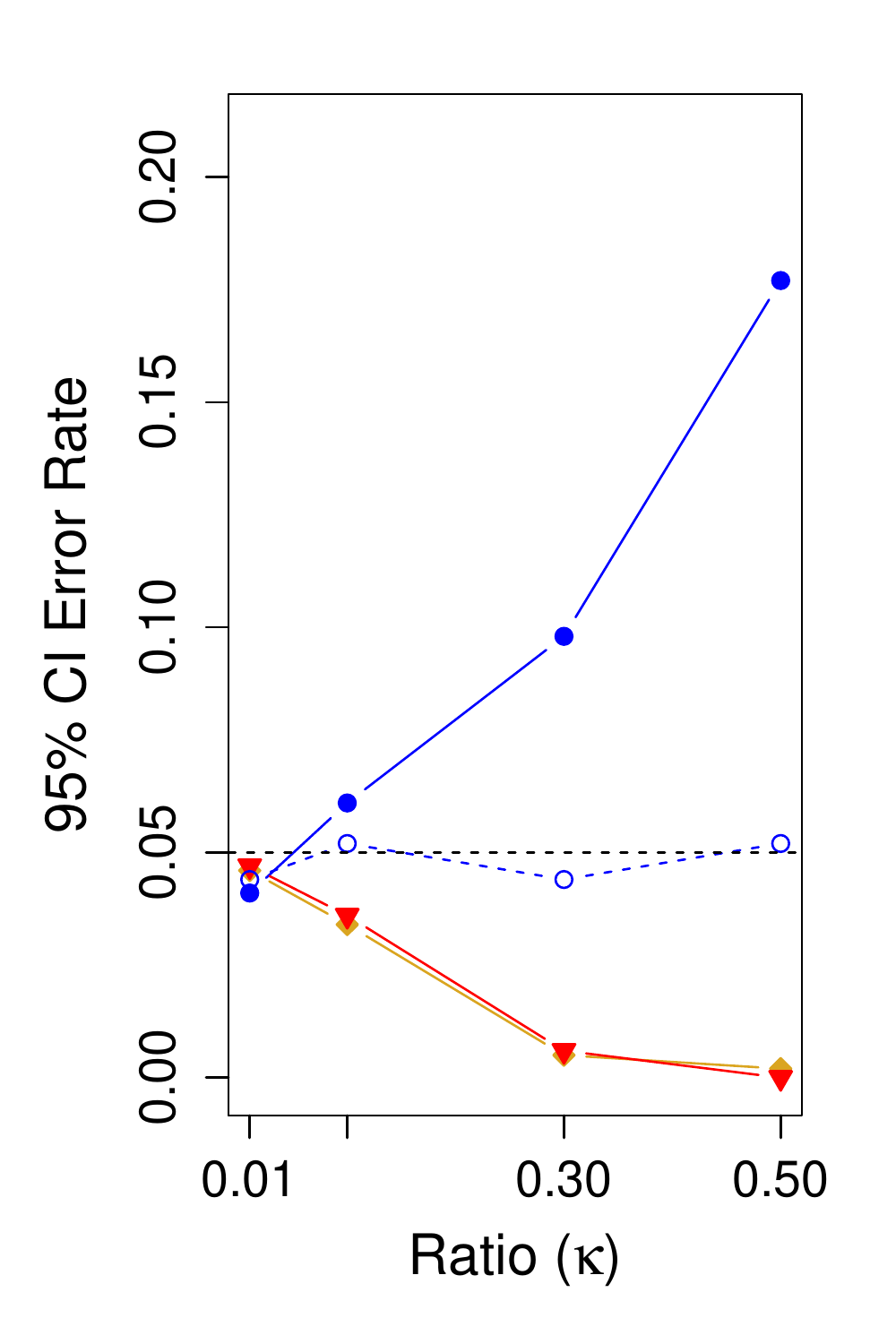} \label{subfig:basicCIErrorDesign:EN} }
	\subfloat[][Ellip. $X$, Exp]{\includegraphics[width=.3\textwidth]{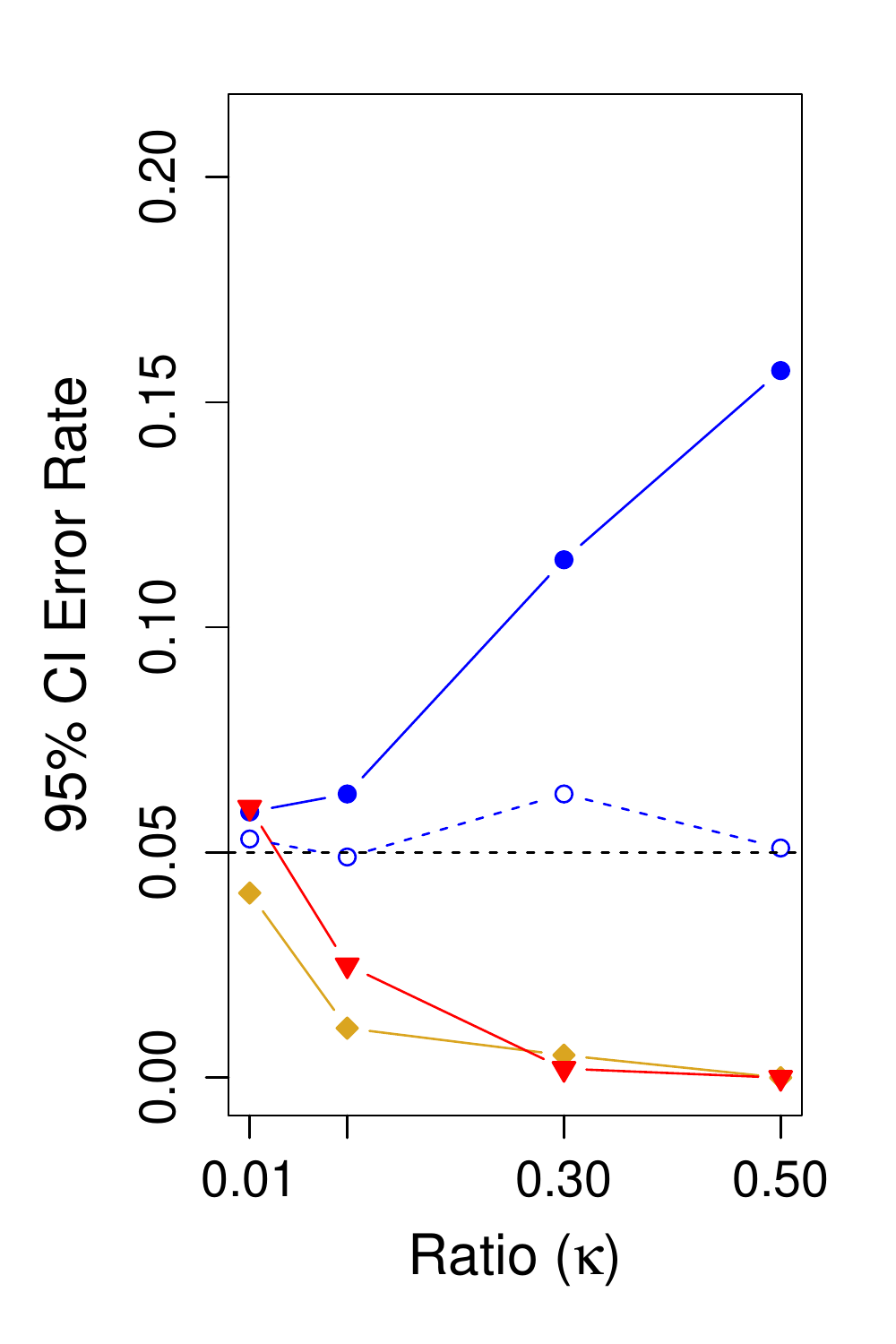} \label{subfig:basicCIErrorDesign:EE} }
\caption{
 \textbf{Performance of 95\% confidence intervals of $\beta_1$ for $L_2$  loss (elliptical design $X$): }  Here we show the coverage error rates for 95\% confidence intervals for $n=500$ with different distributions of the design matrix $X$ using ordinary least squares regression: 
\protect\subref{subfig:basicCIErrorDesign:N} $N(0,1)$, 
\protect\subref{subfig:basicCIErrorDesign:EU} elliptical with $\lambda_i\sim U(.5,1.5)$,  \protect\subref{subfig:basicCIErrorDesign:EN} elliptical with $\lambda_i\sim N(0,1)$, and \protect\subref{subfig:basicCIErrorDesign:EE} elliptical with $Exp(\sqrt{2})$. In all of these plots, the error is distributed $N(0,1)$ and the loss is $L_2$. See the caption of Figure \ref{fig:basicCIError} for additional details. 
}\label{fig:basicCIErrorDesign}
\end{figure}

\begin{figure}[t]
\centering
\centering
\subfloat[][$L_1$ loss]{\includegraphics[width=.3\textwidth]{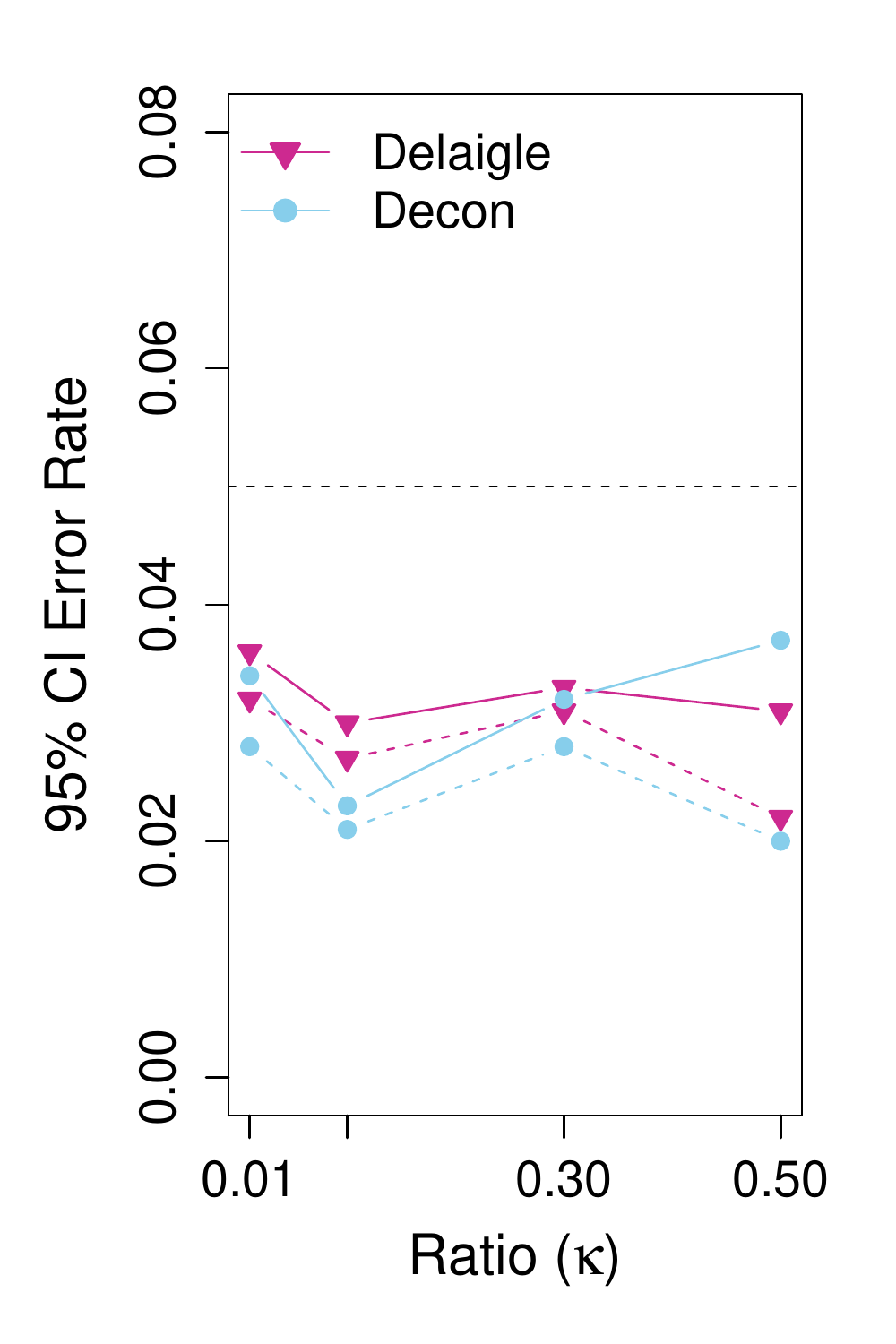}\label{subfig:bandwidth:L1}}
\subfloat[][Huber loss]{\includegraphics[width=.3\textwidth]{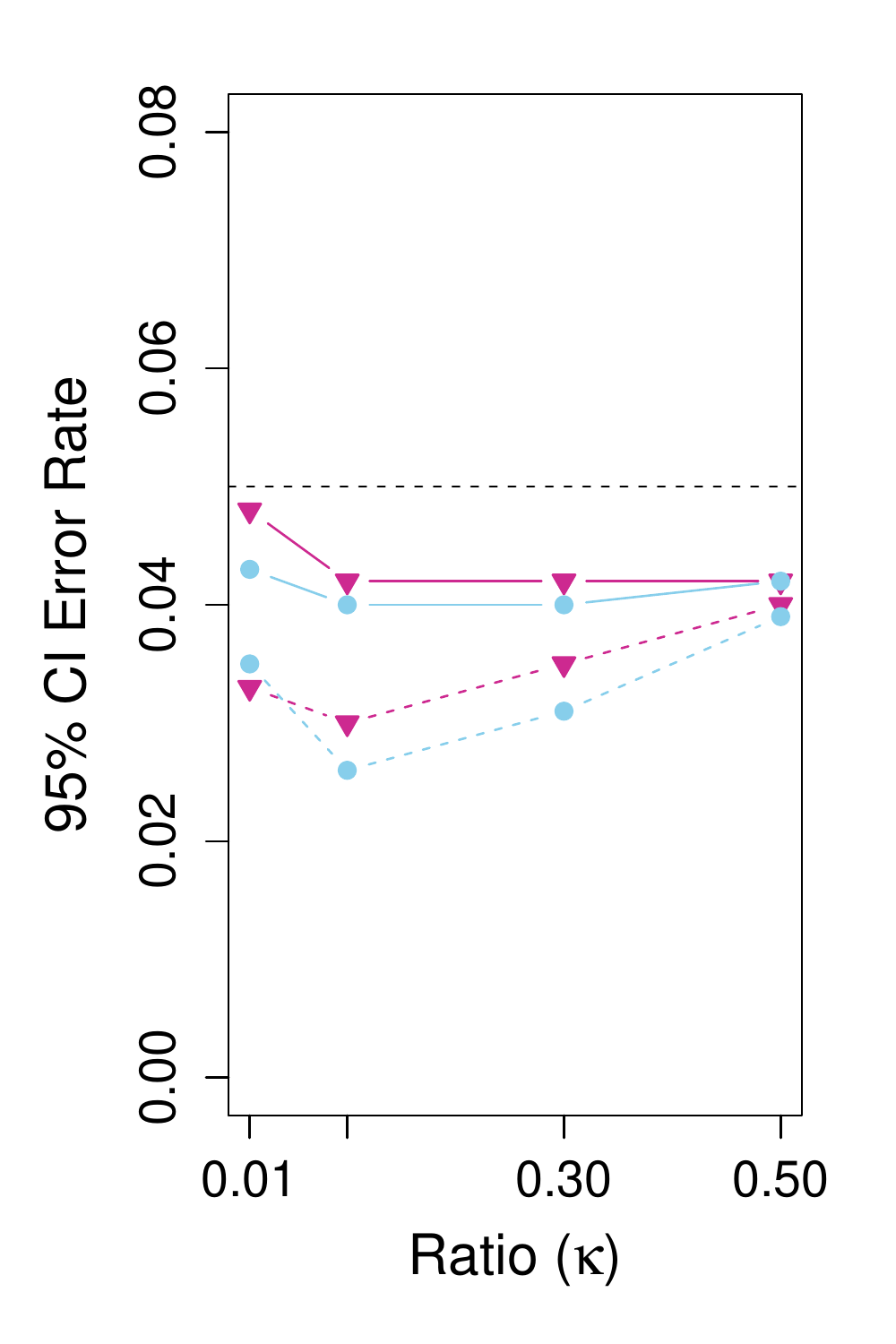} \label{subfig:bandwidth:Huber}}
\caption{\textbf{Different bandwidths for Method 1:} We plotted the error rate of 95\% confidence intervals for the deconvolution bootstrap (Method 1) using two different choices of bandwith: the bw.dboot2 in \texttt{decon} (light blue) or that of \citet{Delaigle2002,Delaigle2004} (maroon). The solid lines refer to bootstrapping by drawing $\{\eps^*_i\}_{i=1}^{n}$ as a i.i.d draws from $\hat{\epsdist}$; the dashed lines refer to $\{\eps^*_i\}_{i=1}^{n}$ drawn from repeated resampling of a single draw ($\{\hat{\eps}_i\}_{i=1}^{n}$) from $\hat{\epsdist}$. See section \ref{subsec:deconvbootmethods} below. Note that the y-axis for these plots is different than that shown in the main text. 
\label{supfig:bandwidth}
}
\end{figure}

\clearpage
\begin{center}
\textbf{\textsc{Supplementary Tables}}
\end{center}
\begin{table}
	\centering
	\subfloat[][$L_1$ loss]{
\begin{tabular}{rrrr}
  \hline
 & Residual & Jackknife & Pairs \\ 
  \hline
r=0.01 & 0.063 & 0.089 & 0.035 \\ 
  r=0.1 & 0.113 & 0.005 & 0.013 \\ 
  r=0.3 & 0.137 & 0.000 & 0.003 \\ 
  r=0.5 & 0.210 & 0.000 & 0.000 \\ 
   \hline
\end{tabular}

	 }\\
	\subfloat[][Huber loss]{
\begin{tabular}{rrrr}
  \hline
 & Residual & Jackknife & Pairs \\ 
  \hline
r=0.01 & 0.057 & 0.054 & 0.054 \\ 
  r=0.1 & 0.068 & 0.037 & 0.041 \\ 
  r=0.3 & 0.090 & 0.015 & 0.004 \\ 
  r=0.5 & 0.198 & 0.002 & 0.000 \\ 
   \hline
\end{tabular}

	 }\\
	\subfloat[][$L_2$ loss]{
\begin{tabular}{rrrr}
  \hline
 & Residual & Jackknife & Pairs \\ 
  \hline
r=0.01 & 0.040 & 0.061 & 0.040 \\ 
  r=0.1 & 0.060 & 0.034 & 0.052 \\ 
  r=0.3 & 0.098 & 0.021 & 0.033 \\ 
  r=0.5 & 0.188 & 0.005 & 0.000 \\ 
   \hline
\end{tabular}

	 }
\caption{\textbf{Error rate of 95\% confidence intervals of $\beta_1$ for $n=500$} This tables give the exact error rates plotted in Figure \ref{fig:basicCIError}. See figure caption for more details. }
\label{tab:basicCIError}
\end{table}

\begin{table}
	\centering
\begin{tabular}{rrrr}
  \hline
 & Normal & Ellip. Normal & Ellip. Exp \\ 
  \hline
r=0.01 & 1.001 & 1.001 & 1.017 \\ 
  r=0.1 & 1.016 & 1.090 & 1.156 \\ 
  r=0.3 & 1.153 & 1.502 & 1.655 \\ 
  r=0.5 & 1.737 & 3.123 & 3.635 \\ 
   \hline
\end{tabular}

\caption{\textbf{Ratio of CI Width of Pairs compared to Standard.} This tables give the ratio of the average width of the confidence intervals from pairs bootstrapping to the average for the standard interval given by theoretical results, i.e. using $var(\hat{\beta})=\sigma^2 (X'X)^{-1}$ and creating standard confidence interval. These values were used for Figure \ref{fig:CIWidth} in the text. }
\label{tab:CIWidth}
\end{table}

\begin{table}
	\centering
	\subfloat[][$L_1$ loss]{
\begin{tabular}{rrrr}
  \hline
 & Residual & Std. Pred Error & Deconv \\ 
  \hline
r=0.01 & 0.064 & 0.042 & 0.031 \\ 
  r=0.1 & 0.091 & 0.028 & 0.018 \\ 
  r=0.3 & 0.135 & 0.026 & 0.022 \\ 
  r=0.5 & 0.182 & 0.030 & 0.035 \\ 
   \hline
\end{tabular}

	 }\\
	\subfloat[][Huber loss]{
\begin{tabular}{rrrr}
  \hline
 & Residual & Std. Pred Error & Deconv \\ 
  \hline
r=0.01 & 0.065 & 0.048 & 0.036 \\ 
  r=0.1 & 0.051 & 0.054 & 0.039 \\ 
  r=0.3 & 0.098 & 0.035 & 0.037 \\ 
  r=0.5 & 0.174 & 0.034 & 0.036 \\ 
   \hline
\end{tabular}

	 }\\
\caption{\textbf{Error rate of 95\% confidence intervals using predicted errors.} This tables give the exact error rates plotted in Figure \ref{fig:bootErrorLeaveOut}. See figure caption for more details. }
\label{tab:bootErrorLeaveOut}
\end{table}

\begin{table}
	\centering
	\subfloat[][$L_1$ loss]{
\begin{tabular}{rrrr}
  \hline
 & Residual & Jackknife & Pairs \\ 
  \hline
r=0.01 & 0.064 & 0.073 & 0.032 \\ 
  r=0.1 & 0.091 & 0.002 & 0.005 \\ 
  r=0.3 & 0.135 & 0.001 & 0.001 \\ 
  r=0.5 & 0.182 &  & 0.000 \\ 
   \hline
\end{tabular}

	 }\\
	\subfloat[][Huber loss]{
\begin{tabular}{rrrr}
  \hline
 & Residual & Jackknife & Pairs \\ 
  \hline
r=0.01 & 0.065 & 0.061 & 0.059 \\ 
  r=0.1 & 0.051 & 0.042 & 0.027 \\ 
  r=0.3 & 0.098 & 0.009 & 0.009 \\ 
  r=0.5 & 0.174 & 0.001 & 0.000 \\ 
   \hline
\end{tabular}

	 }\\
	\subfloat[][$L_2$ loss]{
\begin{tabular}{rrrr}
  \hline
 & Residual & Jackknife & Pairs \\ 
  \hline
r=0.01 & 0.052 & 0.052 & 0.052 \\ 
  r=0.1 & 0.056 & 0.036 & 0.045 \\ 
  r=0.3 & 0.114 & 0.018 & 0.022 \\ 
  r=0.5 & 0.155 & 0.008 & 0.002 \\ 
   \hline
\end{tabular}

	 }
\caption{\textbf{Error rate of 95\% confidence intervals of $\beta_1$ for double exponential error} This tables give the exact error rates plotted in Figure \ref{fig:basicCIErrorLap}. See figure caption for more details. }
\label{tab:basicCIErrorLap}
\end{table}

\begin{table}
	\centering
	\subfloat[][Ellipical, Unif]{
\begin{tabular}{rrrr}
  \hline
 & Residual & Jackknife & Pairs \\ 
  \hline
r=0.01 & 0.055 & 0.053 & 0.049 \\ 
  r=0.1 & 0.070 & 0.029 & 0.047 \\ 
  r=0.3 & 0.103 & 0.013 & 0.024 \\ 
  r=0.5 & 0.157 & 0.008 & 0.000 \\ 
   \hline
\end{tabular}

	 }\\
	\subfloat[][Elliptical, Normal]{
\begin{tabular}{rrrr}
  \hline
 & Residual & Jackknife & Pairs \\ 
  \hline
r=0.01 & 0.041 & 0.046 & 0.047 \\ 
  r=0.1 & 0.061 & 0.034 & 0.036 \\ 
  r=0.3 & 0.098 & 0.005 & 0.006 \\ 
  r=0.5 & 0.177 & 0.002 & 0.000 \\ 
   \hline
\end{tabular}

	 }\\
	\subfloat[][Elliptical, Exp]{
\begin{tabular}{rrrr}
  \hline
 & Residual & Jackknife & Pairs \\ 
  \hline
r=0.01 & 0.059 & 0.041 & 0.060 \\ 
  r=0.1 & 0.063 & 0.011 & 0.025 \\ 
  r=0.3 & 0.115 & 0.005 & 0.002 \\ 
  r=0.5 & 0.157 & 0.000 & 0.000 \\ 
   \hline
\end{tabular}

	 }
\caption{\textbf{Error rate of 95\% confidence intervals of $\beta_1$ for elliptical design $X$} This tables give the exact error rates plotted in Figure \ref{fig:basicCIErrorDesign}. See figure caption for more details. }
\label{tab:basicCIErrorDesign}
\end{table}

\begin{table} 
\begin{center}
\begin{tabular}{|*{11}{c}|}\hline 
$\kappa$ & 0.05 &   0.10  &  0.15 &   0.20 &   0.25 &   0.30 &   0.35 &   0.40 &   0.45 &   0.50 \\ \hline
$\alpha(\kappa)$ & 0.9938  &  0.9875  &  0.9812  &  0.9688  &  0.9562  &  0.9426  &  0.9352  &  0.9277  &  0.9222  &  0.9203 \\ \hline
\end{tabular}
\end{center}
\caption{Values of $\alpha(\kappa)$ to use to fix the variance estimation issue in high-dimensional pairs-bootstrap}
\label{tab:GoodWeightProportions}
\end{table}

\begin{table}
	\centering
	\subfloat[][Jackknife]{
\begin{tabular}{rrrr}
  \hline
 & L2 & Huber & L1 \\ 
  \hline
r=0.01 & 0.964 & 0.991 & 2.060 \\ 
  r=0.1 & 1.115 & 1.173 & 5.432 \\ 
  r=0.3 & 1.411 & 1.613 & 10.862 \\ 
  r=0.5 & 1.986 & 2.671 & 14.045 \\ 
   \hline
\end{tabular}

	 }\\
	\subfloat[][Pairs Bootstrap]{
\begin{tabular}{rrrr}
  \hline
 & L2 & Huber & L1 \\ 
  \hline
r=0.01 & 1.078 & 0.923 & 1.081 \\ 
  r=0.1 & 1.041 & 1.098 & 1.351 \\ 
  r=0.3 & 1.333 & 1.954 & 2.001 \\ 
  r=0.5 & 2.808 & 4.507 & 3.156 \\ 
   \hline
\end{tabular}

	 }\\
\caption{\textbf{Over estimation of variance for Pairs bootstrap and Jackknife} This tables give the median values of the boxplots plotted in Figures \ref{fig:bootOverEstFactor} and \ref{fig:bootOverEstFactorJack}. See relevant figure captions for more details. }
\label{tab:bootOverEstFactor}
\end{table}

\clearpage
\begin{center}
\textbf{\textsc{Bibliography}}
\end{center}

\bibliographystyle{plain}
\bibliography{research,Bioresearch,addBoot}

\end{document}